\documentclass[12pt]{article}

\usepackage{amsmath,amssymb,latexsym,amsfonts}
\newcommand{\qed}{\hfill \ensuremath{\Box}}

\usepackage[left=2.5cm,right=2.5cm,top=2.5cm,bottom=2.5cm]{geometry}

\usepackage{tikz}
\usepackage{graphicx}
\usepackage{float}
\usepackage{subfigure}

\usepackage[ruled,vlined,linesnumbered]{algorithm2e}
\newfloat{Algorithm}{h}{mel}
\usepackage{alltt}

\usepackage{tabularx}

\usepackage{dsfont} 

\usepackage{hyperref}
\newcommand{\mathj}{\texorpdfstring{$j$}{j}\xspace}
\newcommand{\mathk}{\texorpdfstring{$k$}{k}\xspace}
\newcommand{\elltt}{\texorpdfstring{$\ell_2^2$}{l-2-2}\xspace}

\newcommand{\Linfty}{\texorpdfstring{$\L_{\infty}$}{L-infinity}\xspace}
\newcommand{\mathC}{\texorpdfstring{$\mathcal C$}{C}\xspace}

\makeatletter
\newtheorem{rep@theorem}{\rep@title}
\newcommand{\newreptheorem}[2]{%
\newenvironment{rep#1}[1]{%
 \def\rep@title{#2 \ref{##1}}%
 \begin{rep@theorem}}%
 {\end{rep@theorem}}}
\makeatother

\newcommand{\REAL}{\ensuremath{\mathbb{R}}}
\providecommand{\norm}[1]{\left\lVert#1\right\rVert}
\newcommand{\eps}{\varepsilon}
\renewcommand{\epsilon}{\varepsilon}
\newtheorem{theorem}{Theorem}
\newtheorem{corollary}[theorem]{Corollary}
\newtheorem{lemma}[theorem]{Lemma}

\newtheorem{definition}[theorem]{Definition}

\newtheorem{claim}[theorem]{Claim}

\newtheorem{observation}[theorem]{Observation}

\newenvironment{proof}{\noindent {\bf Proof:}\ }{\qed \par\vskip 4mm\par}

\newcommand{\br}[1]{\left\{#1\right\}}                            

\renewcommand{\S}{2}
\renewcommand{\j}{j}
\newcommand{\dist}{\mathrm{dist}}

\DeclareMathOperator{\opt}{opt}
\DeclareMathOperator{\cost}{cost}

\newcommand{\ie}{i.\,e.,}
\newcommand{\eg}{e.\,g.,}

\DeclareFontFamily{OT1}{pzc}{}
\DeclareFontShape{OT1}{pzc}{m}{it}{<-> s * [1.2] pzcmi7t}{}
\DeclareMathAlphabet{\mathpzc}{OT1}{pzc}{m}{it}

\DeclareFontFamily{OT1}{pzcl}{}
\DeclareFontShape{OT1}{pzcl}{m}{it}{<-> s * [1] pzcmi7t}{}
\DeclareMathAlphabet{\mathpzcl}{OT1}{pzcl}{m}{it}

\newcommand{\ucc}[1]{\mathpzc{#1}}

\newreptheorem{theorem}{Theorem}
\newreptheorem{lemma}{Lemma}
\newreptheorem{corollary}{Corollary}

\newcommand{\RR}{\mathbb{R}}

\renewcommand{\j}{j}
\renewcommand{\L}{\ensuremath{\mathcal{L}}}
\allowdisplaybreaks

\newcommand{\algstr}{\textsc{Streaming-Subspace-Approximation}}
\newcommand{\algout}{\textsc{Output-Coreset}}

\newcommand{\coresetsize}{\text{CoresetSize}}

\newcommand{\ranges}{\mathrm{ranges}}
\newcommand{\range}{\mathrm{range}}
\newcommand{\image}{\RR^{\ge 0}}

\newcommand{\kk}{\texorpdfstring{$k$}{k}}
\newcommand{\jj}{\texorpdfstring{$j$}{j}}

\begin{document}
\thispagestyle{empty}
\title{Turning Big data into tiny data: \\Constant-size coresets for\\ $k$-means, PCA and projective clustering}

\author{Dan Feldman\thanks{University of Haifa, Robotics \& Big Data Lab. Email: dannyf.post@gmail.com}
			\and Melanie Schmidt\thanks{Rheinische Friedrich-Wilhelms-Universit\"at Bonn, Germany, Email: melanieschmidt@uni-bonn.de }
			\and Christian Sohler\thanks{TU Dortmund, Germany, Email: christian.sohler@tu-dortmund.de. The author acknowledges the support
			of the collaborative research center 876, project A2, funded by the German Science Foundation.}}

\date{}

\maketitle

\begin{abstract}
We develop and analyze a method to reduce the size of a very large set of data points in a high dimensional Euclidean space $\REAL^d$ to a small set of weighted points such that the result of a predetermined data analysis task on the reduced set is approximately the same as that for the original point set. For example, computing the first $k$ principal components of the reduced set will return approximately the first $k$ principal components of the original set or computing the centers of a $k$-means clustering on the reduced set will return an approximation for the original set. Such a reduced set is also known as a \emph{coreset}.
The main new feature of our construction is that the cardinality of the reduced set is independent of the dimension $d$ of the input space \emph{and} that the sets are mergable. The latter property means that the union of two reduced sets is a reduced set for the union of the two original sets (this property has recently also been called \emph{composable}, see~\cite{indyk2014composable}). It allows us to turn our methods into streaming  or distributed algorithms using standard approaches.
For problems such as $k$-means and subspace approximation the coreset sizes are also independent of the number of input points.

Our method is based on projecting the points on a low dimensional subspace and reducing the cardinality of the points
inside this subspace using known methods. The proposed approach works for a wide range of data analysis techniques including $k$-means clustering, principal component analysis and subspace clustering.

The main conceptual contribution is a new coreset definition that allows to charge costs that appear for every solution to an additive constant.
\end{abstract}
\setcounter{page}{1}

\renewcommand{\paragraph}[1]{\medskip\noindent\textbf{{#1}.}}
\section{Introduction\label{sec:int}}

In many areas of science, progress is closely related to the capability to analyze massive amounts of data.
Examples include particle physics where according to the webpage dedicated to the Large Hadron collider beauty experiment \cite{lhcb} at CERN, after a first filtering phase 35 GByte of data per second need to be processed
``to explore what happened after the Big Bang that allowed matter to survive and build the Universe we inhabit today'' \cite{lhcb}.
The IceCube neutrino observatory ``searches for neutrinos from the most violent astrophysical sources: events like exploding stars, gamma ray bursts, and cataclysmic phenomena involving black holes and neutron stars.'' \cite{IceCube}. According to the webpages \cite{IceCube}, the datasets obtained are of a projected size of about 10 Teta-Bytes per year.
Also, in many other areas the data sets are growing in size because they are increasingly being gathered by ubiquitous information-sensing mobile devices, aerial sensory technologies (remote sensing), genome sequencing, cameras, microphones, radio-frequency identification chips, finance (such as stocks) logs, internet search,  and wireless sensor networks~\cite{Hellerstein, Segaran}. 

The world's technological per-capita capacity to store information has roughly doubled every 40 months since the 1980s~\cite{Hilbert}; as of 2012, every day 2.5  quintillion bytes($2.5\times 10^{18}$) of data were created~\cite{IBM}. Data sets as the ones described above and the challenges involved when analyzing them is often subsumed in the term \emph{Big Data}.
Big Data is also sometimes described by the ``3Vs" model~\cite{Beyer}:
increasing \emph{volume} $n$ (number of observations or records), its \emph{velocity} (update time per new observation) and its \emph{variety} $d$ (dimension of data, features, or range of sources).

In order to analyze data that for example results from the experiments above, one needs to employ automated data analysis methods that can identify important patterns and substructures in the data, find the most influential features, or reduce size and dimensionality of the data.
Classical methods to analyze and/or summarize data sets include clustering, i.e., the partitioning of data into subsets of similar characteristics, and principal component analysis which allows to consider the dimensions of a data set that have the highest variance. Examples for such methods include $k$-means clustering, principal component analysis (PCA), and subspace clustering.

The main problem with many existing approaches is that they are often efficient for large number $n$ of input records, but are not efficient enough to deal with Big Data where also the dimension $d$ is asymptotically large.
One needs special algorithms that can easily handle massive streams of possibly high dimensional measurements and that can be easily parallelized and/or applied in a distributed setting.

In this paper, we address the problem of analyzing Big Data by developing and analyzing a new method to reduce the data size while approximately keeping its main characteristics in such a way that any approximation algorithm run on the
reduced set will return an approximate solution for the original set. This reduced data representation (semantic compression) is sometimes called a \emph{coreset}.
Our method reduces any number of items to a number of items that only depends on some problem parameters (like the number of clusters) and the quality of the approximation, but not on the number of input items or the dimension of the input space.

Furthermore, we can always take the union of two data sets that were reduced in this way and the union provides an approximation for the two original data sets.
The latter property is very useful in a distributed or streaming setting and allows for very simple algorithms using standard techniques. For example, to process a very large data set on a cloud, a distributed system or parallel computer,
we can simply assign a part of the data set to each processor, compute the reduced representation, collect it somewhere
and do the analysis on the union of the reduced sets. This merge-and-reduce method is strongly related to MapReduce and its popular implementations (e.g. Hadoop~\cite{Hadoop}). If there is a stream of data or if the data is stored on a secondary storage device, we can read chunks that fit into the main memory of each individual computer and then reduce the data in this chunk. In the end, we apply our data analysis tools on the union of the reduced sets via small communication of only the coresets between the computers.

Our main result is a dimensionality reduction algorithm for $n$ points in high $d$-dimensional space to $n$ points in $O(j/\eps^2)$ dimensional space, such that the sum of squared distances to every object that is contained in a $j$-dimensional subspace is approximated up to a $(1+\eps)$-factor. This result is applicable to a wide range of problems like PCA, $k$-means clustering and projective clustering. For the case of PCA, i.e., subspace approximation, we even get a coreset of cardinality $O(j/\eps)$ (here $j$ is just the dimension of the subspace). The cardinality of this coreset is constant in the sense that it is independent of the input size: both its original cardinality $n$ and dimension $d$. A coreset of such a constant cardinality is also obtained for $k$-means queries, i.e., approximating the sum of squared distances over each input point to its closest center in the query.

For other objectives, we combine our reduction with existing coreset constructions to obtain very small coresets. A construction that computes coresets of cardinality $f(n,d,k)$ will result in a construction that computes coresets of cardinality $f(n,O(k/\eps^2),k)$, i.e., independent of $d$. This scheme works as long as there is such a coreset construction, e.g., it works for $k$-means or $k$-line-means. For the projective clustering problem (more precisely, the affine $j$-subspace $k$-clustering problem that we define below), such a coreset construction does not and can not exist. We circumvent this problem by requiring that the points are on an integer grid (the resulting size of the coreset will depend polylogarithmically on the size of the grid and n).  

A more detailed (technical) description of our results is given in Section~\ref{ourresults} after a detailed discussion about the studied problems and concepts.

\subsection*{Erratum}

We remark that in the conference version of this paper, some of the coreset sizes resulting from applying our new technique were incorrect. We have updated the results in this paper (see Section \ref{ourresults} for an overview). In particular, the
coreset size for projective clustering is not independent of $n$.

\subsection*{Previous publications}
The main results of this work have been published in~\cite{FSS13}. However, the version at hand is significantly different. We carefully derive the concrete application to several optimization problems, develop explicit streaming algorithms, explain and re-prove some related results that we need, correct errors from the conference version (see above), provide pseudo code for most methods and add a lot of explanations compared to the conference version. The PhD thesis~\cite{S14} also contains a write-up of the main results, but without the streaming algorithms for subspace approximation and projective clustering.

\section{Preliminaries}\label{background}
In this section we formally define our notation and the problems that we study.

\paragraph{Matrix notation}
The set of all real-valued $n\times d$ matrices is denoted by $\REAL^{n\times d}$. Our input data set is a set of $n$ points in $\REAL^d$. We will represent it by a matrix $A\in\REAL^{n\times d}$, whose rows are the input points. The entry in the $i$th row and $j$th column of $A$ is denoted by $A_{ij}$. We use $A_{i*}$ to denote the $i$-th row of a $A$ and $A_{*j}$ to denote its $j$-th column.
We use $I_d$ to denote the $d \times d$ identity matrix, or just $I$ if the dimension is clear from the context. We say that a matrix $X\in\REAL^{d\times j}$
has orthonormal columns if its columns are orthogonal unit vectors.  Notice that every such matrix $X$ satisfies $X^T X = I$.
 If $A$ is also a square matrix, it is called an \emph{orthogonal} matrix.

Any matrix $A\in\REAL^{n\times d}$ has a singular value decomposition (SVD), which is a factorization $A= U \Sigma V^T$, where $U$ is an $n\times n$ orthogonal matrix, $V$ is a $d \times d$ orthogonal matrix and $\Sigma$ is an $n \times d$ rectangular diagonal matrix
whose diagonal entries are non-negative and non-increasing.
We use $\sigma_1,\cdots, \sigma_{\min\br{n,d}}$ to denote the diagonal elements $\Sigma_{1,1}, \cdots,\Sigma_{\min\br{n,d},\min\br{n,d}}$ of $\Sigma$.
The $n$ columns of $U$ are called the \emph{left singular vectors} of $A$. Similarly, the $d$ columns of $V$ are called the \emph{right singular vectors} of $A$. Note that the right singular vectors are the eigenvectors of $A^T A$ in order of non-increasing corresponding eigenvalues.

The number of non-zeroes entries in $\Sigma$ is the rank $r$ of $A$, which is bounded by $\min\br{n,d}$, that is $r = |\br{i \mid \sigma_i >0, i=1,\ldots,\min\br{n,d}}|\leq \min\br{n,d}$. This motivates the \emph{thin SVD} $A=U_r \Sigma_r (V_r)^T$, where  $U_r\in\REAL^{n\times r}$, $\Sigma_r\in\REAL^{r\times r}$ and $V_r\in\REAL^{d\times r}$ denote the first $r$ columns of $U$, first $r$ columns/rows of $\Sigma$ and first $r$ columns of $V$, respectively.
Notice that the matrix product is still equal to $A$ (it is still a factorization). If we keep less than $r$ entries of $\Sigma$, then we get an approximation of $A$ with respect to the squared Frobenius norm. We will use this approximation frequently in this paper.
\begin{definition}\label{def:rank}
Let $A\in\REAL^{n\times d}$ and $U\Sigma V^T=A$ be its SVD.
Let $m\in[1,\min\br{d,n}]$ be an integer and define $\Sigma	^{(m)}$ to be the $d\times d$ diagonal matrix whose first $m$ diagonal entries are the same as that of $\Sigma$ and whose remaining entries are $0$.
Then the \emph{$m$-rank approximation} $A^{(m)}\in\REAL^{n\times d}$ of $A$ is defined as $A^{(m)}=U\Sigma^{(m)}V^T$.
\end{definition}

\paragraph{Subspaces}
The columns of a matrix $X$ span a linear subspace $L$ if the set of all linear combinations of columns of $X$ equals $L$.
In this case we also say that $X$ spans $L$:
A $j$-dimensional linear subspace $L \subseteq \REAL^d$ can be represented by a matrix $X\in\REAL^{d\times j}$ with orthonormal columns that span $L$.
The projection of a point set (matrix) $A\in\REAL^{n\times d}$ on a linear subspace $L$ represented by $X$ is the point set (matrix) $AX\in\REAL^{n\times \j}$.
The projection of the point set (matrix) $A$ on $L$ using the coordinates of $\REAL^d$ is the set of rows of $AXX^T$.

Given $p\in\REAL^d$ and a set $S$ we denote by $p+S=\br{p+s\mid s\in S}$ the translation of $S$ by $p$. Similarly, $A+p$ is the translation of each row of $A$ by $p$. An \emph{affine subspace} is a translation of a linear subspace and as such can be written as $p+L$, where $p \in \REAL^d$ is the translation vector and $L$ is a linear subspace.

\paragraph{Distances and norms}
The origin of $\REAL^d$ is denote by $\vec{0}$, where the dimension follows from the context.
For a vector $x \in \REAL^d$ we write $\|x\|_2$ to denote its $\ell_2$-norm, the square root of the sum of its squared entries. More generally, for an $n \times d$ matrix $A$ we write
$\|A\|_\S = \max_{x\in \REAL^d, \|x\| \not=0} \frac{\|Ax\|_2}{\|x\|_2} = \sigma_1$
to denote its spectral norm and $$\|A\|_F = \sqrt{\sum_{1\le i \le n } \sum_{1\le j \le d} A_{ij}^2}=\sqrt{\sum_{1\le j \le \min\br{d,n}} \sigma^2_j}$$
to denote its Frobenius norm. It is known that the Frobenius norm does not change under orthogonal transformations, i.e.,
$\|A\|_F = \|AQ\|_F$ for an $n\times d$ matrix $A$ and an orthogonal matrix $Q$. This also implies the following observation
that we will use frequently in the paper.
\begin{observation}\label{observation:simple}
Let $A$ be an $n \times d$ matrix and $B$ be a $j \times d$ matrix with orthonormal columns. Then
$$\|A\|_F^2 \ge \|AB^T\||_F^2.$$
\end{observation}
\begin{proof}
Let $B'$ be a $d\times d$ orthogonal matrix whose first $j$ columns agree with $B$. Then we have
$\|A\|_F^2 = \|A(B')^T\|_F^2 \ge \|AB^T\|_F^2$.
\end{proof}

\begin{claim}\label{Claim:Pythagorean}[Matrix form of the Pythagorean Theorem]
Let $X$ be a $d \times j$ matrix with orthonormal columns and $Y$ be a $d \times (d-j)$ matrix with orthonormal columns that spans the orthogonal complement of $X$.
Furthermore, let $A$ be any $n \times d$ matrix. Then we have
$$
\|A\|_F^2 = \|AX\|_F^2 + \|AY\|_F^2.
$$
\end{claim}

\begin{proof}
Let $B$ be the $d \times d$ matrix whose first $j$ columns equal $X$ and the second $d-j$ columns equal $Y$. Observe that $B$ is an orthogonal matrix.
Since the Frobenius norm does not change under multiplication with orthogonal matrices, we get
$$
\|A\|_F^2 = \|AB\|_F^2
$$
The result follows by observing that $\|AB\|_F^2 = \|AX\|_F^2 + \|AY\|_F^2$.
\end{proof}

\newcommand{\Sset}{S}
\newcommand{\Sitem}{s}
For a set $C \subseteq\REAL^d$ and a vector $p$ in $\REAL^d$, we denote the Euclidean distance between $p$ and (its closest point in) $C$ by $\dist(p,C):=\inf_{c \in C}\norm{p-c}_2$ if $C$ is non-empty (notice that the infimum always exists because the distance is lower bounded by zero, but the minimum might not exist, \eg\ when $C$ is open), or $\dist(p,C):=\infty$ otherwise. In this paper we will mostly deal with the squared Euclidean distance, which we denote by $\dist^2(.,.)$.
For an $n\times d$ matrix $A$, we will slightly abuse notation and write $\dist(A_{i*},C)$ to denote the distance $(A_{i*})^T$ to $C$. The sum of the squared distances of the rows of $A$ to $C$ by
$\label{eq8}
\dist^2(A,C) = \sum_{i=1}^n \dist^2(A_{i*},C).
$
If the rows of the matrix are weighted by non-negative weights $w_1, \dots,w_n$ then we sometimes use
$
\dist^2_w(A,C) = \sum_{i=1}^n w_i \cdot \dist^2(A_{i*},C).
$
Let $L \subseteq\REAL^d$ be a $j$-dimensional linear subspace represented by a matrix $X\in\REAL^{d\times j}$ with orthonormal columns that spans $L$.
Then the orthogonal complement $L^{\bot}$ of $L$ can also be represented by a matrix with orthonormal columns which spans $L^{\bot}$. We usually name it $Y \in\REAL^{d \times (d-j)}$.
The distance from a point (column vector) $p\in\REAL^d$ to $L$ is the norm of its projection on $L^{\bot}$, $\dist(p,L)=\norm{p^TY}_F$.
The sum of squared distances from the rows of a matrix $A\in\REAL^{n\times d}$ to $L$ is thus $\dist^2(A,L)=\norm{AY}_F^2$.

\paragraph{Range spaces and VC-dimension}
In the following we will introduce the definitions related to range spaces and VC-dimension that are used in this paper.

\begin{definition}
 A range space is a pair $(F,\ranges)$ where $F$ is a set, called \emph{ground set} and $\ranges$ is a family (set) of subsets of $F$, called \emph{ranges}.
\end{definition}

\begin{definition}[VC-dimension]
The VC-dimension of a range space $(F,\ranges)$ is the size $|G|$ of the largest subset $G\subseteq F$ such that
\[
\Big|\{G\cap\range \mid \range\in\ranges\}\Big|= 2^{|G|}.
\]
\end{definition}

In the context of range spaces we will use the following type of approximation.

\begin{definition}[\cite{HPS11}] 
Let $\eta,\eps>0$, and $(F,\ranges)$ be a range space with finite $F \not= \emptyset$. An $(\eta,\eps)$-approximation of
 $(F,\ranges)$ is a set $S\subseteq F$ such that for all $\range \in \ranges$ we have
$$
\bigg| \frac{|\range \cap S|}{|S|} -\frac{|\range \cap F|}{|F|}\bigg| \le \eps \cdot \frac{|\range \cap F|}{|F|} , \text{ if $|\range \cap F| \ge \eta |F|$,}
$$
and
$$
\bigg| \frac{|\range \cap S|}{|S|} - \frac{|\range \cap F|}{|F|}\bigg| \le \eps \eta, \text{ if $|\range \cap F| \le \eta |F|$.}
$$
\end{definition}

We will also use the following bound from \cite{LLS01} (see also \cite{HPS11}) on the sample size required to obtain a $(\eta,\epsilon)$-approximation.

\begin{theorem}\cite{LLS01}\label{approximation}
Let $(F,\ranges)$ with finite $F \neq \emptyset$ be a range space with VC-dimension $d$, $\eta>0$ and $\eps,\delta\in(0,1)$. There is a universal
constant $c>0$ such that a sample of
$$
\frac{c}{\eta \eps^2} \cdot \big(d \log \frac{1}{\eta}+ \log \frac{1}{\delta} \big)
$$
elements drawn independently and uniformly at random from $F$ is a $(\eta,\epsilon)$-approximation for $(F,\ranges)$ with probability at least $1-\delta$,
where $d$ denotes the $VC$-dimension of $(F,\ranges)$.
\end{theorem}


\subsection{Data Analysis Methods}


In this section we briefly describe the data analysis methods for which we will develop coresets in this paper.
The first two subsection define and explain two fundamental data analysis methods: $k$-means clustering and principal
component analysis. Then we discuss the other techniques considered in this paper, which can be viewed as
generalizations of these problems.
We always start to describe the motivation of a method, then give the technical problem definition and in the end discuss the state of the art.

\paragraph{$k$-Means Clustering}\label{paragraph:kmeans}
The goal of clustering is to partition a given set of data items into subsets such that items in the same subset are similar and items in different subsets are dissimilar. Each of the computed subsets can be viewed as a class of items and, if done properly, the classes have some semantic interpretation. Thus, clustering is an unsupervised learning problem. In the context of Big Data, another important aspect is that many clustering formulations are based on the concept of a cluster center, which can be viewed as some form of representative of the cluster. When we replace each cluster by its representative, we obtain a concise description of the original data. This description is much smaller than the original data and can be analyzed much easier (and possibly by hand). There are many different clustering formulations, each with its own advantages and drawbacks and we focus on some of the most widely used ones.
Given the centers, we can typically compute the corresponding partition by assigning each data item to its closest center. Since in the context of Big Data storing such a partition may already be difficult, we will focus on computing the centers in the problem definitions below.

Maybe the most widely used clustering method is $k$-means clustering. Here the goal is to minimize the sum of squared error to $k$
cluster centers.

\begin{definition}[The $k$-means clustering problem (sum of squared error clustering)]
Given $A\in\REAL^{n\times d}$, compute a set $C$ of $k$ centers (points) in $\REAL^d$ such that its sum of squared distance to the rows of $A$,
$
\dist^2(A,C),
$
is minimized.
\end{definition}

The $k$-means problem is studied since the fifties. It is NP-hard, even for two centers~\cite{ADHP09} or in the plane~\cite{MNV09}. When either the number of clusters $k$ is a constant (see, for example, \cite{FMS07,KSS10,FL11}) or the dimension $d$ is constant \cite{FRS16,CKM16}, it is possible to compute a $(1+\epsilon)$-approximation for fixed $\epsilon>0$ in polynomial time. In the general case, the $k$-means problem is APX-hard and cannot be approximated better than 1.0013~\cite{ACKS15,LSW17} in polynomial time. On the positive side, the best known approximation guarantee has recently been improved to 6.357~\cite{ANSW16}.


\subsubsection{Principal Component Analysis}\label{sec:pca}


Let $A$ be an $n \times d$ matrix whose rows are considered as data points. Let us assume that $A$
has mean $\vec 0$, i.e., the rows sum up to the origin $\vec 0$ of $\REAL^d$. Given such a matrix $A$, the goal of principal component analysis is to transform its representation
in such a way that the individual dimensions are linearly uncorrelated and are ordered by their importance.
In order to do so, one considers the co-variance matrix $A^T A$ and computes its eigenvectors.
They are sorted by their corresponding eigenvalues and normalized to form an orthonormal basis of the input space.
The eigenvectors can be computed using the singular value decomposition. They are simply the right singular vectors of $A$ (sorted by their corresponding singular values).

The eigenvectors corresponding to the largest eigenvalues point into the direction(s) of highest variance. These are the most important directions. Ordering all eigenvectors according to their eigenvalues means that one gets a basis for $A$ which is ordered by importance. Consequently, one typical application of PCA is to identify the most important dimensions.
This is particularly interesting in the context of high dimensional data, since maintaining a complete basis of the input space requires $\Theta(d^2)$ space.
Using PCA, we can keep the $j$ most important directions. We are interested in computing an approximation of these directions.

An almost equivalent geometric formulation of this problem, which is used in this paper, is to find a linear
subspace of dimension $j$ such that the variance of the projection of the points on this
subspace is maximized; this is the space spanned by the first $j$
right singular vectors (note that the difference between the two problem definitions is that knowing the subspace
does not imply that we know the singular vectors. The subspace may be given by any basis.)
By the Pythagorean Theorem, the problem of finding this subspace is equivalent to the problem of minimizing the sum of squared distances to a subspace,
i.e., find a subspace $L$ such that $\dist^2(A,L)$ is minimized over all $j$-dimensional subspaces of $\REAL^d$. We remark that in this problem formulation we are not assuming that the data is normalized, i.e., that the mean of the rows of $A$ is $0$.

\begin{definition}[linear $j$-subspace problem]
Let $A\in\REAL^{n \times d}$ and $j\in[1,d-1]$ be an integer. The $j$-subspace problem
is to compute a $j$-dimensional subspace $L$ of $\REAL^d$ that minimizes $\dist^2(A,L)$.
\end{definition}

We may also formulate the above problem as finding a matrix $Y\in\REAL^{d\times (d-j)}$ with
orthonormal columns that minimizes $\norm{AY}_F^2$ over every such possible matrix $Y$.
Such a matrix $Y$ spans the orthogonal complement of $L$. For the subspace $L^*$ that minimizes the squared Frobenius norm
we have
\[
\norm{AY}_F^2=\dist^2(A,L^*)=\norm{A-A^{(j)}}_F^2.
\]

If we would like to do a PCA on unnormalized data, the problem is better captured by the affine $j$-subspace problem.

A coreset for $j$-subspace queries, i.e., that approximates the sum of squared distances to a given $j$-dimensional subspace was suggested by Ghashami, Liberty, Phillips, and Woodruff~\cite{liberty2013simple,ghashami2016frequent}, following the conference version of our paper. This coreset is composable and has cardinality of $O(j/\eps)$. It also has the advantage of supporting streaming input without the merge-and-reduce tree as defined in Section~\ref{sec:int} and the additional $\log n$ factors it introduces. However, it is not clear how to generalize the result for affine $j$-subspaces~\cite{jeff} as defined below.

\begin{definition}[affine $j$-subspace problem]
Let $A\in\REAL^{n \times d}$ and $j\in[1,d-1]$ be an integer. The affine $j$-subspace problem
is to compute a $j$-dimensional affine subspace $p+L$ of $\REAL^d$ that minimizes $\dist^2(A,p+L)$.
\end{definition}

The singular value decomposition was developed by different mathematicians in the 19th century (see~\cite{S93} for a historic overview). Numerically stable algorithms to compute it were developed in the sixties~\cite{GK65,GR70}. Nowadays, new challenges include very fast computations of the SVD, in particular in the streaming model (see Section~\ref{prelim:streaming}). Note that the projection of $A$ to the optimal subspace $L$ (in which we are interested in the case of PCA), is called \emph{low-rank approximation} in the literature since the input matrix $A \in \REAL^{n\times d}$ is replaced by a matrix (representing $L$) that has lower rank, namely rank $j$.
Newer $(1+\varepsilon)$-approximation algorithms for low-rank approximation and subspace approximation are based on randomization and significantly reduce the running time compared to computing the SVD~\cite{CW09,CW13,DV06,DR10,DTV11,feldman2010coresets,NDT09,Sarlos06,SV12}.
More information on the huge body of work on this topic can be found in the surveys by Halko, Martinsson and Tropp~\cite{HMT11} and Mahoney~\cite{Mahoney}.

It is known~\cite{DFKVV04} that computing the $k$-means on the low $k$-rank of the input data (its first $k$ largest singular vectors), yields a $2$-approximation for the $k$-means of the input. Our result generalizes this claim by replacing $2$ with $(1+\eps)$ and $k$ with $O(k/\eps)$, as well as approximating the distances to any $k$ centers that are contained in a $k$-subspace. 

The coresets in this paper are not subset of the input. Following papers aimed to add this property, e.g. since it preserves the sparsity of the input, easy to interpret, and more numerically stable. However, their size is larger and the algorithms are more involved. The first coreset for the $j$-subspace problem (as defined in this paper) of size that is independent of both $n$ and $d$, but are also subsets of the input points, was suggested in~\cite{FeldmanVR15,feldman2016dimensionality}. The coreset size is larger but still polynomial in $(k/\eps)$. A coreset of size $O(k/\eps^2)$ that is a bit weaker (preserves the spectral norm instead of the Frobenius norm) but still satisfies our coreset definition was suggested by Cohen, Nelson, and Woodruff in~\cite{CohenNW16}. This coreset is a generalization of the breakthrough result by Batson, Spielman, and Srivastava~\cite{batson2012twice} that suggested such a coreset for $k=d-1$. Their motivation was graph sparsification, where each point is a binary vector of 2 non-zeroes that represents an edge in the graph. An open problem is to reduce the running time and understand the intuition behind this result.




\subsubsection{Subspace Clustering}


A generalization of both the $k$-means and the linear $j$-subspace problem is linear $j$-subspace $k$-clustering.
Here the idea is to replace the cluster centers in the $k$-means definition by linear subspaces and then to minimize
the squared Euclidean distance to the nearest subspace.
The idea behind this problem formulation is that the important information of the input points/vectors lies in their
direction rather than their length, i.e., vectors pointing in the same direction correspond to the same type of
information (topics) and low dimensional subspaces can be viewed as combinations of topics describe by basis vectors
of the subspace.
For example, if we want to cluster webpages by their TFIDF (term frequency inverse document frequency) vectors that contain for each word its frequency inside a given webpage divided by its frequency over all webpages, then a subspace might be spanned by one basis vector for each of the words ``computer",``laptop", ``server", and  ``notebook", so that the subspace spanned by these vectors contains all webpages that discuss different types of computers.

A different view of subspace clustering is that it is a combination of clustering and PCA: The subspaces provide for each cluster
the most important dimensions, since for one fixed cluster the subspace that minimizes the sum of squared distances is
the space spanned by the right singular vectors of the restricted matrix. First provable PCA approximation of Wikipedia were obtained using coresets in~\cite{feldman2016dimensionality}.

\begin{definition}
[Linear/Affine $j$-subspace $k$-clustering]
Let $A\in\REAL^{n \times d}$. The linear/affine $j$-subspace $k$-clustering problem is to find a set $C$ that is the union of $k$ linear/affine $j$-dimensional subspaces, such that the sum of squared distances to the rows of $A$,
$
\dist^2(A, C),
$
is minimized over every such set $C$.
\end{definition}

Notice that $k$-means clustering is affine subspace clustering for $j=0$ and the linear/affine $j$-subspace problem
is linear/affine $j$-subspace $1$-clustering.
An example of a linear $1$-subspace $2$-clustering is visualized in Figure~\ref{fig:projectiveclustering}.

\begin{figure}
\begin{center}
\includegraphics{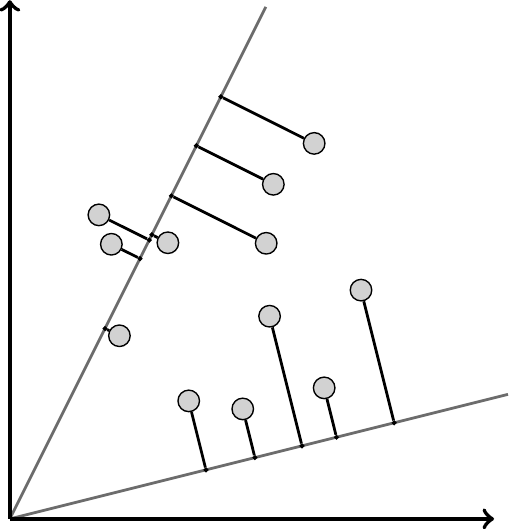}
\end{center}
\caption{Points clustered with two arbitrary 1-dimensional \lq centers\rq, an example for subspace clustering with $j=1$ and $k=2$\label{fig:projectiveclustering}.}
\end{figure}

Affine $j$-subspace $k$-clustering is NP-hard to approximate, even for $j=1$ and $d=2$.
This is due to a result by Megiddo and Tamir~\cite{MT82}, who show that it is NP-complete to decide whether a set of $n$ points in $\REAL^2$ can be covered by $k$ affine subspaces of dimension one. Any multiplicative approximation would have to decide whether it is possible to find a solution of zero cost.
Feldman, Fiat and Sharir~\cite{FFS06} give a $(1+\varepsilon)$-approximation algorithm for the affine $1$-subspace $k$-clustering problem (which is called $k$-line mean problem in their paper) for constant $d,k$ and $\epsilon$.

Deshpande, Rademacher, Vempala and Wang propose a polynomial time $(1+\epsilon)$-approximation algorithm for the $j$-subspace $k$-clustering problem~\cite{deshpande2006matrix} when $k$ and $j$ are constant.
Newer algorithms with faster running time are based on the sensitivity sampling framework by Feldman and Langberg~\cite{FL11}. We discuss~\cite{FL11} and the results by Varadarajan and Xiao \cite{VX12-soda} in detail in Section~\ref{sec:projclustering}.

\newcommand{\CC}{\mathcal{C}}
\paragraph{$\mathcal C$-Clustering under $\ell_2^2$-Distance}		
In order to keep our notation concise we will summarize the above problems in the larger class of clustering problems under $\ell_2^2$ distance, which is defined as follows. Let $\CC$ be a family of subsets of $\REAL^d$. The set $\CC$ can be thought of as a set of candidate solutions and in this paper we will typically think of each $C \in \mathcal C$ as a union of $k$ centers, i.e., in $k$-means clustering
$C$ is a set of $k$ points, in $j$-subspace $k$-clustering, $C$ is the union of $k$ subspaces, etc.

\begin{definition}[$\mathcal C$-Clustering Problem under $\ell_2^2$-distance]\label{c-clustering-problem}
Given a matrix $A\in\REAL^{n\times d}$, and a set $\mathcal C$ of sets in $\REAL^d$, the $\mathcal C$-clustering problem under $\ell_2^2$-distance is to compute a set $C \in \mathcal C$ that minimizes
$
\dist^2(A,C)
$.
\end{definition}
It is easy to see that the previously mentioned problems are special cases of the $\mathcal C$-Clustering problem under $\ell_2^2$-distance for different choices of $\mathcal C$. For example, when we choose $\mathcal C$ to be the family of all $j$-dimensional subspaces of $\REAL^d$ we obtain the $j$-subspace problem or when $\mathcal C$ is the family of sets of $k$ centers we obtain the $k$-means problem.


\subsection{Coresets and Dimensionality Reductions}\label{intro:coresets}


A coreset for an optimization problem is a (possibly weighted) point set that approximates the cost (sum of squared distances) of every feasible solution to the problem up to a small factor. In the case of a clustering problem as defined above, the set of feasible solutions is simply the set $\mathcal C$.
There are a number different definitions for coresets for clustering problems that have different properties.
A commonly used definition for the $k$-means problem goes back to the work of Har-Peled and Mazumdar \cite{HM04}: a coreset is a weighted set of points that approximates the sum of squared distances of the original input set to every candidate solution, up to a factor of $1\pm\eps$.
In this paper we introduce a new definition for a coreset that generalizes the definition of Har-Peled and Mazumdar \cite{HM04}. The main difference is that we allow to have an additive constant $\Delta$ that may be added to the coreset cost.
The main idea behind this definition is that in high dimensional spaces that we can partition the input data into a ''pseudorandom'' part, i.e., noise,
and a structured part. The pseudorandom part can then be removed from the data and will just result in an additive constant and the true information
is maintain in the structured part.
The value of the additive constant $\Delta$ may depend on the input point set $A$, the value of $\epsilon$ and the family $\mathcal C$, but it must not depend on the particular choice of $C$, i.e., for all $C\in \mathcal C$ the value of $\Delta$ will be identical. Also note that this does {\bf not} introduce an additive error, i.e., the desired value $\dist^2(A,C)$ is approximated up to a multiplicative factor of $1\pm\eps$.
Below is the definition of a coreset as it is used in this paper.

\begin{definition}[coreset for $\mathcal C$-clustering under $\ell_2^2$-distance]\label{coresetdef}
Let $\mathcal C$ be a family of non-empty sets in $\REAL^d$.
Let $A\in\REAL^{n\times d}$, $k\geq1$ be an integer, and $\eps>0$. A tuple $(S,\Delta,w)$ of a matrix $S\in\REAL^{m\times d}$ with a vector of $n$ non-negativ weights $w=(w_1,\ldots,w_m)\in \REAL^m$ associated with its rows and a value $\Delta = \Delta(A, \epsilon,\mathcal C)$
 is an \emph{$\epsilon$-coreset for the $\mathcal C$-clustering problem} under $\ell_2^2$-distance for $A$, if for every $C\in \mathcal C$ we have
$$
(1-\epsilon) \dist^2(A,C) \le \sum_{i=1}^m w_i\dist^2(S_{i*},C) + \Delta \le (1+\epsilon) \cdot \dist^2(A,C).
$$
\end{definition}

 In the first place it seems to be surprising that the addition of $\Delta$ helps us to construct smaller coresets. The intuition is that
if the input is in high dimensional space and the shape is contained in a low-dimensional space, we can split the contribution of any point to the
$\ell_2^2$-distance into a part that corresponds to its distance to some low-dimensional subspace (that may depend on the considered shape of the cluster centers) and its contribution inside the subspace. By projecting the points to a minimum cost subspace of sufficient dimensionality, we can reduce the first part and get a low-dimensional point set. $\Delta$ takes care of the reduced costs.

Many coreset constructions (without $\Delta$) have been proposed for the $k$-means problem. Early algorithms compute exponential grids or similar geometric structures to merge close enough points into coreset points~\cite{HM04,FS05,HPK07}. This approach leads to a number of coreset points which is exponential in the dimension.
Chen~\cite{C09} showed how to reduce the size to a polynomial in $k$, $\epsilon$, $\log n$ and $d$ by combining geometric arguments with sampling.
Further improvement was then based on refined sampling approaches. Langberg and Schulman~\cite{LS10} defined the \emph{sensitivity} of an input point and showed how to compute coresets of size $\tilde{\mathcal{O}}(d^2k^3\epsilon^{-2})$. The sensitivity-based framework by Feldman and Langberg~\cite{FL11} then yields coresets of size $\tilde{\mathcal{O}}(kd\epsilon^{-4})$ for the $k$-means problem.

For the general $j$-subspace $k$-clustering problem, coresets of small size do not exist~\cite{H04,HP06}. Edwards and Varadarajan~\cite{EV05} circumvent this problem by studying the problem under the assumption that all input points have integer coordinates. They compute a coreset for the $(d-1)$-subspace $k$-clustering problem with maximum distance instead of the sum of squared distances. We discuss their result together with the work of Feldman and Langberg~\cite{FL11} and Varadarajan and Xiao~\cite{VX12-soda} in Section~\ref{sec:projclustering}. The latter paper proposes coresets for the general $j$-subspace $k$-clustering problem with integer coordinates.

Definition~\ref{coresetdef} requires that the coreset approximates the sum of squared distance for \emph{every possible solution}. We require the same strong property when we talk about \emph{dimensionality reduction} of a point set. The definition is verbatim except that instead of a matrix $S \in \REAL^{n \times d}$, we want a matrix $S$ with $n$ rows but of lesser (intrinsic) dimension. A famous example for this idea is the application of the Johnson-Lindenstrauss-Lemma: It allows to replace any matrix $A\in \REAL^{n\times d}$ with a matrix $S \in \REAL^{n \times \mathcal{O}(n/\epsilon^2)}$ while preserving the $k$-means cost function up to a $(1+\epsilon)$-factor.
Boutsidis, Zouzias and Drineas~\cite{BZD10} develop a dimensionality reduction that is also based on a random projection onto $\Theta(k/\epsilon^{2})$ dimensions. However, the approximation guarantee is $2+\epsilon$ instead of $1+\epsilon$ as for the Johnson-Lindenstrauss-Lemma.

Drineas et. al. ~\cite{DFKVV04} developed an SVD based dimensionality reduction for the $k$-means problem. They projected onto the $k$ most important dimensions and solved the lower dimensional instance to optimality (assuming that $k$ is a constant). This gives a $2$-approximate solution. Boutsidis, Zouzias, Mahoney, and Drineas~\cite{BZMD15} show that the exact SVD can be replaced by an approximate SVD, giving a $2+\epsilon$-approximation to $k$ dimensions with faster running time.
Boutsidis et. al.~\cite{BMD09, BZMD15} combine the SVD approach with a sampling process that samples dimensions from the original dimensions, in order to obtain a projection onto features of the original point set. The approximation guarantee of their approach is $2+\varepsilon$, and the number of dimensions is reduced to $\Theta(k/\epsilon^2)$.




\subsection{Streaming algorithms}\label{prelim:streaming}


A \emph{stream} is a large, possibly infinitely long list of data items that are presented in arbitrary (so possibly worst-case) order. An algorithm that works in the data stream model has to process this stream of data on the fly. It can store some amount of data, but its memory usage should be small. Indeed, reducing the space complexity is the main focus when developing streaming algorithms. In this paper, we consider algorithms that have a constant or polylogarithmic size compared to the input that they process. The main influence on the space complexity will be from the model parameters (like the number of centers $k$ in the $k$-means problem) and from the desired approximation factor.

There are different streaming models known in the literature. A good introduction to the area is the survey by Mutukrishnan~\cite{Muthu05}. We consider the \emph{Insertion-Only} data stream model for geometric problems. Here, the stream consists of points $x_1, x_2, \ldots$ from $\mathbb{R}^d$ which arrive in arbitrary order. At any point in time $t$ (i.e., after seeing $x_1,\ldots,x_t$) we want to be able to produce an approximate solution for the data seen so far. This does not mean that we always have a solution ready. Instead, we maintain a coreset of the input data as described in Section~\ref{intro:coresets}. Since the cost of any solution is approximated by the coreset, we can always compute an approximate solution by running any approximation algorithm on the coreset (as long as the algorithm can deal with weights, since the coreset is a weighted set).

A standard technique to maintain coresets is the merge-and-reduce method, which goes back to Bentley and Saxe~\cite{bent} and was first used to develop streaming algorithms for geometric problems by Agarwal et al. \cite{AHPV04}. It processes chunks of the data and reduces each chunk to a coreset. Then the coresets are merged and reduced in a tree-fashion that guarantees that no input data point is part of more than $\mathcal{O}(\log n)$ reduce operations. Every reduce operation increases the error, but the upper bound on the number of reductions allows the adjustment of the precision of the coreset in an appropriate way (observe that this increases the coreset size). We discuss merge-and-reduce in detail in Section~\ref{sec:streaming}.

Har-Peled and Mazumdar initiated the development of coreset-based  streaming algorithms for the $k$-means problem. Their algorithm stores at most $\mathcal{O}(k \varepsilon^{-d} \log^{2d+2} n)$ during the computation. The coreset construction by Chen~\cite{C09} combined with merge-and-reduce gave the first the construction of coresets of polynomial size (in $\log n$, $d$, $k$ and $1/\varepsilon$) in the streaming model. Various additional results exist that propose coresets of smaller size or coreset algorithms that have additional desirable properties like good implementability or the ability to cope with point deletions~\cite{AMRSLS12,FGSSS13,feldman2010coresets,FS05,HPK07,LS10,BFLSY17}.
The construction with the lowest space complexity is due to Feldman and Langberg~\cite{FL11}.

Recall from Section~\ref{paragraph:kmeans} that the $k$-means problem can be approximated up to arbitrary precision when $k$ or $d$ is constant, and that the general case allows for a constant approximation. Since one can combine the corresponding algorithms with the streaming algorithms that compute coresets for $k$-means, these statements are thus also true in the streaming model.


\subsection{Our results and closely related work}\label{ourresults}



Our main conceptual idea can be phrased as follows. For clustering problems with low dimensional centers any high dimensional input
point set may be viewed as consisting of a structured part, i.e. a part that can be clustered well and a ''pseudo-random'' part, i.e.
a part that induces roughly the same cost for every cluster center (in this way, it behaves like a random point set). This idea is
captured in the new coreset definition given in Definition \ref{coresetdef}.

Our new idea and the corresponding coreset allows us to use the following approach. We show that for any clustering problem
whose centers fit into a low-dimensional subspace, we can replace the input matrix $A$ by its low-rank approximation $A_\ell$ for a
certain small rank $\ell$ that only depends on the shape and number of clusters and the approximation parameter $\epsilon$. The low
rank approximation $A_{\ell}$ may be viewed as the structured part of the input. In order to take care of the ``pseudo-random'' part, we
add the cost of projecting $A$ onto $A_{\ell}$ to any clustering.

Our new method allows us to obtain coresets and streaming algorithms for a number of problems. For most of the problems our coresets
are independent of the dimension and the number of input points and this is the main qualitative improvement over previous results.

In particular, we obtain (for constant error probability) a coreset of size
\begin{itemize}
\item
 $O(j/\epsilon)$ for the linear and affine $j$-subspace problem,
\item
 $\tilde {\mathcal{O}}(k^3/\epsilon^4)$ for the $k$-means problem\footnote{When we use $\tilde{\mathcal{O}}(X)$, then factors that are polylogarithmic in $X$ are hidden in the stated term.
},
\item
$\tilde{\mathcal{O}}(k^{\mathcal{O}(k)} \epsilon^{-4} \log^2 n)$ for the $k$-line means problem,
\item
and $\tilde{\mathcal{O}}(\log (Mn)^{h(j,k)}/\epsilon^2)$ for the $j$-dimensional subspace $k$-clustering problem when the input points are
integral and have maximum $l_2$-norm $M$ and where $h(j,k)>0$ is a function that depends only on $j$ and $k$.
\end{itemize}

We also provide detailed streaming algorithms for subspace approximation, $k$-means, and $j$-dimensional subspace $k$-clustering. We do not explicitly state
an algorithm that is based on coresets for $k$-line means as it follows using similar techniques as for $k$-means and $j$-dimensional subspace
$k$-clustering and a weaker version also follows from the subspace $k$-clustering problem with $j=1$.

Furthermore, we develop a different method for constructing a coreset of size independent of $n$ and $d$ and show that this construction works for a restricted
class of Bregman divergences.

The SVD and its ability to compute the optimal solution for the linear and affine subspace approximation problem has been known for over a century. About ten years ago, Drineas, Frieze, Kannan, Vempala, Vinay~\cite{DFKVV04} observed that the SVD can be used to obtain approximate solutions for the $k$-means problem. They showed that projecting onto the first $k$ singular vectors and then optimally solving $k$-means in the lower dimensional space yields a $2$-approximation for the $k$-means problem.

After the publication of the conference version of this work, Cohen, Elder, Musco, Musco and Persu~\cite{CEMMP15} observed that the dimensionality reduction and the coreset construction for subspace approximation can also be used for the $k$-means problem because the $k$-means problem can be seen as a subspace approximation problem with side constraints (in $\mathbb{R}^n$ instead of $\mathbb{R}^d$). By this insight, they show that $\lceil k / \epsilon\rceil$ dimensions suffice to preserve the $k$-means cost function. Additionally, they show that this is tight, i.e., projecting to less singular vectors will no longer give a $(1+\epsilon)$-guarantee.

\newcommand{\proj}{\mathrm{proj}}


\section{Coresets for the linear \mathj-subspace problem}\label{sec:linearsubspace}


We will first develop a coreset for the problem of approximating the sum of squared distances of a point set to a \emph{single}
linear $j$-dimensional subspace for an integer $j\in[1,\min\br{d,n}-1]$.
Let $L\subseteq \REAL^{d}$ be a $\j$-dimensional subspace represented by $X\in \REAL^{d \times j}$ whose columns are orthonormal and span $L$. Similarly, Let $L^{\bot}$ be the subspace that spans the orthogonal complement to $L$, represented by a $d\times (d-j)$ matrix $Y$ with orthonormal columns.

Recall that for a given matrix $A\in\REAL^{n\times d}$ containing $n$ points of dimension $d$ as its rows,
the sum of squared distances (cost) $\|AY\|_F^2$ of the points to $L$ is at least $\sum_{i=\j+1}^{\min\br{d,n}} \sigma_i^2$, where $\sigma_i$ is the $i$th singular
value of $A$ (sorted non increasingly). Furthermore, the subspace that is spanned by the first $j$ right singular vectors of $A$
achieves the minimum cost $\sum_{i=\j+1}^{\min\br{d,n}} \sigma_i^2$.

Now, we will show that $m:=\j+\lceil\j/\varepsilon \rceil-1$ appropriately chosen vectors suffice to approximate the cost of \emph{every} $\j$-dimensional subspace $L$.
We obtain these vectors by considering the singular value decomposition $A=U \Sigma V^T$. Our first step is to replace the matrix $A$ by its rank $m$ approximation $A^{(m)}$ as defined in Definition~\ref{def:rank}.
We show the following simple lemma regarding the error of this approximation with respect to squared Frobenius norm.

\begin{figure}
\begin{center}
\includegraphics{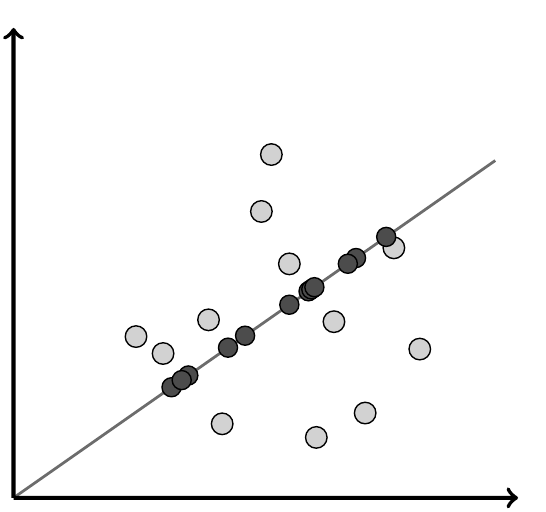}
\
\includegraphics{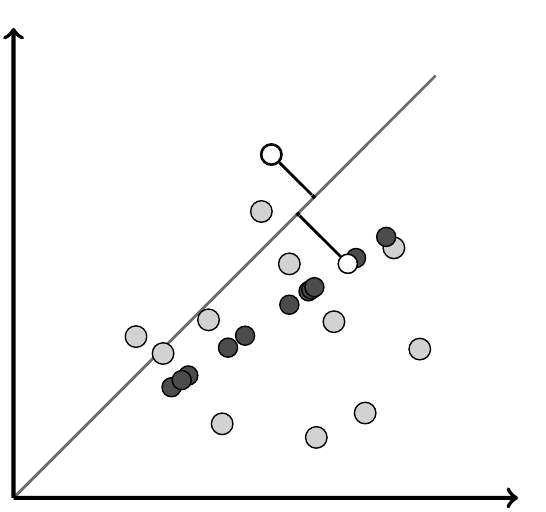}
\end{center}
\caption{
A point set is projected to a $1$-dimensional subspace. In the coreset, the projected points have to approximate the distance to any arbitrary query subspace, at least when looking at the sum of the squared distances for all input points / all projected points.
Notice that both subspaces are of the same dimension to keep the picture 2-dimensional, but in our construction the query subspace has smaller dimension.
}
\end{figure}

\begin{lemma}\label{thm:main}
Let $A\in\REAL^{n\times d}$, and let $X\in\REAL^{d\times j}$ be a matrix whose columns are orthonormal.
Let $\eps\in(0,1]$ and $ m\in[1,\min\br{n,d}-1]$ be an integer.
Then
\[
0 \le\|AX\|_F^2 - \|A^{(m)}X\|^2_F  \le j \cdot \sigma_{m+1}^2. 
\]
\end{lemma}
\begin{proof}
Using the singular value decomposition we write $A=U\Sigma V^T$ and $A^{(m)} = U \Sigma^{(m)} V^T$.
We first observe that $\|U \Sigma V^TX\|_F^2 - \|U\Sigma^{(m)}V^TX\|^2_F$ is always non-negative. Then
\begin{align}
\label{firsteq-linear} \|U\Sigma V^TX\|_F^2 - \|U\Sigma^{(m)}V^TX\|_F^2
= \|\Sigma V^TX\|_F^2 - \|\Sigma^{(m)}V^TX\|_F^2
\end{align}
holds  since $U$ has orthonormal columns. Now we observe that
$M:=V^TX$ and its rows $M_{1*},\cdots,M_{d*}$ satisfy that
\[
\begin{split}
\norm{\Sigma M}^2_F-\norm{\Sigma^{(m)}M}^2_F
&= \sum_{i=1}^{\min\br{n,d}} \sigma^2_{i} \norm{M_{i*}}_F^2-\sum_{i=1}^m \sigma^2_i \norm{M_{i*}}_F^2\\
&= \sum_{i=m+1}^{\min\br{n,d}}\sigma^2_i \norm{M_{i*}}_F^2 = \norm{(\Sigma-\Sigma^{(m)})M}^2_F.
\end{split}
\]
We can thus continue and get
\begin{align*}
(\ref{firsteq-linear}) = \|(\Sigma-\Sigma^{(m)})V^TX\|_F^2
&\le  \|\Sigma-\Sigma^{(m)}\|_2^2 \cdot \|V^TX\|_F^2
= \|\Sigma-\Sigma^{(m)}\|_2^2 \cdot \|X\|_F^2\\
 &=   j \cdot \sigma_{m+1}^2.
\end{align*}
To see the inequality, recall that the spectral norm is compatible with the Euclidean norm (\cite{QSS10}), set $D=\Sigma-\Sigma^{(m)}$ and $M=V^TX$ and observe that
\[
\|D M\|_F^2 = \sum_{\ell=1}^j \|D M_{\ast \ell}\|_2^2 \le \sum_{\ell=1}^j \|D\|_2^2 \|M_{\ast \ell}\|_2^2 = \|D\|_2^2 \|M\|_F^2.
\vspace*{-2\baselineskip}
\]
\end{proof}

In the following we will use this result to give an estimate for the $\ell_2^2$-distance to a
$j$-dimensional subspace $L$. We will represent the orthogonal complement of $L$ by a $d\times(d-j)$ matrix $Y$
with orthonormal columns. Recall that $\dist^2(A,L) = \|AY\|_F^2$. We then split $A$ into its low rank approximation
$A^{(m)}$ for some suitable value of $m$. This will be the ''structured'' part of the input. Furthermore,
we will view the cost $\Delta=\|A-A{(m)}\|_F^2$ of projecting $A$ onto the optimal $m$-dimensional subspace w.r.t. the
$m$-subspace problem as taking care of the ''pseudorandom'' part of the input.
The argument is formalized in the next lemma.

\newcommand{\msize}{j+j/\eps}
\begin{lemma}\label{pcaoff2}
Let $A\in\REAL^{n\times d}$, $j\in[1,d-1]$ be an integer and $\eps>0$.
For every integer $ m\in [1,\min\br{n,d}-1]$, and every matrix $Y\in\REAL^{d \times (d-j)}$ with orthonormal columns, by letting $\Delta = \norm{A-A^{(m)}}_F^2$, we have
 $$
 0 \le  \big(\|A^{(m)} Y\|_F^2 + \Delta\big) -\|AY\|_F^2  \le j \cdot \sigma^2_{m+1}.
 $$
\end{lemma}

\begin{proof}
By the triangle inequality and the fact that $Y$ has orthonormal columns we have \[
\|AY\|_F^2 \le \|A^{(m)}Y\|_F^2 + \|(A-A^{(m)})Y\|_F^2 \le\|A^{(m)}Y\|_F^2 + \|A-A^{(m)}\|_F^2,
\] which proves that
$\|A^{(m)} Y\|_F^2 + \Delta - \|AY\|_F^2 \geq 0$.
Let $X$ by a $d \times j$ matrix that spans the orthogonal complement of the column space of $Y$.
Using Claim \ref{Claim:Pythagorean}, $\norm{A}^2_F = \sum_{i=1}^{\min\br{n,d}} \sigma_i^2$, $\Delta = \sum_{i=m+1}^{\min\br{n,d}} \sigma_i^2$, $\norm{A^{(m)}}^2_F
= \sum_{i=1}^m \sigma_i^2$ and $\norm{A-A^{(m)}}^2_F
= \sum_{i=m+1}^{\min\br{n,d}}\sigma_i^2$ we obtain
\begin{eqnarray*}
\|A^{(m)}Y\|_F^2 + \Delta
- \|AY\|_F^2
&= & \norm{A^{(m)}}^2_F - \norm{A^{(m)}X}^2_F  + \Delta - \norm{A}^2_F + \norm{AX}^2_F\\
&= & \norm{AX}^2_F - \norm{A^{(m)}X}^2_F\\
&\le & j \cdot \sigma^2_{m+1}\\
\end{eqnarray*}
where the inequality follows from Lemma~\ref{thm:main}.
\end{proof}

\begin{corollary}\label{cor:linearsubspace:dimred}
\label{pcaoff3}
Let $A\in\REAL^{n\times d}$, $\eps>0$ and $j\in[1,d-1]$ be an integer. Let
 $m \geq \lceil j/\epsilon \rceil + j -1 $ and suppose that $m\leq \min\br{n,d}-1$.
For $\Delta = \norm{A-A^{(m)}}_F^2$ and every matrix $Y\in\REAL^{d \times (d-j)}$ whose columns are orthonormal, we have
 $$
 \|AY\|_F^2 \le  \|A^{(m)} Y\|_F^2 + \Delta \le (1+\epsilon) \cdot \|AY\|_F^2
 $$
\end{corollary}

\begin{proof}
From our choice of $m$ it follows that
\begin{equation}\label{byby}
j \sigma_{m+1}^2  \le  \eps  \cdot (m-j+1) \sigma_{m+1}^2 \le \eps \cdot \sum_{i=j+1}^{m+1} \sigma_i^2 \le \eps \cdot \sum_{i=j+1}^{\min\br{n,d}} \sigma_i^2 \le \epsilon \cdot \|AY|\|_F^2 ,
\end{equation}
where the last inequality follows from the fact that the optimal solution to the $j$-subspace problem has cost
$\sum_{i=j+1}^{\min\br{n,d}} \sigma_i^2$. Now the Corollary follows from Lemma \ref{pcaoff2}.
\end{proof}

The previous corollary implies that we can use $A^{(m)}$ and $\Delta$ to approximate the cost of any $j$-dimensional subspace. The number of rows in $A^{(m)}$ is $n$ and so it is not a small coreset. However,  $A^{(m)}$
has small rank, which we can exploit to obtain the desired coreset. In order to do so, we observe that by orthonormality of the columns of $U$ we have $\|UM\|_F^2 = \|M\|_F^2$ for any matrix $M$, which implies
that $\|U\Sigma V^TY\|_F^2 = \|\Sigma V^TY\|_F^2$.
Thus, we can replace the matrix $U \Sigma^{(m)} V^T$ in the above corollary by $\Sigma^{(m)}V^T$.
This is interesting, because all rows except the first $m$ rows of this new matrix have only $0$ entries
and so they don't contribute to $\|\Sigma^{(m)}V^TY\|_F^2$.
Therefore, we define our coreset $S$ to be the matrix $\tilde{A}$
consisting of the first $m=O(j/\eps)$ rows of $\Sigma^{(m)} V^T$.
The rows of this matrix are the coreset points. We summarize our coreset construction in the following
algorithm.
\newcommand{\apxalg}{\textsc{subspace-Coreset}}
\begin{algorithm}[ht]
    \caption{$\apxalg(A,j,\eps)$\label{algk}}
{\begin{tabbing}
\textbf{Input: \quad } \=$A\in\REAL^{n\times d}$, an integer $j\geq1$ and an error parameter $\eps>0$.\\
\textbf{Output:  } \>A pair $(S,\Delta)$ that satisfies Theorem~\ref{col1}.
\end{tabbing}}
\vspace{-0.3cm}
   Set $m \gets \min\br{n,d,j + \lceil j/\epsilon\rceil -1}$.\\
	 Set $A^{(m)}\gets U\Sigma^{(m)}V^T$ to be the $m$-rank approximation of $A$; see Definition~\ref{def:rank}.\\
     Set $S\gets \Sigma^{(m)}V^T$\\
	 Set $\Delta \gets \norm{A-A^{(m)}}_F^2$. \\
	 Set $w$ to be all $1$ vector of dimension $m$ \\
	 \Return $(S,\Delta,w)$\\
\end{algorithm}

In the following, we summarize the properties of our coreset construction.
\begin{theorem}[Coreset for $j$-subspace]\label{subapproxcor2}\label{col1}\label{thm:subcor}\label{thm:linearsubspace:coreset}
Let $A\in\REAL^{n\times d}$, $j\geq 1$ be an integer and $\eps>0$. Let $(S,\Delta,w)$ be the output of a call to $\apxalg (P,j,\eps)$; see Algorithm~\ref{algk}. Then $S\in\REAL^{m\times d}$ where $m \leq j + \lceil j/\epsilon \rceil -1$, $\Delta>0$, and for every $j$-dimensional linear subspace $L$ of $\REAL^d$ we have that
$$
 \dist^2(A,L) \le \sum_{i=1}^m w_i \cdot \dist^2(S_{i*},L) + \Delta \le (1+\eps) \cdot \dist^2(A,L) .
 $$
 This takes $O(\min\br{nd^2,dn^2})$ time.
\end{theorem}
\begin{proof}
The correctness follows immediately from Corollary \ref{pcaoff3} and the above discussion together with the observation that all $w_i$ are $1$.
The running time follows from computing the exact SVD~\cite{pearson1901}.
\end{proof}

\subsection{Discussion}

If one is familiar with the coreset literature it may seem a bit strange that the resulting point set is unweighted, i.e., we replace $n$ unweighted points by $m$ unweighted points. However, for this problem the weighting is implicitly done by scaling. Alternatively, we could also define our coreset to be the set of the first $m$ rows of $V^T$ where the $i$th row is weighted by $\sigma_i$, and $A=UDV^T$ is the SVD of $A$.

As already described in the Preliminaries, principal component analysis requires that the data is translated such that its mean is the origin of $\REAL^d$. If this is not the case, we can easily enforce this by subtracting the mean before the coreset computation. However, if we are taking the union of two or more coresets, they will have different means and cannot be easily combined. This limits the applicability to streaming algorithms to the case where we a priori know that the data set has the origin as its mean.
Of course we can easily maintain the mean of the data, but a simple approach such as subtracting it from the coresets points at the end
of the stream does not work as it invalidates the properties of the coreset.
In the next section we will
show how to develop a coreset for the affine case, which allows us to deal with data that is not a priori normalized.


\section{Coresets for the Affine \mathj-Subspace Problem}
\label{section:affine}


We will now extend our coreset to the affine $j$-subspace problem. The main idea of the new construction is very simple:
Subtract the mean of $A$ from each input point to obtain a matrix $A'$, compute a coreset $S'$ for $A'$ and then add the mean to the points in the coreset to obtain a coreset $S$. While this works in principle, there are two hurdles that we have to overcome. Firstly, we need to ensure that the mean of the coreset is $\vec 0$ before we add $\mu(A)$. Secondly, we need to scale and weight our coreset. The resulting construction is given as pseudo code in Algorithm~\ref{algkk-unweighted}.

\newcommand{\apxalgb}{\textsc{affine-$j$-subspace-Coreset}}
\begin{algorithm}[ht]
    \caption{$\apxalgb(A,j,\eps)$\label{algkk-unweighted}}
{\begin{tabbing}
\textbf{Input: \quad   }  \= $A\in\REAL^{n\times d}$, an integer $j\geq1$ and an error parameter $\eps>0$.\\
\textbf{Output: } \> A triple $(S,\Delta,w)$ that satisfies Theorem~\ref{col22}.
\end{tabbing}}
\vspace{-0.3cm}
  {
   Set $\mu(A) = \frac{1}{n}\sum_{i=1}^n A_{i*}$ \tcc{this is the mean row of $A$}
	 Set $(S',\Delta) \gets \apxalg(A-\mathds{1} \cdot \mu(A)^T, j, \epsilon)$ \label{twoalg1-unweighted} \\
   Set $\displaystyle S\gets \mathds{1} \cdot \mu(A)^T+\sqrt{\frac{m}{n}}\cdot
     \begin{bmatrix} S'\\
      -S'
     \end{bmatrix}
     $\label{four-unweighted}\\\label{alg-affine-weights-unweighted}
	 Set $w$ to be the $2m$-dimensional vector with all entries $\frac{n}{2m}$\\	
	 \Return{$(S,\Delta,w)$}
   }
\end{algorithm}

\begin{lemma}\label{lem:Cdist0}
Let $M\in\REAL^{n\times d}$ with $\mu(M)=\vec 0$ and let $C=t+L$ with $t \in L^{\bot}$ be an affine $j$-dimensional subspace of $\REAL^d$ for $j \le d-1$. Then
$
\dist^2(M,C) = \dist^2(M,L) + n\cdot ||t||^2.
$
\end{lemma}
\begin{proof}
Assume that $Y$ spans $L^{\bot}$. Then it holds that
\begin{align}
\label{translation}\dist^2(M,C) & = \sum_{i=1}^n \dist^2(M_{i\ast},C)
 = \sum_{i=1}^n \dist^2(M_{i\ast}-t,L)\\
\label{pyth1} & = \Big(\sum_{i=1}^n ||M_{i\ast}-t||^2\Big) - ||MYY^T||^2\\
\label{magic} & = \Big(\sum_{i=1}^n ||M_{i\ast}||^2+ n\cdot ||t||^2 \Big)  - ||MYY^T||^2\\
\label{pyth2} & = \dist^2(M,L) + n\cdot ||t||^2
\end{align}
where~\eqref{translation} follows because translating $M$ and $C$ by $t$ does not change the distances, ~\eqref{pyth1}, ~\eqref{magic} follows by $\mu(M)=\vec 0$
and~\eqref{pyth2} follows by Claim \ref{Claim:Pythagorean} and the fact that $Y$ has orthonormal columns.
\end{proof}

First observe that Lemma~\ref{lem:Cdist0} is only applicable to point sets with mean $\vec 0$. This is true for $A' = A - \mu(A)$, and also for any matrix that has $S'$ and $-S'$ as its rows, even if we scale it.
We know that $S'$ is a coreset for $A$ for linear subspaces, satisfying that $\dist^2(S',L)$ is approximately equal to $\dist^2(A,L)$ for any linear subspace $L$. Since we double the points, a likely coreset candidate would by
\[
S'' = \frac{1}{\sqrt{2}}      \begin{bmatrix} S'\\
      -S'
     \end{bmatrix}
\]
since $\dist^2(S'',L)=\dist^2(S',L)$ for any linear subspace $L$, and $\mu(S'')=\vec 0$ is also satisfied. What is the problem with $S''$? Again consider Lemma~\ref{lem:Cdist0}. Assume that we have a linear subspace $L$ and start to move it around, obtaining an affine subspace $C=L+t$ for $t \in L^\bot$. Then the distance of $A$ and $S''$ increases by a multiple of $||t||^2$ -- but the multiple depends on the number of points. Thus, we either need to increase the number of points in $S''$ (clearly not in line with our idea of a coreset), or we need to weight the points by $n/2m$. However, $(2m/n) \dist^2(S'',L)$ is not comparable to $\dist^2(A',L)$ anymore. To compensate for the weighting, we need to scale $S''$ by $\sqrt{2m/n}$ (notice that the $\sqrt{2}$ now cancels out). This is how we obtain Line~\ref{alg-affine-weights-weighted} of Algorithm~\ref{algkk-weighted}. We conclude by stating and showing the coreset guarantee. Notice that all rows in $S''$ receive the same weight, so we do not need to deal with the weights explicitly and rather capture the weighting by a multiplicative factor in the following theorem.

\begin{theorem}[Coreset for affine $j$-subspace]\label{subapproxcor3}\label{col22}\label{thm:affinesubspace:coreset}
Let $A\in\REAL^{n\times d}$, $j\in[1,d-1]$ be an integer, and $\eps>0$.
Let $(S,\Delta,w)$ be the output of a call to $\apxalgb (P,j,\eps)$; see Algorithm~\ref{algkk-unweighted}. Then $S\in\REAL^{(2m)\times d}$ where $m\leq j + \lceil j/\epsilon \rceil -1$, $\Delta>0$, and for every \emph{affine} $j$-dimensional subspace $C$ of $\REAL^d$ we have that
 $$
 \dist^2(A,C) \le \sum_{i=1}^{2m} w_i \cdot \dist^2(S_{i*},C) + \Delta \le (1+\eps) \cdot \dist^2(A,C).
 $$
 This takes $\min\br{nd^2,dn^2}$ time.
\end{theorem}
\begin{proof}
The running time follows from computing the exact SVD~\cite{pearson1901}.
Let $C=p+L$ be any affine $j$-dimensional subspace of $\REAL^d$, where $L$ is a linear subspace. We assume w.l.o.g. that $p$ is chosen such that $p - \mu(A) \in L^\bot$.
Let $A'$ be the translation of $A$ by $-\mu(A)$, \ie\ $A'_{i\ast} = A_{i\ast}-\mu(A)$ for all $i\in[n]$.
Set $S'' = S - \mathds{1} \cdot \mu(A)^T$ and observe that $\dist^2(S'',L) = (2m/n) \cdot \dist^2(S',L)$.
This fact together with Theorem~\ref{thm:subcor} yields that
\begin{align}\vspace*{-\baselineskip}
\label{eq:guar} \dist^2(A',L) \le (n/(2m)) \cdot \dist^2(S'',L) + \Delta \le (1+\eps) \cdot \dist^2(A',L)
\end{align}
because $S'$ was constructed as a coreset for $A'$.
Set $t = p - \mu(A)$, i.e., $L + t = C - \mu(A)$. By our assumption above, $t \in L^\bot$. We get that

\begin{align*}
 \frac{n}{2m} \cdot \dist^2(S,C) - \dist^2(A,C)
=~& \frac{n}{2m} \cdot \dist^2(S'',C-\mu(A)) - \dist^2(A',C-\mu(A))\\
=~& \frac{n}{2m} \cdot (\dist^2(S'',L) + (2m) \cdot ||t||^2) - (\dist^2(A',L) + n\cdot ||t||^2)\\
=~& \frac{n}{2m} \cdot \dist^2(S'',L) - \dist^2(A',L)
\end{align*}
where we first translate $S$ and $A$ by $-\mu(A)$ and then exploit $\mu(A')=\mu(S'')=\vec 0$ to use Lemma~\ref{lem:Cdist0} twice.
Now \eqref{eq:guar} yields the statement of the theorem since all $w_i$ equal $n/(2m)$.
\end{proof}

\subsection{Weighted Inputs}

There are situation where we would like to apply the coreset computation on a weighted set of input points (for example,
lateron in our streaming algorithms). If the point weights are integral then we can reduce to the unweighted case by replacing a point by a corresponding
number of copies. Finally, we observe that the same argument works for general point weights, if we reduce the problem
to an input set where each point has a weight $\delta$ and we let $\delta$ go to $0$. This blows up the input set, but
we will only require this to argue that the analysis is correct. In the algorithm we use that for the linear subspace problem
scaling by a factor of $\sqrt{w}$ is equivalent to assigning a weight of $w$ to a point.
The algorithm can be found below.

\begin{algorithm}[ht]
    \caption{\textsc{affine-$j$-subspace-Coreset-Weighted-Inputs}$(A,j,\eps,w)$\label{algkk-weighted}}
{\begin{tabbing}
\textbf{Input: \quad   }  \= $A\in\REAL^{n\times d}$, an integer $j\geq1$ and an error parameter $\eps>0$.\\
\> weight vector $w=(w_1,\dots, w_n)$\\
\textbf{Output: } \> A coreset $(S,\Delta,w)$ that satisfies the guarantees of Theorem~\ref{col22}.
\end{tabbing}}
\vspace{-0.3cm}
  {
	 Set $W= \sum_{i=1}^n w_i$\\
   Set $\mu(A) = \frac{1}{W}\sum_{i=1}^n w_i \cdot A_{i*}$ \\
	 Set $B$ to be the $n \times d$-matrix with rows $B_{i*} = \sqrt{w_i} (A_{i*} - \mu(A)^T)$ \\
	Set $(S',\Delta) \gets \apxalg(B, j, \epsilon)$ \label{twoalg1-weighted} \\
   Set $\displaystyle S\gets \mathds{1} \cdot \mu(A)^T+\sqrt{\frac{m}{W}}\cdot
     \begin{bmatrix} S'\\
      -S'
     \end{bmatrix}
     $\label{four-weighted}\\\label{alg-affine-weights-weighted}
	 Set $w$ to be the $2m$-dimensional vector with all entries $\frac{W}{2m}$\\	
	 \Return{$(S,\Delta,w)$}
   }
\end{algorithm}


\section{Dimensionality Reduction for Clustering Problems under \elltt-distance}\label{sec:affinesubspace}


In this chapter we show that the results from the previous chapter can be used to define a general dimensionality reduction for clustering
problems under the $\ell_2^2$-distance, if the cluster centers are contained in a low dimensional subspace. For example, in $k$-means clustering
the cluster centers are contained in a $k$-dimensional subspace.
To define the reduction, let $L$ be an arbitrary linear $j$-dimensional subspace represented by a $d \times j$ matrix $X$ with orthonormal columns and with $Y$ being an $d \times (d-j)$ matrix with orthonormal columns that span $L^{\bot}$. We can think of $L$ as being an arbitrary subspace that contains
a candidate solution to the clustering problem at hand.
Our first step will be to show that if we project both $A$ and $A^{(m)}:=U\Sigma^{(m)}V^T$ on $L$ 
 by computing $A XX^T$ and $A^{(m)}XX^T$, then  the sum of squared distances between the corresponding rows of the projection is small compared to the cost of the projection.
In other words, after the projection the points of $A$ will on average be relatively close to their counterparts of $A^{(m)}$.
Notice the difference from Lemma~\ref{thm:main}: In Lemma ~\ref{thm:main}, we showed that if we project $A$ to $L$ and sum up the squared \emph{lengths} of the projections, then this sum is approximately the sum of the squared lengths of the projections of $A^{(m)}$. In the following corollary, we look at the distances between a projection of a point from $A$ and the projection of the corresponding point in $A^{(m)}$, then we square these distances and show that the sum of them is small.
\begin{corollary}
\label{col2}
Let $A\in \REAL^{n\times d}$, $\eps>0$. Let $j\in[1,d-1]$ and $m\ge\j+\lceil \j/\eps \rceil -1$ be a pair of integers, and suppose that $m\leq \min\br{n,d}-1$.
Let $X \in \REAL^{d \times j}$ be a matrix whose columns are orthonormal, and let $Y\in \REAL^{d\times (d-j)}$ be a matrix with orthonormal columns that span
the orthogonal complement of the column space of $X$. Then
$$
0\leq \norm{A XX^T - A^{(m)}XX^T}_F^2 \le \eps \cdot \|AY\|_F^2.
$$
\end{corollary}
\begin{proof}
Using the singular value decomposition of $A$ we get
\[
\begin{split}
&\norm{AXX^T - A^{(m)}XX^T}_F^2
=\|(A-A^{(m)})XX^T\|_F^2
=\|(A-A^{(m)})X\|_F^2\\
&=\|(U\Sigma V^T-U\Sigma^{(m)} V^T)X\|_F^2 =\|U(\Sigma -\Sigma^{(m)}) V^TX\|_F^2\\
&=\|(\Sigma -\Sigma^{(m)}) V^TX\|_F^2
\leq \sum_{i=m+1}^{m+j}\sigma_i^2
\leq j\sigma_{m+1}^2,
\end{split}
\]
where the first and second equality follows since the columns of $X$ and $U$ respectively are orthonormal.
By~\eqref{byby}, $j\sigma_{m+1}^2\leq \eps\norm{AY}_F^2$, which proves the theorem.
\end{proof}

In the following we will prove our main dimensionality reduction result. The result states that we can use $A^{(m)}$ as
an approximation for $A$ in any clustering or shape fitting problem of low dimensional shapes, if we add $\|A - A^{(m)}\|_F^2$ to the cost.
Observe that this is simply the cost of projecting the points on the subspace spanned by the first $m$ right singular vectors, i.e., the cost of ``moving'' the points in $A$ to
$A^{(m)}$. In order to do so, we use the following \lq weak triangle inequality\rq, which is well known in the coreset literature.

\begin{corollary}\label{lem:movement:simple}
Let $\varepsilon>0$, $A\in \mathbb{R}^{n\times d}$ and $B \in \mathbb{R}^{n\times d}$ be two matrices.
Let $C \subset \REAL^d$ be an arbitrary nonempty set.
Then
\[
| \dist^2(A,C) - \dist^2(B,C) | \le \varepsilon \cdot \dist^2(A,C) + (1+\frac{1}{\epsilon}) \cdot \|A-B\|_F^2.
\]
\end{corollary}
\begin{proof}
Let $p$ be a row in $B$ and $q$ be a corresponding row in $A$.
Using the triangle inequality,
\begin{equation}\label{aab}
\begin{split}
& |\dist^2(p,C)-\dist^2(q,C)| \\
&=|\dist(p,C)-\dist(q,C)|\cdot(\dist(p,C)+\dist(q,C))\\
&\leq \norm{p-q}_2\cdot(2\dist(p,C)+\norm{p-q}_2)\\
&= \norm{p-q}^2_2+2\dist(p,C)\norm{p-q}_2\\
&= \norm{p-q}^2_2+2\sqrt{\eps} \cdot \dist(p,C)\cdot\frac{\norm{p-q}_2}{\sqrt{\eps}}\\
&\leq \norm{p-q}^2_2+\eps\cdot\dist^2(p,C)+\frac{\norm{p-q}^2_2}{\eps}\\
&= \eps\cdot\dist^2(p,C)+ (1+\frac{1}{\eps}) \cdot {\norm{p-q}^2_2}
\end{split}
\end{equation}
where the last inequality is since $2ab\leq a^2+b^2$ for every $a,b\in \REAL$.
Summing the last inequality over all the $n$ rows of $A$ and $B$ yields
\[
\begin{split}
| \dist^2(A,C) - \dist^2(B,C) |
&=
\left| \sum_{i=1}^n \dist^2(A_{i*},C) - \dist^2(B_{i*},C) \right|\\
&\leq  \sum_{i=1}^n \left|\dist^2(A_{i*},C) - \dist^2(B_{i*},C) \right|\\
&\leq  \sum_{i=1}^n \left(\eps\cdot\dist^2(A_{i*},C) + (1+\frac{1}{\eps}) \cdot \norm{A_{i*}-B_{i*}}^2_F\right)\\
&=  \eps\ \cdot \dist^2(A,C) + (1+\frac{1}{\eps})\cdot \norm{A-B}^2_F.\\
\end{split}
\]
\end{proof}

The following theorem combines Lemma~\ref{pcaoff2} with Corollary \ref{col2} and \ref{lem:movement:simple} to get the dimensionality reduction result.

\begin{theorem}\label{thm:main2}
Let $A\in \REAL^{n \times d}$, $j\in[1,d-1]$ be an integer, and $\epsilon\in (0,1]$.
Let $m \ge \lceil 8 j/\eps^2\rceil -1$ and suppose that
Let $m \le \min\br{n,d} -1$.
Then for any non-empty set $C$, which is contained in a $j$-dimensional subspace, we have
$$
\left|  \left(\dist^2(A^{(m)},C) + \norm{A-A^{(m)}}_F^2 \right)- \dist^2(A,C) \right| \le \eps \cdot \dist^2(A,C).
$$
\end{theorem}
\begin{proof}
Let $L$ denote the $j$-dimensional subspace that spans $C$, and
let $X\in\REAL^{d\times j}$ be a matrix with orthonormal columns that span $L$.
Let $Y\in\REAL^{d\times (d-j)}$ denote a matrix with orthonormal columns that span the orthogonal complement of $L$.
By the Pythagorean theorem (see also Claim \ref{Claim:Pythagorean}) we get
$$\dist^2(A^{(m)},C) = \|A^{(m)}Y\|_F^2 + \dist^2(A^{(m)}XX^T,C)$$ and
\begin{equation}\label{acaca}
\dist^2(A,C) = \|AY\|_F^2 + \dist^2(AXX^T,C).
\end{equation}
Hence,
\begin{align}
    \nonumber& \left|  \left(\dist^2(A^{(m)},C) + \norm{A-A^{(m)}}^2 \right) - \dist^2(A,C) \right| \\
 \nonumber=   & \left| \|A^{(m)}Y\|_F^2  + \dist^2(A^{(m)}XX^T,C)
+ \norm{A-A^{(m)}}^2
- \left(\|AY\|_F^2 + \dist^2(AXX^T,C) \right) \right|\\
\label{do1}\le & \Big| \|A^{(m)}Y\|_F^2 + \|A-A^{(m)}\|_F^2 - \|AY\|_F^2 \Big| + \left|\dist^2(A^{(m)}XX^T,C) -\dist^2(AXX^T,C)\right| \\
\label{do2}\le & \frac{\varepsilon^2}{8}  \cdot \|AY\|_F^2 + \left| \dist^2(A^{(m)}XX^T,C) - \dist^2(AXX^T,C) \right|\\
\label{do3}\le & \frac{\varepsilon^2}{8}  \cdot \dist^2(A,C) + \left| \dist^2(A^{(m)}XX^T,C) - \dist^2(AXX^T,C) \right|,
\end{align}
where~\eqref{do1} is by the triangle inequality,~\eqref{do2} is by replacing $\eps$ with $\eps^2/8$ in Corollary~\ref{pcaoff3},
and ~\eqref{do3} is by~\eqref{acaca}.

By Corollary~\ref{col2},
\[
\norm{A^{(m)}XX^T-A XX^T}_F^2 \le \frac{\eps^2}{8} \cdot \|AY\|_F^2.
\]
Since $C$ is contained in $L$, we have $\norm{AY}^2_F = \dist^2(A,L) \leq \dist^2(A,C)$.
Using Corollary~\ref{lem:movement:simple} while substituting $\epsilon$ by $\eps/2$, $A$ by $A^{(m)}XX^T$ and $B$ by $AXX^T$ yields
\begin{equation}
| \dist^2(A^{(m)}XX^T,C)-\dist^2(A XX^T,C)| \le \frac{\eps}{4} \cdot \dist^2(A XX^T,C) + (1 + \frac{4}{\epsilon}) \cdot \norm{A^{(m)}XX^T-A XX^T}_F^2.
\end{equation}
By~\eqref{acaca}, $\dist^2(A XX^T,C)\leq \dist^2(A,C)$.
Combining the last two inequalities with ~\eqref{do3} proves the theorem, as
\[
\begin{split}
  &\left|  \left(\dist^2(A^{(m)},C) + \norm{A-A^{(m)}}_F^2 \right) - \dist^2(A,C) \right| \\
	\le &  \; \frac{\varepsilon^2}{8}  \cdot \dist^2(A,C) + \frac{\eps}{4}\cdot \dist^2(A,C) + \frac{\eps^2}{8} \cdot (1+\frac{4}{\eps}) \cdot \dist^2(A,C)\\
\leq & \; \eps \cdot \dist^2(A,C),
\end{split}
\]
where in the last inequality we used the assumption $\eps\leq1$.
\end{proof}

Theorem~\ref{thm:main2} has a number of surprising consequences. For example, we can solve $k$-means or any subspace clustering problem
approximately by using $A^{(m)}$ instead of $A$.

\newcommand{\dimappalg}{\textsc{Dimensionality-Reduction-$k$-means}}
\begin{algorithm}[ht]
  \caption{\dimappalg({$A$, $k$, $\epsilon$, $\alpha$})\label{dimappalg}}
{\begin{tabbing}
\textbf{Input: \quad   }  \= $A\in\REAL^{n\times d}$, an integer $k\geq1$ and error parameters $\alpha\geq 0$ and $\eps\in(0,1/2)$.\\
\textbf{Output: } \> A $\alpha(1+\eps)$-approximation $C$ for $k$-means; see Corollary~\ref{cor:demred:kmeans}.
\end{tabbing}}
	\SetAlgoLined\DontPrintSemicolon
  \vspace*{-\baselineskip}
   $m = k + \lceil 72 k/\epsilon^2\rceil -1$\;
   Compute the singular value decomposition $A = U \Sigma V^T$\;
	 Set $A^{(m)} = U \Sigma^{(m)} V^T$, where $\Sigma^{(m)}$ contains only the first $m$ diagonal entries of $\Sigma$ and
	is $0$ otherwise \;
	 Let $C$ be a set of $k$ centers that is an $\alpha$-approximation to the optimal $k$-means clustering of $A^{(m)}$\;
	 \Return{$C$}\;
\end{algorithm}

\begin{corollary} [Dimensionality reduction for $k$-means clustering]\label{cor:demred:kmeans}
Let $A\in\REAL^{n\times d}$, $k\geq1$ be an integer, $\eps\in(0,1/3]$,  and $\alpha\geq1$.
Suppose that $C$ is the output set of a call to\\ $\dimappalg(A,k,\eps,\alpha)$.
Then $C$ is an $(\alpha (1+\epsilon))$-approximation to the optimal $k$-means
clustering problem of $A$. In particular, if $\alpha=1$, then $C$ is a $(1+\epsilon)$-approximation.
\end{corollary}
\begin{proof}
Let $\eps\in(0,1/3]$ be an input parameter. Let $C^*$ denote an optimal set of $k$ centers for the $k$-means objective function on input $A$. We apply Theorem \ref{thm:main2} with parameter $\epsilon/3$ and for both $C$ and $C^*$ in order to get that
$$
\big|\dist^2(A^{(m)},C) + \norm{A-A^{(m)}}_F^2 - \dist^2(A,C) \big| \le \epsilon/3 \cdot \dist^2(A,C).
$$
and
$$
\big|\dist^2(A^{(m)},C^*) + \norm{A-A^{(m)}}_F^2 - \dist^2(A,C^*) \big| \le \epsilon/3 \cdot \dist^2(A,C^*).
$$
From these inequalities we can deduce that
$$
(1-\epsilon/3) \cdot \dist^2(A,C) \le  \dist(A^{(m)},C) +\norm{A-A^{(m)}}_F^2
$$
and
$$
\dist^2(A^{(m)},C^*) + \norm{A-A^{(m)}}_F^2 \le (1+\epsilon/3) \cdot \dist^2(A,C^*).
$$
Since $C$ is an $\alpha$-approximation, we also have $\dist(A^{(m)},C) \le \alpha \cdot \dist(A^{(m)},C^*)$.
It follows that
\[
\begin{split}
(1-\epsilon/3) \cdot \dist^2(A,C) &\le  \dist(A^{(m)},C) + \norm{A-A^{(m)}}_F^2\\
&\le \alpha \cdot  \dist(A^{(m)},C^*) + \norm{A-A^{(m)}}_F^2 \\
&\le \alpha \cdot  \big(\dist(A^{(m)},C^*) + \norm{A-A^{(m)}}_F^2 \big) \\
&\le \alpha \cdot (1+\epsilon/3) \cdot \dist^2(A,C^*).
\end{split}
\]
Since $\epsilon <1/3$ we have $\frac{1+\epsilon/3}{1-\epsilon/3} \le 1+\epsilon$ and so the corollary follows.
\end{proof}

Our result can be immediately extended to the affine $j$-subspace $k$-clustering problem. The proof is similar to the proof of the previous corollary.

\newcommand{\affinejk}{\textsc{affine $j$-subspace $k$-clustering approximation}}
\begin{algorithm}[h!t]
  \caption{\affinejk($A$, $k$, $\epsilon$)}
{\begin{tabbing}
\textbf{Input: \quad   }  \= $A\in\REAL^{n\times d}$, an integer $k\geq1$ and error parameters $\alpha\geq 0$ and $\eps\in(0,1/2)$.\\
\textbf{Output: } \> A $\alpha(1+\eps)$-approximation $C$ for the affine $j$-subspace $k$-clustering problem;\\
\> see Corollary~\ref{cor:dimred:jk}.
\end{tabbing}}
\vspace*{-\baselineskip}
  \SetAlgoLined\DontPrintSemicolon
   $m = k(j+1) + \lceil 72 \cdot k(j+1)/\epsilon^2\rceil -1$\;
   Compute the singular value decomposition $A = U \Sigma V^T$\;
	 Set $A^{(m)} = U \Sigma^{(m)} V^T$, where $\Sigma^{(m)}$ contains only the first $m$ diagonal entries of $\Sigma$ and
	is $0$ otherwise \;
	 Let $C$ be a set of $k$ affine $j$-subspaces that is an $\alpha$-approximation to the optimal affine $j$-subspace $k$-clustering\;
	 \Return{$C$}\;
\end{algorithm}

\begin{corollary} [Dimensionality reduction for affine $j$-subspace $k$-clustering]\label{cor:dimred:jk}
A call to Algorithm $\affinejk$ returns an $(\alpha (1+\epsilon))$-approximation to the optimal solution for the affine $j$-subspace $k$-clustering problem on input $A$. In particular, if $\alpha=1$, the solution is a $(1+\epsilon)$-approximation.
\end{corollary}


\section{Small Coresets for \mathC-Clustering Problems}\label{sec:coresets:generalandk-means}


In this section we use the result of the previous section to prove that any $\mathcal C$-clustering problem, which is closed under rotations and reflections, has a  coreset of cardinality independent of the dimension of the space, if it has a coreset for a constant number of dimensions.

\begin{definition}\label{def:closedsubset}
A set $\mathcal C$ of non-empty subsets of $\REAL^d$ is said to be closed under rotations and reflections, if
for every $C\in \mathcal C$ and every orthogonal matrix  $U\in\REAL^{d\times d}$ we have $U(C) \in \mathcal C$, where $U(C):=\{Ux: x \in C\}$.
\end{definition}

In the last section, we showed that the projection $A^{(m)}$ of $A$ approximates $A$ with respect to the $\ell_2^2$-distance to any low dimensional shape. $A^{(m)}$ still has $n$ points, which are $d$-dimensional but lie in an $m$-dimensional subspace.
To reduce the number of points, we will apply known coreset constructions to $A^{(m)}$ within the low dimensional subspace. At first glance, this means that the coreset property
 only holds for centers that are also from the low dimensional subspace, but of course we want that the centers can be chosen from the full dimensional space. We get around this problem by applying the coreset constructions to a slightly larger space than the subspace that $A^{(m)}$ lies in. The following lemma provides us with the necessary tool to complete the argumentation.

\begin{lemma}\label{rotationlemma}
Let $S$ be an $r$-dimensional subspace of $\REAL^d$ and let $L$ be an $(r+j)$-dimensional subspace of $\REAL^d$
that contains $S$. Let $V$ be a $j$-dimensional subspace of $\REAL^d$. Then there is an orthogonal matrix $U$ such
that $Ux=x$ for every $x\in S$, and $Uc \in L$ for every $c\in V$.
\end{lemma}
\begin{proof}
Let $B_1\in\REAL^{d\times d}$ be an orthogonal matrix whose first $r$
columns span $S$ and whose first $r+j$ columns span $L$. Let $B_2\in\REAL^{d\times d}$ be an orthogonal matrix whose first $r$ columns are the same as the first $r$ columns of $B_1$, and whose first $r+j$ columns span a subspace that contains $V$. Define the orthogonal matrix $U=B_1 B_2^T$. For every $x\in S$, the last $d-r$ entries of the vector $y=B_1^Tx$ are all zeroes, and $x=B_2y$. Thus $Ux = B_1B_2^T B_2y = B_1y =x$ as desired. Furthermore,
for every $c \in V$ there is $z\in\REAL^d$ whose last $d-(r+j)$ entries are all zeroes and $c=B_2z$. Hence, $Uc = B_1 B_2^T B_2 z = B_1 z \in L$, as desired.
\end{proof}

\begin{corollary}\label{cordanny}
Let $A\in\REAL^{n\times d}$ be a matrix of rank $r$ and let $L$ be an $(r+j+1)$-dimensional subspace of $\REAL^d$ that contains the row vectors  $(A_{i*})$ for every $1\le i \le n$. Then for every affine $j$-dimensional subspace $V$ of $\REAL^d$ there is a corresponding affine $j$-dimensional subspace $V'\subseteq L$ such that for every $i\in[n]$ we have
\[
\dist(A_{i*},V)=\dist(A_{i*},V').
\]
\end{corollary}
\begin{proof}
Let $A\in\REAL^{n\times d}$ be a matrix of rank $r$, let $S$ be an $r$-dimensional subspace of $\REAL^d$ that contains the row vectors
$A_{i*}$ for every $1\le i \le n$, and let $L$ be an $(r+j+1)$-dimensional subspace of $\REAL^d$ that contains $S$. Let $V$ be an arbitrary
affine $j$-dimensional subspace of $\REAL^d$ and let $V_l$ be a $(j+1)$-dimensional linear subspace that contains $V$. We apply
Lemma~\ref{rotationlemma} with $S$, $L$ and $V_l$ to obtain an orthogonal matrix $U$ such that for every $x\in S$ we have $Ux=x$ and
$Uc \in L$ for every $c\in V_l$. This implies in particular that $A_{i*} U^T = A_{i*}$ for every $1\le i \le n$ and that $Uc \in L$ for every $c \in V$.

Since a transformation by an orthogonal matrix preserves distances, we also know for $V' = \{Uc : c \in \REAL^d\}$ that
$$
\dist(A_{i*},V) = \dist(A_{i*} U^T,V') = \dist(A_{i*}, V').
$$
\end{proof}

Now consider a $\mathcal C$-clustering problem, where $\mathcal C$ is closed under rotations and reflections. Furthermore,
assume that each set $C\in \mathcal C$ is contained in a $j$-dimensional subspace. Our plan is to apply the above Corollary
to the matrix $A^{(m)}$. Then we know that there is a space $L$ of dimension $m+j$ such that for every subspace $V$
there is an orthogonal matrix $U$ that moves $V$ into $L$ and keeps the points described by the rows of $A^{(m)}$ unchanged.
Furthermore, since applying $U$ does not change Euclidean distance
we know that the sum of squared distances of the rows of $A^{(m)}$ to $C$ equals the sum of squared distances to
$U(C):=\{Ux: x \in C\}$ and $U(C)$ is contained in $L$ (by the above Corollary) and in $\mathcal C$ since $\mathcal C$ is closed
under rotations and reflections.

Now assume that we have a coreset for the subspace $L$. As oberserved, we have $U(C) \in \mathcal C$ and $U(C) \subseteq L$.
In particular, the sum of squared distances to $U(C)$ is approximated by the coreset. But this is identical to the sum of
squared distances to $C$ and so this is approximated by the coreset as well.

Thus, in order to construct a coreset for a set of $n$ points in $\REAL^d$ we proceed as follows. In the first step we use the dimensionality reduction from
the previous chapter to reduce the input point set to a set of $n$ points that lies in an $m$-dimensional subspace. Then we construct a coreset for an $(m+j)$-dimensional subspace that contains the low-dimensional point set. By the discussion above, the output will be a coreset for the original set in the original $d$-dimensional space.

\begin{theorem}[Dimensionality reduction for coreset computations]\label{thm:dim}
Let $\epsilon \in (0,1]$ and $A \in \REAL^{n\times d}$. Let $\CC$ be a (possibly infinite) set of non-empty subsets of $\REAL^d$ that is closed under rotations and reflections such that each $C\in \CC$ is contained in a $j$-dimensional subspace. Let $m= \min\br{n,d,j + \lceil 32 j/\epsilon^2 \rceil}-1$ and $L$ be a subspace of dimension at most $m+j$ that contains the row vectors of $A^{(m)}$.
Suppose that $(S,\Delta',w)$ is an $(\eps/8)$-coreset (see Definition~\ref{coresetdef}) for input point set $A^{(m)}$ in the input space $L$.

Then $(S,\Delta'+\norm{A-A^{(m)}}_F^2,w)$ is an $\eps$-coreset for for the $\CC$-clustering problem in $\REAL^d$ and with input $A$, i.e.,
\[
(1-\epsilon) \cdot \dist^2(A,C) \le  \sum_{i=1}^r w_i \cdot \dist^2(S_{i*}, C) +\Delta'+\norm{A-A^{(m)}}_F^2 \leq (1+\epsilon) \cdot \dist^2(A,C).
\]
\end{theorem}

\begin{proof}
We first apply Theorem \ref{thm:main2} with $\epsilon$ replaced by $\epsilon/2$ to obtain for every $C\in \mathcal C$:
$$
\left|  \left(\dist^2(A^{(m)},C) + \norm{A-A^{(m)}}_F^2 \right)- \dist^2(A,C) \right| \le \frac{\eps}{2} \cdot \dist^2(A,C).
$$
Now let $(S,\Delta',w)$ be an $(\eps/8)$-coreset for the $\CC$-clustering problem in the subspace $L$ and with input set $A^{(m)}$.
By Corollary \ref{cordanny} and the discussion prior to Theorem \ref{thm:dim} we know that the coreset
property holds for the whole $\REAL^d$ (rather than just $L$) and so we obtain for every $C\in \mathcal C$:
$$
\big|\sum_{i=1}^r w_i \cdot \dist^2(S_{i*}, C)+\Delta' -\dist^2(A^{(m)},C)\big|   \le \frac{\epsilon}{8} \cdot \dist^2(A^{(m)},C).
$$
By the triangle inequality,
\begin{equation}\label{blast}
\begin{split}
&\left| \left(\sum_{i=1}^r w_i \cdot \dist^2(S_{i*}, C)+\Delta'\right) -\left( \dist^2(A,C) -\norm{A-A^{(m)}}_F^2 \right)\right|\\
&\leq \left|\left(\sum_{i=1}^r w_i \cdot \dist^2(S_{i*}, C)+\Delta'\right) -\dist^2(A^{(m)},C)\right|\\
&\quad+\left|  \dist^2(A^{(m)},C) -\left( \dist^2(A,C) -\norm{A-A^{(m)}}_F^2 \right) \right|\\
&\le \frac{\eps}{2} \cdot \dist^2(A,C) + \frac{\eps}{8} \cdot \dist^2(A^{(m)},C).
\end{split}
\end{equation}
Using $\eps=1$ in Corollary~\ref{lem:movement:simple} we obtain
\[
|\dist^2(A,C)-\dist^2(A^{(m)},C)|\leq \dist^2(A,C) + 2 \cdot \norm{A-A_m}^2,
\]
so
\begin{equation}\label{alast}
\dist^2(A^{(m)},C)\leq 2 \cdot \dist^2(A,C)+2 \cdot \norm{A-A^{(m)}}^2\leq 4 \cdot \dist^2(A,C),
\end{equation}
where the last inequality is since $C$ is contained in a $j$-subspace and $j\geq m$.
Plugging~\eqref{alast} in~\eqref{blast} yields
 \[
\left|  \left(\sum_{i=1}^r w_i \cdot \dist^2(S_{i*}, C) +\Delta'+ \norm{A-A^{(m)}} \right)- \dist^2(A,C) \right|
\le \eps \cdot \dist^2(A,C).
 \]
 \end{proof}


\subsection{The Sensitivity Framework \label{senssec}}


Before turning to specific results for clustering problems, we describe a framework introduced by Feldman and Langberg \cite{FL11}
that allows to compute coresets for certain optimization problems (that minimize sums of cost of input objects) that also include the
clustering problems considered in this paper. The framework is based on a non-uniform sampling technique. We sample
points with different probabilities in such a way that points that have a high influence on the optimization problem are sampled with higher
probability to make sure that the sample
contains the important points. At the same time, in order to keep the sample unbiased, the sample points are weighted reciprocal to their
sampling probability. In order to analyze the quality of this sampling process Feldman and Langberg \cite{FL11} establish a reduction to
$(\eta ,\eps)$-approximations of a certain range space.

The first related sampling approach in the area of coresets for clustering problems was by Chen \cite{C09} who partitions the input point set in a way
that sampling from each set uniformly results in a coreset. The partitioning is based on a constant bicriteria approximation (the idea to use bicriteria approximations as a basis for coreset constructions goes back to Har-Peled  and Mazumdar~\cite{HM04}, but their work did not involve sampling), i.e.,
we are computing a solution with $O(k)$ (instead of $k$) centers, whose cost is at most a constant times the cost of the best solution with $k$ centers.
In Chen's construction, every point is assigned to its
closest center in the bicriteria approximation. Uniform sampling is then applied to each subset of points. Since the points in the same subset have a
similar distance to their closest center, the sampling error can be charged to this contribution of the points and this is
sufficient to obtain coresets of small size.

A different way is to directly base the sampling probabilities on the distances to the centers from the bicriteria approximation.
This idea is used by Arthur and Vassilvitskii~\cite{arthur2007k} for computing an approximation for the $k$-means problem, and it is used for the
construction of (weak) coresets by Feldman, Monemizadeh and Sohler~\cite{FMS07}. The latter construction uses a set of centers that provides
an approximative solution and distinguishes between points that are close to a center and points that are further away from their closest center.
Uniform sampling is used for the close points. For the other points, the probability is based on the cost of the points. In order to keep the sample
unbiased the sample points are weighted with $1/p$ where $p$ is the sampling probability.

Instead of sampling from two distributions, Langberg and Schulman~\cite{LS10} and Feldman and Langberg~\cite{FL11} define a single distribution,
which is a mixture of the two distributions in \cite{FMS07}. For the analysis they define the notion of \emph{sensitivity} of
points which is an even more direct way of measuring the importance of a point. Their work is not restricted to the $k$-means problem but works for a
large class of optimization problems. We review their technique in the following. The shape fitting framework is to describe problems of the following form:
We are given a set of input objects $F$ and a set of candidate shapes $\ucc{Q}$. Each input object is described by a function $f: \ucc{Q} \rightarrow \RR^{\ge 0}$
that encodes how well each
candidate shape fits the input object (the smaller the value the better the fit). Let $F$ be the set of functions corresponding to the input objects.
Then the shape fitting problem can be posted as minimizing $\sum_{f\in F} f(Q)$ over all $Q \in \ucc{Q}$.
As an example, let us consider the linear $j$-subspace approximation problem for a $d$-dimensional point set that is represented by the rows
of a matrix $A \in \RR^{n \times d}$. In this example, the set $\ucc{Q}$ is the set of all linear $j$-dimensional
linear subspaces. For each input point $A_{i*}$, we define a function $f_{A_{i\ast}} : \ucc{Q} \to \RR^{\ge 0}$ by setting
 $f_{A_{i\ast}} (C)= \dist^2(A_{i*}, C)$ for all $j$-dimensional linear subspaces $C$.
This way, the problem can be described as a shape fitting problem. More generally, for the affine $j$-subspace $k$-clustering problem
a shape $Q\in \ucc{Q}$ is the union of all sets of $k$ affine subspaces of dimenson $j$.

The sensitivity of a function is now defined as the maximum share that it can contribute to the sum of the function values for any given shape. The total sensitivity of the input objects with respect to the shape fitting problem is the sum of the sensitivities over all $f \in F$. We remark that the functions will be weighted later on. However, a weight will simply encode a multiplicity of a point and so we will first present the framework for unweighted sets.

\begin{definition}[Sensitivity,~\cite{LS10,FL11}]
Let $F$ be a finite set of functions, where each function $f\in F$ maps every item in $\ucc{Q}$ to a non-negative number in $\RR^{\ge 0}$. The \emph{sensitivity} $\sigma(f)$ of $f$ is defined as
\[
\sigma(f) := \sup \frac{f(Q)}{\sum_{h \in F} h(Q)},
\]
where the $\sup$ is over all $Q \in \ucc{Q}$ with $\sum_{h \in F} h(Q)>0$ (if the set is empty we define $ \sigma(f):= 0$).
The total sensitivity of $F$ is $\mathfrak{S}(F) := \sum_{f \in F} \sigma(f)$.
\end{definition}
We remark that a function with sensitivity $0$ does not contribute to any solution of the problem and can be removed from the input.
Thus, in the following we will assume that no such functions exist.

Notice that sensitivity is a measure of the influence of a function (describing an input object) with respect to the cost function of the shape fitting
optimization problem. If a point has a low sensitivity, then there is no set of shapes to which cost the object contributes significantly. In contract,
if a function has high sensitivity then the object is important for the shape fitting problem.
For example, if in the $k$-means clustering problem there is one point that is much further away from the cluster centers than all other points then it
contributes significantly to the cost function and we will most likely not be able to approximate the cost if we do not sample this point.

How can we exploit sensitivity in the context of random sampling? The most simple sampling approach (that does not exploit sensitivity)
is to sample a function $f^*$ uniformly at random and assign a weight $n$ to the point (where $|F|=n)$.
For each fixed $Q\in \ucc{Q}$ this gives an unbiased estimator, i.e., the expected value of $n \cdot f^*(Q)$ is $\sum_{f\in F} f(Q)$.
Similarly, if we would like to sample $s$ points we can assign a weight $n/s$ to any of them to obtain an unbiased estimator.
The problem with uniform sampling is that it may miss points that are of high influence to the cost of the shape fitting problem
(for example, a point far away from the rest in a clustering problem). This also leads to a high variance of uniform sampling.

The definition of sensitivity allows us to reduce the variance by defining the probabilities based on the sensitivity. The basic idea is very simple: If a function
contributes significantly to the cost of some shape, then we need to sample it with higher probability. This is where the sensitivity
comes into play. Since the sensitivity measures the maximum influence a function $f$ has on any shape, we can sample $f$
with probability $\sigma(f) / \mathfrak{S}(F)$. This way we make sure that we sample points that have a strong impact on the cost function
for some $Q \in \ucc{Q}$ with higher probability. In order to ensure that the sample remains unbiased, we rescale a function $f$
that is sampled with probability $\sigma(f) / \mathfrak{S}(F)$ with a scalar  $\mathfrak{S}(F)/\sigma(f)$ and call the rescaled function $f'$
and let $F'$ be the set of rescaled functions from $F$.
This way, we have for every fixed $Q \in \ucc{Q}$ that the expected contribution of $f'$ is $\sum_{f\in F} \frac{\sigma(f)}{\mathfrak{S}(F)}\cdot
\frac{\mathfrak{S}(F)}{\sigma(f)} \cdot f(Q) = \sum_{f\in F} f(Q)$, i.e., $f'$ is an unbiased estimator for the cost of $Q$.
The rescaling of the functions has the effect that the ratio between the maximum contribution a function has on a shape and
the average contribution can be bounded in terms of the total sensitivity, i.e., if the total sensitivity is small then all functions
contribute roughly the same to any shape. This will also result in a reduced variance.

Now the main contribution of the work of Feldman and Langberg \cite{FL11} is to establish a connection to the theory of range
spaces and VC-dimension. In order to understand this connection we rephrase the non-uniform sampling process as described above
by a uniform sampling process. We remark that this uniform sampling process is only used for the analysis of the algorithm and must not
be carried out by the sampling algorithm.
The reduction is as follows. For some (large) value $n^*$, we replace each rescaled function $f'\in F'$ by $n^* \cdot \sigma(f)$
copies of $f'$ (for the exposition at this place let us assume that $n^* \cdot \sigma(f)$ is integral).
This will result in a new set $F_{\text{new}}$ of $n^* \cdot \mathfrak{S}(F)$ functions.
We observe that sampling uniformly from $F_{\text{new}}$ is equivalent to sampling a function $f\in F$ with probability
$\sigma(f) / \mathfrak{S}(F)$ and rescaling it by $\mathfrak{S}(F)/\sigma(f)$.
Thus, this is again an unbiased estimator for $F$ (i.e., $\sum_{f'\in F_{\text{new}}} \frac{1}{|F_{\text{new}}|} f' = \sum_{f\in F} f(Q)$ holds.).
Also notice that
$\frac{1}{n^* \cdot \mathfrak{S}(F)} \cdot \sum_{f'\in F_{\text{new}}} f'(Q) = \sum_{f\in F} f(Q)$, which means that relative error bounds for $\sum_{f'\in F_{\text{new}}} f'(Q)$ carry over to error bounds for $\sum_{f\in F} f(Q)$.

We further observe that for any fixed $Q\in \ucc{Q}$ and any function $f' \in F_{\text{new}}$ that corresponds to $f \in F$
we have that $\frac{f'(Q)}{\sum_{g'\in F_{\text{new}}} g'(Q)}
\le \sigma(f) \cdot \frac{1}{n^* \cdot \mathfrak{S}(F)} \cdot \frac{\mathfrak{S}(F)}{\sigma(f)} = \frac{1}{n^*}$. Furthermore, the average value of
$\frac{f'(Q)}{\sum_{g'\in F_{\text{new}}}g'(Q)}$ is $\frac{1}{n^* \cdot \mathfrak{S}(F)}$. Thus, the maximum contribution of an $f'$ only slightly deviates
from its average contribution.


Now we can discretize the distance from any $Q$ to the input points into ranges according to their relative
distance from $Q$. If we know the number of points inside
these ranges approximately, then we also know an approximation of $\sum_{f\in F} f(Q)$.

In order to analyze this, Feldman and Langberg \cite{FL11} establish a connection to the theory of range spaces and the Vapnik-Chervonenkis
dimension (VC dimension). In our exposition we will mostly follow a more recent work by Braverman et al. \cite{braverman2016new} that obtains
stronger bounds.

\begin{definition}
\label{dimfunctions}
Let $F$ be a finite set of functions from a set $\ucc{Q}$ to $\image$.
For every $Q\in \ucc{Q}$ and $r \geq 0$, let
$$
\range(F,Q,r)=\{f\in F \mid  f(Q) \ge r \}.
$$
Let
$$
\ranges(F)= \{\range(F, Q, r)\mid Q\in \ucc{Q}, r \geq 0\}.
$$
Finally, let $\mathfrak R_{\ucc{Q},F}:= \big(F,\ranges(F)\big)$ be the \emph{range space} induced by $\ucc{Q}$ and $F$.
\end{definition}

In our analysis we will be interested in the VC-dimension of the range space $\mathfrak R_{\ucc{Q},F_{\text{new}}}$.
We recall that $F_{\text{new}}$ consists of (possibly multiply) copies of rescaled functions from the set $F$.
We further observe that multiple copies of a function do not affect the VC-dimension. Therefore, we will be interested
in the VC-dimension of the range space $\mathfrak R_{\ucc{Q},F^*}$ where $F^*$ is obtained from $F$ by rescaling each
function in $F$ by a non-negative scalar.

Finally, we remark that the sensitivity of a function is typically unknown. Therefore, the idea is to show that it suffices to
be able to compute an \emph{upper bound} on the sensitivity. Such an upper bound can be obtained in different ways. For example,
for the $k$-means clustering problem, such bounds can be obtained from a constant (bi-criteria) approximation.

In what follows we will prove a variant of a Theorem from \cite{braverman2016new}.
The difference is that in our version we guarantee that the weight of a coreset point is at least its weight in the input set, which will be useful in the context of streaming when the sensitivity is a function of the number of input points.
The bound on the weight follows by including all points of very high sensitivity approximation value directly into the coreset.

Observe that in the context of the affine $j$-subspace $k$-clustering problem, the sum of the weights of a coreset for an unweighted $n$ point set cannot exceed $(1+\eps) n$ (since we can put the centers to infinity).\footnote{If for a different problem it is not possible to directly obtain an upper bound on the weights (for example, in the case of linear subspaces), one can add an artificial set of centers that enforces the bound on the weights in a similar way as in the affine case. However, we will not need this argument when we apply Theorem~\ref{sensitivitysamplingscheme}.}
Thus, when we apply Theorem~\ref{sensitivitysamplingscheme} later on, we know that the weight of each point in the coreset is at least its weight in the input set, and that the total weight is not very large.

\begin{theorem}[Variant of a Theorem in \cite{braverman2016new}]\label{sensitivitysamplingscheme}
Let $F$ be a finite weighted set of functions from a set $\ucc{Q}$ to $[0,\infty)$, with weights $w_f > 0$ for every $f\in F$, and let $\delta,\eps\in(0,1/2)$.
Let $\tilde{\sigma}(f)\geq \max\{\frac{1}{|F|},\sigma(f)\}$ for every $f\in F$, and $\mathfrak{\tilde{S}}(F)=\sum_{f\in F}\tilde{\sigma}(f)$.
Given $\tilde{\sigma}$, one can compute in time $O(|F|)$ a set $S \subset F$ of
$$
O\left(\frac{\mathfrak{\tilde{S}}(F)}{\eps^2} \cdot \left(d \log \mathfrak{\tilde{S}}(F)+\log \frac{1}{\delta}\right)\right)
$$
 weighted functions such that with probability $1-\delta$ we have for all $Q \in \ucc Q$ simultaneously
\[
(1-\eps)\sum_{f\in F}w_f \cdot f(Q)\leq \sum_{f\in S}u_f \cdot f(Q)\leq (1+\eps)\sum_{f\in F}w_f \cdot f(Q),
\]
where $u_f \ge w_f$ denotes the weight of a function $f \in S$,
and where $d$ is an upper bound on the VC-dimension of every range space $\mathfrak R_{\ucc{Q},F^*}$ induced by $F^*$ and $Q$
that can be obtained by defining $F^*$ to be the set of functions from $F$ where each function is scaled by a separate non-negative scalar.
\end{theorem}

\begin{proof}
Our analysis follows the outline sketched in the previous paragraphs, but will be extended to non-negatively weighted sets of functions.
The point weights will be interpreted as multiplicities. If each function $f\in F$ has a weight $w_f > 0$, the definition of sensitivity
becomes
\[
\sigma(f) := \sup \frac{w_f \cdot f(Q)}{\sum_{h \in F} w_{h} \cdot h(Q)}.
\]
Since the sensitivities may be hard to compute we will sample according to a function $\tilde {\sigma}$ that provides an upper bound on the sensitivity and we will use $\tilde{\mathfrak{S}}(F) = \sum_{f\in F} \tilde{\sigma}(f)$, i.e., our plan is to sample a function $f$ with
probability $\tilde \sigma(f) / \tilde{\mathfrak{S}}(F) $ and weight the sampled function
 with $w_f \cdot \tilde{\mathfrak{S}}(F) / \tilde \sigma(f)$. More precisely, since we want to sample $s$ functions
 i.i.d. from our probability distribution for an $s$ defined below, the weight will become $\frac{1}{s} \cdot w_f \cdot \tilde{\mathfrak{S}}(F) / \tilde \sigma(f)$ to keep the sample unbiased. We then analyze the quality by using the reduction to uniform sampling and applying Theorem \ref{approximation} to get the desired approximation.

In the following we will describe this analysis in more detail and in this process deal with a technicality that arises when we would like to ensure that the weight of a coreset point is at least its input weight.
 Namely, if we would like to sample $s$ functions i.i.d. according to our sampling distribution and there exists a
function in the input set with $\tilde \sigma(f) / \tilde{\mathfrak{S}}(F) > 1/s$ then $f$ will receive a weight $\frac{1}{s} \cdot w_f
\cdot \tilde{\mathfrak{S}}(F) / \tilde \sigma(f) < w_f$, which is a case that we would like to avoid (for example, in the streaming case this may have
sometimes undesirable consequences).\footnote{Observe that in this case the expected number of copies of $f$ in the sample is bigger than $1$, so they could typically be combined to a single point with weight at least $1$. However, there is also some probability that this is not possible, which is why we deal with the functions that satisfy $\tilde \sigma(f) / \tilde{\mathfrak{S}}(F) > 1/s$ explicitly.}
In order to deal with this issue, we simply remove all functions with $\sigma(f) / \mathfrak{S}(F) > 1/s$ and put a copy of this weighted function into the coreset.
We then sample from the remaining functions.
We only need to take care of the fact that removing functions from the input set also affects the total sensitivity.

Let us start with a detailed description. Let $s \ge \frac{c}{4\eta\eps^2} \cdot \big(d \log \frac{1}{\eta}+ \log 2/\delta \big)$ where the constant
$c$ is from Theorem \ref{approximation}, $\eta = 1/\tilde{\mathfrak{S}}(F) $ and $d$ is as in the description of the theorem.
We define $S_1 = \{f\in F \mid \tilde \sigma(f)/\tilde{\mathfrak{S}}(F) > 1/s\}$.  Clearly, we have $|S_1|< s$. The functions in $S_1$ are put
in the final coreset using their original weights. Let us define $F_1 = F \setminus S_1$ to be the set of remaining functions; it remains to show
that we can approximate the cost of $F_1$.
For all functions in $F_1$ we define $\tilde \sigma_1(f)$ such that $\tilde \sigma_1(f) \ge \tilde \sigma(f)$, $\tilde \sigma_1(f) \le \tilde{\mathfrak{S}}(F)/s$ and
$\tilde{\mathfrak{S}_1}(F_1) = \tilde {\mathfrak{S}}(F)$, where $\tilde{\mathfrak{S}_1}(F_1) = \sum_{f\ \in F_1} \tilde \sigma_1(f)$. This
ensures that in the set $F_1$ there are no functions with $\tilde \sigma_1(f) / \tilde{\mathfrak{S}_1}(F_1) > 1/s$ and so each sampled function $f$
will receive a weight of at least $w_f$. We remark that the choice of
$\tilde \sigma_1$ does not necessarily satisfy the sensitivity definition for the set $F_1$.
However, we have $\tilde \sigma_1(f) \ge \frac{w_f f(Q)}{\sum_{h\in F} w_h h(Q)}$ for all $Q \in \ucc Q$ and $f \in F_1$.

For the remainder of the analysis, it will be convenient to move from weighted functions to unweighted functions. This can be easily done by
replacing each weighted function $f$ with weight $w_f$ by a function $g$ with $g(Q) = w_f \cdot f(Q)$ for all $Q$ to obtain a set $F_2$.
Note that this does not affect the sensitivity and so we can define $\tilde \sigma_2(g) = \tilde \sigma_1(f)$ where $g \in F_2$ is the unweighted
function obtained from $f\in F_1$.
This implies that
$\tilde \sigma_2(g) \ge \frac{g(Q)}{\sum_{h\in F} w_h h(Q)}$ for all $Q \in \ucc Q$ and $g \in F_2$.
We then define $\tilde{\mathfrak{S}_2}(F_2) = \sum_{g\ \in F_2} \tilde \sigma_2(g)$.

For the sake of analysis we will now apply our reduction to uniform sampling to the set $F_2$.
We replace every function $g \in F_2$ by $\lceil n^* \tilde \sigma_2(g) \rceil$ copies and we rescale each function $g\in F_2$
by $\frac{1}{\lceil n^* \tilde \sigma_2(g) \rceil}$ and call the resulting set of functions $F_{\text{new}}$.
We observe that this scaling is different from what has been discussed in the previous paragraphs. This new scaling will make some technical arguments
in proof a bit simpler and we rescale the sampled functions in the end one more time.
 We can make $n^*$ arbitrarily large, which makes the error induced by the rounding arbitrarily small. We assume that $n^*$
is large enough that the probability that the reduction behaves different from the original sampling process is at most $\delta/2$. In order to keep
the presentation simple, we will assume in the following that all $n^* \tilde \sigma_2(g)$ are integral and so we can argue as explained before.
We observe that the VC-dimension of the range space $\mathfrak R_{\ucc{Q},F_{\text{new}}}$ is at most $d$.
We also observe that $\sum_{f\in F_{\text{new}}} f(Q) = \sum_{g\in F_2} g(Q)$.
We recall that sampling uniformly from $F_{\text{new}}$ is equivalent to sampling a scaled copy of $g \in F_2$ with probability
$\tilde \sigma_2{g} / \tilde{\mathfrak{S}_2}$.

It now follows from Theorem \ref{approximation} that an i.i.d. sample of $s$ functions from the uniform distribution over $F_{\text{new}}$ is an $(\eta,\epsilon/2)$-approximation for the range space $\mathfrak R_{\ucc{Q},F_{\text{new}}}$ with probability at least $1-\delta/2$. We call this sample $S$. In the following, we show that $S$ (suitably scaled) together with $S_1$ is a coreset for $F$. In order to do so, we show that $S$ approximates the cost of every $Q \in \ucc Q$
for $F_2$ (and so for $F_1$).
For this purpose let us fix an arbitrary $Q \in \ucc{Q}$ and let us assume that $S$ is indeed an $(\eta,\epsilon/2)$-approximation.
We would like to estimate $\sum_{g \in F_2} g(Q) = \sum_{f \in F_{\text new}} f(Q)$ upto small error.
First we observe that
\begin{align*}
\sum_{f \in F_{\text new}} f(Q) = \sum_{f\in F_{\text new}} \int_{r=0}^{\infty} \mathbf{1}(f(Q) \ge r) \text{d} r
&= \int_{r=0}^{\infty} \sum_{f\in F_{\text new}} \mathbf{1}(f(Q) \ge r) \text{d} r\\
&= \int_{r=0}^{\infty} |\range(F_{\text new},Q,r)| \text{d} r
\end{align*}
where $\mathbf {1}(f(Q)\ge r)$ is the indicator function of the event $f(Q)\ge r$.
If there are more than $\eta \cdot |F_{\text{new}}|$ functions with $f(Q) \ge r$
then our approximation provides a relative error.
Let $I_1 \subset \RR^{\ge 0}$ be the set of all $r\ge 0$ with $\range(F_{\text new},Q,r) \ge \eta \cdot |F_{\text{new}}|$.
Then we know that $\big|\frac{|F_{\text new}|}{|S|}|\range(S,Q,r)|-|\range(F_{\text new},Q,r)|\big| \le \frac{\epsilon}{2}\cdot|\range(F_{\text new},Q,r)|$ for all $r \in I_1$.

Let $I_2 = \RR^{\ge 0} \backslash I_1$ contain the values for $r$ for which $\range(F_{\text new},Q,r) < \eta \cdot |F_{\text{new}}|$.
For these, we obtain an additive error of $\frac{\epsilon \cdot \eta}{2}  |F_{\text{new}}|$.
Let $r_{\max}$ be the maximum value of $f(Q)$ for any $f \in F_{\text{new}}$. For any $r \ge r_{\max}$,
$\range(F_{\text new},Q,r)$ will contain all functions of $F_{\text{new}}$ and so $r\ge r_{\max}$ implies
$r \notin I_2$. This implies
\begin{align*}
\int_{r \in I_2} \Big|\frac{|F_{\text new}|}{|S|}  |\range(S,Q,r)|-|\range(F_{\text new},Q,r)|\Big|
 &\le \int_{r=0}^{r_{\max}} \frac{\epsilon\cdot \eta}{2} \cdot |F_{\text{new}}|\\
 &= r_{\max} \cdot \frac{\epsilon\cdot \eta}{2} \cdot |F_{\text{new}}|.
\end{align*}
In order to charge this error, consider $f \in F_{\text{new}}$ with corresponding $g \in F_2$. We know that
$$
\frac{g(Q)}{{\sum_{h\in F}}w_h \cdot h(Q)} \le \tilde \sigma_2(g)
$$
and that $g$ was replaced by $n^* \tilde \sigma_2(g)$ copies $f$ in $F_{\text{new}}$ scaled by $1/(n^* \tilde \sigma_2(g))$.
Hence,
$$
f(Q) =  \frac{1}{n^* \cdot \tilde \sigma_2(g)} \cdot g(Q)
$$
and so
$$
\frac{f(Q)}{{\sum_{h\in F}}w_h \cdot h(Q)} \le \frac{1}{n^*}.
$$

%
%
%
This implies $r_{\max} \le \frac{1}{n^*} \sum_{h\in F_{\text new}}w_h h(Q)$.
Combining both facts and the choice of $\eta$, we obtain that the error is bounded by
\begin{align*}
r_{\max} \cdot \frac{\epsilon\cdot \eta}{2} |F_{\text{new}}|
\le&
\frac{1}{n^*} \sum_{h\in F} w_q h(Q) \cdot \frac{\epsilon\cdot \eta}{2} |F_{\text{new}}|\\
\le & \frac{1}{n^*} \sum_{h\in F_{\text new}} w_q h(Q) \cdot \frac{\epsilon\cdot \eta}{2} n^* {\tilde{\mathfrak{S}_2}}(F_2)
= \frac{\epsilon}{2}\sum_{h \in F} w_h h(Q),
\end{align*}
where we use that $|F_{\text{new}}| = n^* \cdot \tilde{\mathfrak{S}_2}(F_2)$ and
 $\tilde{\mathfrak{S}_2}(F_2) = \tilde{\mathfrak{S}_1}(F_1) = \tilde{\mathfrak{S}}(F) = 1/\eta$.
Combining the two error bounds and using $\sum_{f\in F_{\text{new}}} f(Q) \le \sum{h\in F} w_h h(Q)$ we obtain
\begin{align*}
 &\int_{r=0}^{\infty} \Big| \frac{|F_{\text new}|}{|S|}|\range(S,Q,r)|-|\range(F_{\text new},Q,r)| \Big| \\
\le & \int_{r \in I_1} \frac{\epsilon}{2}\cdot|\range(F_{\text new},Q,r)|
 + \frac{\epsilon}{2}\sum_{h \in F} w_h h(Q)
\le \epsilon \cdot \sum_{h \in F} w_h f(Q)\\
\end{align*}
This implies
$$
\Big| \frac{|F_{\text new}|}{|S|} \cdot \sum_{f \in S} f(Q) - \sum_{f \in F_{\text new}} f(Q)\Big| \le \epsilon \cdot \sum_{h \in F} w_h h(Q)\\
$$
Thus, when we rescale the functions inf $S$ by $\frac{|F_{\text new}|}{|S|}$ to obtain a new set of function $S'$
we obtain
\begin{align*}
\Big| \sum_{g \in S'} g(Q) - \sum_{g \in F_{2}} g(Q)\Big| \le \epsilon \cdot \sum_{h \in F} w_h h(Q).
\end{align*}
We observe that the functions in $S'$ correspond to functions $g$ in $F_2$ rescaled by $\tilde{\mathfrak{S}_2} / \tilde \sigma_2(g)$, which in turn
corresponds to function $f\in F$ with weight $w_f \cdot \tilde {\mathfrak{S}_2} / \tilde \sigma_2(g)$.
It follows that $S'\cup S_1$ is a coreset for $F$.
Since the non-uniform sample as well as all preprocessing steps can be implemented in $O(|F|)$ time \cite{vose1991linear}, the theorem follows.
\end{proof}

\subsection{Bounds on the VC dimension of clustering problems}
\newcommand{\sgn}{\mathrm{sgn}}

In this section we show how to obtain a bound on the VC-dimension on a range space as in the previous Theorem in the case of the affine $j$-subspace $k$-clustering problem. In order to bound this, we use a method due to Warren \cite{warren1968lower}.
We consider a weighted set of $n$ points and for every set $Q$ of $k$ affine $j$-dimensional subspaces.
Then we consider the range defined by the subset of input points whose weighted squared distance to $Q$ is at least $r$.
We show that the VC-dimension of this range space is $O(djk \log k)$.
We remark that in some previous papers a bound of $O(djk)$ has been claimed for a related range space, but we could not fully reproduce the proofs.
In what follows, $\sgn(x)$ denotes the sign of $x \in \REAL$. More precisely, $\sgn(x)=1$ if $x>0$, $\sgn(x)=-1$ if $x<0$, and $\sgn(x)=0$ otherwise.
We will use the following theorem.

\begin{theorem}[Theorem 3 in~\cite{warren1968lower}\label{rrr}]
Let $f_1,\cdots,f_m$ be real polynomials in $d^* \le m$ variables, each of degree at most $\ell \ge 1$. Then
the number of sign sequences $(\sgn f_1(x),\ldots,\sgn f_m(x))$,  $x\in\REAL^{d^*}$, that consist of $1$ and $-1$ is at most
$(4e\ell m/d^*)^{d^*}$.
\end{theorem}

\begin{corollary}[Corollary 3.1 in~\cite{warren1968lower}]
\label{rrr2}
If $\ell \ge 2$ and $m \ge 8{d^*} \log \ell$, then the number of distinct sequences as in the above theorem is less than $2^{m}$.
\end{corollary}

We use these results to obtain.

\begin{corollary}
\label{vc-dim}
Let $d,j,k$ be positive integers such that $j\leq d-1$. Let $\ucc{Q}_{jk}$ be the family of all sets which are the union of $k$ affine subspaces of
$\mathbb{R}^d$, each of dimension $j$. Let $P=\br{p_1,\cdots,p_n}$ be a set of $n$ points in $\REAL^d$ with weights $w:P\to[0,\infty)$. Let $F^* =\br{f_1,\cdots,f_n}$ where $f_i(Q)=w(p_i) \cdot \dist^2(p_i,Q)$ for every $i\in[n]$, $Q\in  \ucc {Q}_{jk}$. Then the dimension of the range space $\mathfrak R_{\ucc{Q}_{jk},F^*}$
that is induced by $\ucc{Q}_{jk}$ and $F^*$ is $O(jdk \log k)$.
\end{corollary}
\begin{proof}
We first show that in the case $k=1$ the VC-dimension of the range space $\mathfrak R_{\ucc{Q},F^*}$ is $O(jdk)$. Then the result follows from the fact
that the $k$-fold intersection of range spaces of VC-dimension $O(jdk)$ has VC-dimension $O(jdk \log k)$ ~\cite{BEHW89,eisenstat2007vc}.

If $n<d$ then the result is immediate. Thus, we consider the case $n \ge d$.
We will first argue that the weighted distance to a subspace can be written as a polynomial in $O(jd)$ variables.
Let $Q$ be an arbitrary $j$-dimensional affine subspace. By the Pythagorean Theorem we can write $\dist^2(p_i,Q) =
\norm{p_i}^2 - \norm{X_Q p_i - b_Q}^2$ where $X_Q\in\REAL^{j\times d}$ with $X_Q^TX_Q=I$ and $b_Q\in\REAL^{j}$. Therefore,
$f_i(Q)-r$ is a polynomial of constant degree $\ell$ with $d^* \in O(jd)$ variables.

Consider a subset $G \subset F^*$ with $|G|=m$, denote the functions in $G$ by $f_1,\ldots,f_m$. Our next step will be to give an upper bound on the number of different ranges in our range space $\mathfrak R_{\ucc{Q}_{jk},F^*}$ for $k=1$ that intersect with $G$.
Recall that the ranges are defined as
$$
\{p \in P \mid w(p) \cdot \dist(p,Q) \ge r\}
$$
for $Q \in \ucc{Q}_{jk}$ and $r \ge 0$.
We observe that $w(p_i) \cdot \dist^2(p_i,Q)\geq r$, iff $\sgn(f_i(Q) -r) \ge 0$.
Thus, the number of ranges is at most
\[
\big|\br{(\sgn (f_1(Q)-r),\ldots,\sgn (f_m(Q)-r) )\mid x\in\REAL^\ell}\big|.
\]
We also observe that for every sign sequence that has zeros, there is a sign sequence corresponding to the same range that only
contains $1$ and $-1$ (this can be obtain by infinitesimally changing $r$). Thus, by Theorem \ref{rrr} the number of such sequences
is bounded by $(4e\ell m/d^*)^{d^*}$, where $\ell=O(1)$. By Corollary \ref{rrr2} we know that for $\ell \ge 2$ (which we can always assume as $\ell$ is
an upper bound for the degree of the involved polynomials) and $m \ge 8 d^* \log \ell$ the number of such ranges is less than $2^m$.
At the same time, a range space with VC-dimension $d$ must contain a subset $G$ of size $d$ such that any subset of $G$ can be written as $G\cap \range$ for some $\range\in \ranges$, which implies that the number of such sets is $2^d$. Since this is not possible for $G$ if $m \ge 8 d^* \log \ell$, we know that the VC dimension of our range space is bounded by $8 d^* \log \ell \in O(jd)$ (for the case $k=1$).
 Now the he result follows by observing that in the case of $k$ centers every range is obtained by taking the intersection of $k$ ranges of the range space for $k=1$.
\end{proof}


\subsection{New Coreset for \mathk-Means Clustering}


We now apply the results from the previous section to the $k$-means problem. For this problem,
it is known how to compute \lq weak\rq\ coresets of size independent of $n$ and $d$. A weak coreset is a weighted point set that approximates the cost of the objective function for some, but not all possible solutions (not for all sets of $k$ centers).
For these results, see~\cite{FMS07,braverman2016new,FL11} and references therein.
However, in this paper we focus on coresets that approximate the cost for \emph{every} set of $k$ centers in $\REAL^d$ as stated in Definition~\ref{coresetdef}.
We will use the following result that follows from the sensitivity framework as presented in the previous section
and as suggested in \cite{braverman2016new}. We also remark that a slightly better version can be obtained (see \cite{braverman2016new}).
In order to keep this presentation self-contained, we will use the following theorem.

\begin{theorem}[Coreset for $k$-means]
\label{k-means-coreset}
Let $A\in\REAL^{n\times d}$, $k\geq1$ be an integer, and $\eps,\delta\in(0,1)$.
Let $\CC=\br{C\subset \REAL^d\mid |C|=k}$ be the family of all sets of $k$ centers
 in $\REAL^d$.
Then, with probability at least $1-\delta$, an $\eps$-coreset $(S,0,w)$ for the $\CC$-clustering problem of $A$ of size
\[
|S|= O\left(\frac{k^2 \log k}{\eps^2} (d\log(k)+\log(1/\delta))\right)
\]
can be computed in time $O(ndk\log(1/\delta))$.
\end{theorem}
\begin{proof}
We will apply the sensitivity framework. We define $Q=\CC$ to be the family of sets of $k$ centers in $\REAL^d$
and $F$ to be a set that has one function for each input point (row of $A$) and define $f(Q)$ to be the distance
of this point to the nearest center in $Q$. According to Corollary \ref{vc-dim} the previous section, the VC-dimension of any range space
$\mathfrak R_{\ucc{Q},F^*}$ is $O(dk\log k)$, where $F^*$ is obtained from $F$ by rescaling the functions in $F$.

Our next step is to observe that given an  $(\alpha,\beta)$-approximation $C'$ ($C'$ is a set of $\beta k$ centers such that
$\dist^2(A,C')\leq \alpha\min_{C\in \CC}\dist^2(A,C)$) the sensitivity of a point $A_{i*}$ in a cluster $J$ is $O(\frac{1}{|J|}
+ \alpha \cdot \dist^2(A_{i*},C')/ \dist^2(J,C'))$ implying a total sensitivity of $O(k)$ (assuming $\alpha$ to be constant).
This can be seen by a case distinction: Either there exists a center within squared distance $O(\dist^2(A_{i*},C'))$,
then the second term is an upper bound. Or there exists a center within squared distance
$O(\alpha \cdot \frac{\dist^2(J,C')}{|J|})$ in which case the first term gives a bound. Or neither of the two is true. In
this case, there is no center within distance $O(\cost(A_{i*},C') + \alpha \cdot \frac{\dist^2(J,C')}{|J|})$. In this case,
by Markov's inequality, at least $|J|/2$ points have distance $\Omega(\dist^2(A_{i*},C'))$ and so the first term is a bound. Now the bound on the coreset size also arises from plugging the bound on the total sensitivity into the upper bound on the sample size in Theorem \ref{sensitivitysamplingscheme}.

A constant $(\alpha,\beta)$-approximation can be computed in $O(ndk \log (1/\delta))$ time with probability at least $1-\delta$
\cite{ADK09}. From this we can compute the upper bounds on the sensitivites and so the result follows from Theorem  \ref{sensitivitysamplingscheme}.
\end{proof}

The following theorem reduces the size of the coreset to be independent of $d$. We remark that also here one can obtain
slightly stronger bounds that are a bit harder to read. We opted for the simpler version.

\begin{theorem}[Smaller Coreset for $k$-means\label{tthm}]
Let $A\in\REAL^{n\times d}$ whose rows are weighted with non-negative weights $w=(w_1,\dots,w_n)$.
Let $k\geq1$ be an integer, $\CC=\br{C\subseteq\REAL^d\mid |C|=k}$ denote the union over every set of $k$ centers in $\REAL^d$, $\eps,\delta\in(0,1)$.
Then an $\eps$-coreset $(S,\Delta,w)$ for the $\CC$ clustering of $A$ of size
\[
|S|\in O\left(\frac{k^3\log^2 k}{\eps^4} \log(1/\delta)\right)
\]
can be computed, with probability at least $1-\delta$, in $O(\min\{nd^2,n^2d\} + \frac{nk}{\epsilon^2} (d+k \log(1/\delta))$
time.
\end{theorem}
\begin{proof}
We would like to apply Theorem \ref{thm:dim} where we need to do minor modifications to deal with weighted points.
We first need to compute an optimal subspace in the weighted setting. We exploit that scaling each row by $\sqrt{w_i}$ and then computing
in $O(\min\{nd^2 , n^2 d\})$ time the singular value decomposition $U\Sigma V^T$ will result in a subspace that minimizes the squared distances
from the weighted points. Next we need to project $A$ on the subspace spanned by the first $m$ right singular vectors for $m=O(k/\eps^2)$, i.e., we compute
$A^* = AV^{(m)}(V^{(m)})^T$ in $O(ndm)$ time where $V^{(m)}$ is the matrix spanned by the first $m$ right singular vectors.
The correctness of this approach follows from dividing the weighted points into infinitesimally weighted points of equal weight.

By replacing $d$ with $m$ in Theorem~\ref{k-means-coreset}, an $(\epsilon/8)$-coreset $(S,0,w)$ of the desired size and probability of failure
can be computed for $A^*$. Plugging this coreset in Theorem~\ref{thm:dim} yields the desired coreset $(S,\Delta,w)$ in time $O(nk^2/\eps^2 \log(1/\delta))$
\end{proof}


\subsection{Improved Coreset for \mathk-Line-Means}


The result in the following theorem is a coreset which is a weighted subset of the input set.
Smaller coresets for $k$-line means whose weights are negative or depends on the queries, as well as weaker coresets, can be found in~\cite{feldman2010coresets,FL11} and may also be combined with the dimensionality reduction technique in our paper. 
\begin{theorem}[Coreset for $k$-line means~\cite{varadarajan}\label{linescoreset}]
Let $A\in\REAL^{n\times d}$, let $k\geq1$ be an integer, and let $\eps,\delta\in(0,1)$.
Let $$\CC_L=\br{C\subset \REAL^d\mid \text{$C$
 is the union of a set of $k$ lines in $\REAL^d$}}$$
be the family of all sets of $k$ lines in $\REAL^d$. Then, with probability at least $1-\delta$, an $\eps$-coreset $(S,\Delta,u)$ for the $\CC_L$-clustering problem of $A$ of size
\[
|S|\in \frac{k^{O(k)}\log^2(n)}{\eps^2} \left(d\log\log n+\log(1/\delta)\right),
\]
can be computed in time $T(d)=n\cdot (dk^k\log n\log(1/\delta)/\eps)^{O(1)}$.
\end{theorem}
\begin{proof}
The proof is based on bounding the sensitivity of each point and then using the sensitivity framework in Theorem~\ref{sensitivitysamplingscheme}, similar to its application for $k$-means in Theorem~\ref{k-means-coreset}. By Corollary~\ref{vc-dim}, $\mathfrak R_{\ucc{Q},F^*} \in O(mk\log k)$ for $F^* =\br{f_1,\cdots,f_n}$ where $f_i(Q)=\dist^2(A_{i^\ast},Q)$ for every $i\in[n]$, $Q\in  \ucc \CC_L$. (observe that $\CC_L=\CC_{1,k}$ in the notation of the corollary).

It is thus left to bound the sensitivity of each point and the total sensitivity. As explained in~\cite{varadarajan}, computing these bounds is based on two steps: firstly we compute an approximation to the optimal $k$-line mean, so we can use Theorem~\ref{sens_move} to bound the sensitivities of the projected sets of points on each line. Secondly, we bound the sensitivity independently for the projected points on each line, by observing that their distances to a query is the same as the distances to $k$ weighted centers. Sensitivities for such queries were bounded in~\cite{schulman} by $k^{O(k)}\log n$. We formalize this in the rest of the proof.

An $(\alpha,\beta)$-approximation for the $k$-line means problem with $\alpha=O(1)$ and $\beta=O(\log n)$ can be computed in time $O(T(d))$ with probability at least $1-\delta/10$, where $T(d)$ is defined in the theorem, using $O(\log(1/\delta))$ runs (amplification) of the algorithm in Theorem~10 in~\cite{FL11}.

Next, due to~\cite{schulman} and~\cite{varadarajan}, any $(\alpha,\beta)$-approximation $C$ for the $k$-line means problem can be used to compute upper bounds on the point sensitivities and then the sum of all point sensitivities is bounded by $O(\alpha)+\beta k^{O(k)} \log n=\beta k^{O(k)} \log^2n$ in additional $T(d)$ time as defined in the theorem.

By combining this bound on the total sensitivity with the bound on the VC dimension in Theorem~\ref{sensitivitysamplingscheme}, we obtain that it is possible to compute a set of size $|S|$ as desired.
\end{proof}

Notice that computing a constant factor approximation (or any finite multiplicative factor approximation that may be depend on $n$) to the $k$-line means problem is NP-hard as explained in the introduction, if $k$ is part of the input. No bicriteria approximation with $\beta \in O(1)$ that takes polynomial time in $k$ is known. This is why we get a squared dependence on $\log n$ in our coreset size. It is possible to compute a constant factor approximation (in time exponential in $k$) : Set the precision to a reasonable constant, say $\eps'=1/2$, and then use exhaustive search on the $\eps'$-coreset to obtain a solution with a constant approximation factor. The constant factor approximation can then be used to compute a coreset of smaller size. However, exhaustive search on the coreset still takes time $|S(1/2,0,\delta)|^k$, meaning that the running time would include $\log^k(n)$ and a term that is doubly exponential in $k$. We thus consider it preferable to use the coreset computation as stated in Theorem~\ref{linescoreset}. This is in contrast to the case of $k$-means where a constant factor approximation can be computed in time polynomial in $k$; see the proof of Theorem~\ref{tthm}.

Now we apply our dimensionality reduction to see that it is possible to compute coresets whose size is independent of $d$. The running time of the computation is also improved compared to Theorem~\ref{linescoreset}.
\begin{theorem}[Smaller Coreset for $k$-line means\label{tlthm-kline}]
Let $A\in\REAL^{n\times d}$, and let $k\geq1$ be an integer, and let $\eps,\delta\in(0,1)$.
Let $$\CC_L=\br{C\subset \REAL^d\mid \text{$C$ is the union of a set of $k$ lines in $\REAL^d$}}$$
be the family of all $k$-lines in $\REAL^d$. Then, with probability at least $1-\delta$, an
$\eps$-coreset $(S(\eps,\delta),\Delta,w) \subset A$ for the $\CC_L$-clustering problem of $A,u$ of size
\[
|S(\eps,\delta)|\in \frac{k^{O(k)}\log^2n}{\eps^2} \left(\log (k)\log\log (n)k^2/\eps^2+\log(1/\delta)\right),
\]
can be computed in time $O(nd^2)+n(k^k\log n\log(1/\delta)/\eps)^{O(1)}$.
\end{theorem}
\begin{proof}
Similarly to the proof of Theorem~\ref{tthm}, we compute $A^{(m)}$ in $O(nd^2)$ time where $m=O(k/\eps^2)$.
By replacing $d$ with $m$ in Theorem~\ref{linescoreset}, a coreset $(S,0,w)$ of the desired size and probability of failure can be computed for $A^{(m)}$. Plugging this coreset in Theorem~\ref{thm:dim} yields the desired coreset $(S,\Delta,w)$.
\end{proof}


\subsection{Computing Approximations Using Coresets}


A well-known application of coresets is to first reduce the size of the input and then to apply an approximation algorithm.
In Algorithm~\ref{kmeansall} below we demonstrate how Theorem~\ref{thm:dim} can be combined with existing coreset constructions and approximation algorithms
to improve the overall running time of clustering problems by applying them on a lower dimensional space, namely, $m+j$ instead of $d$ dimensions.
The exact running times depend on the particular approximation algorithms and coreset constructions.
In Algorithm~\ref{algkmean} below, we consider any $\mathcal C$-clustering problem that is closed under rotations and reflections and such that
each $C\in \mathcal C$ is contained in some $j$-dimensional subspace.

\newcommand{\appalg}{\textsc{approx-solution}}
\begin{algorithm}[h!t]
    \caption{$\appalg(A,j,\eps,\alpha)$\label{algkmean}\label{kmeansall}}
{\begin{tabbing}
\textbf{Input: \quad   }  \= $A\in\REAL^{n\times d}$, an integer $k\geq1$ and error parameters $\alpha\geq 0$ and $\eps\in(0,1/2)$.\\
\textbf{Output: } \> An $\frac{\alpha(1+\eps)}{1-\epsilon}$-approximation $C$ for the $\CC$-clustering of $A$; see Corollary~\ref{corapprox}.
\end{tabbing}}
\vspace{-0.3cm}
   Set $m \gets \min\br{n,d,k +1 +\lceil 16(k+1)/\eps^2\rceil} -1$.\\
	 Compute the rank $m$ approximation $A^{(m)}$ of $A$ \\
   Compute an $\eps/8$-coreset $(S,\Delta',w)$ for $A^{(m)}$  for some $m+j$-dimensional subspace $L$ that contains the row vectors of $A^{(m)}$ \label{lk4}\\
   Compute an $\alpha$-approximation $C$ for the $\CC$-clustering of $(S,0,w)$\\ \tcc{the term $\Delta'$ can be ignored}
    \Return $C$
\end{algorithm}

\begin{corollary}\label{corapprox}
Let $\epsilon \in (0,1/2]$ and $A \in \REAL^{n\times d}$. Let $\CC$ be a (possibly infinite) set of non-empty subsets of $\REAL^d$ that is closed under
rotations and reflections such that each $C\in \CC$ is contained in a $j$-dimensional subspace
Let $C\in\CC$ be the output of a call to $\appalg(A,j,\eps,\alpha)$; see Algorithm~\ref{algkmean}.
Then
\[
\dist^2(A, C)\leq \frac{\alpha(1+\eps)}{1-\eps}\cdot\dist^2(A, C^*).
\]
\end{corollary}
\begin{proof}
Let $C^*\in\CC$ be the set that minimizes $\dist^2(A,C')$ over every $C'\in\CC$.
Let $\Delta=\Delta'+\norm{A-A^{(m)}}^2_F$ and $\dist_w^2(S, C)=\sum_{i=1}^{|S|} w_i \cdot \dist^2(S_{i*},C)$.
Hence,
\begin{align}
\dist^2(A, C)
\label{twokmean}&\leq \big(\dist_w^2(S, C)+\Delta\big)+\eps \cdot \dist^2(A, C)\\
\label{threekmean}&\leq \dist_w^2(S, C^*)+\Delta+(\alpha-1)\dist_w^2(S, C^*)+ \eps\dist^2(A, C)\\
\label{fourkmean}&\leq (1+\eps)\dist^2(A, C^*)+(\alpha-1)\dist_w^2(S, C^*)+ \eps\dist^2(A, C)\\
\label{fivekmean}&\leq (1+\eps)\dist^2(A, C^*)+(\alpha-1)(1+\eps)(\dist^2(A, C^*)) + \eps\dist^2(A, C)\\
\nonumber\label{sixkmean}&=\alpha(1+\eps)\dist^2(A, C^*)+ \eps\dist^2(A, C),
\end{align}
where~\eqref{twokmean},~\eqref{fourkmean} and~\eqref{fivekmean} follows from Theorem~\ref{thm:dim},~\eqref{threekmean} follows since $C$ is an $\alpha$-approximation to the $\CC$-clustering of $(S,0,w)$.
After rearranging the last inequality,
\[
\dist^2(A, C)\leq \frac{\alpha(1+\eps)}{1-\eps}\cdot\dist^2(A, C^*).
\]
where in the last inequality we used the assumption $\eps<1$.
\end{proof}


\section{Streaming Algorithms for Subspace Approximation and \mathk-Means Clustering}\label{sec:streaming}


Our  next step will be to show some applications of our coreset results. We will use the standard
merge and reduce technique \cite{bent} (more recently known as a general streaming method for \emph{composable coresets},
 e.g.~\cite{indyk2014composable,mirrokni2015randomized,aghamolaei2015diversity}), to develop a streaming algorithm \cite{AHV04}.
In fact, even for the off-line case, where all the input is stored in memory, the running time may be improved by using the merge and reduce technique.
\newcommand{\Alg}{\textsc{Subspace-Coreset}}

The idea of the merge and reduce technique is to read a batch of input points and
then compute a coreset of them. Then the next batch is read and a second coreset is
built. After this, the two coresets are merged and a new coreset is build.
Let us consider the case of the linear $j$-subspace problem as an example.
We observe that the union of two coresets is a coreset in the following sense:
Assume we have two disjoint point sets $A_1$ and $A_2$ with corresponding coresets
$(R_1,\Delta_1')$ and $(R_2,\Delta_2')$, such that
$$
 \dist^2(A_1,C) \le \dist^2(R_1,C) + \Delta_1 \le (1+\eps) \cdot \dist^2(A_1,C),
$$
and
$$
 \dist^2(A_2,C) \le \dist^2(R_2,C) + \Delta_2 \le (1+\eps) \cdot \dist^2(A_2,C).
$$
Then it also hold that
$$
\dist^2(A,C) \le \dist^2(R,C) + \Delta_1 + \Delta_2 \le (1+\eps) \cdot \dist^2(A,C)
$$
where $A=A_1 \cup A_2$ and $R=R_1 \cup R_2$. Thus, the set $R$ together with the real value $\Delta_1 + \Delta_2$
is a coreset for $A$.

The merges are arranged in a way such that in an input stream of length $n$, each input point is
involved in $O(\log n)$ merges. Since in each merge we are losing a factor of $(1+\epsilon')$ we
need to put $\epsilon' \approx \epsilon/\log n$ to obtain an $\epsilon$-coreset in the end.
We will now start to work out the details.


\subsection{Streaming Algorithms for the Linear \mathj-Subspace Problem}


We will start with the simplest case, which is a streaming algorithm for the linear subspace
approximation problem. In this case, the coreset construction does not involve weights and it is
deterministic, which allows us to use the most simple form of the merge and reduce paradigm.
We will require a function $\coresetsize(\epsilon, j) := j + \lceil j/\eps \rceil -1$ that denotes the size of the coreset
for the linear subspace approximation problem. Then we can use the algorithm \algstr \ below to maintain our summary of the data
and the algorithm \algout \ to output a coreset. The algorithm uses algorithm {\sc Subspace-Coreset} from Section \ref{sec:linearsubspace} as a subroutine.
We will assume that our algorithm has access to an input stream of points from $\REAL^d$ of unknown length. The algorithm receives
an error parameter $\eps$.

\begin{algorithm}[h!t]
    \caption{$\algstr(\eps, j)$\label{algstream}}
Set $Q\gets\emptyset$\\
\For{every integer $h$ from $1$ to $\infty$}{
Set $S_h \gets \emptyset$; $\Delta_h^S \gets 0$\\
Set $T_i\gets \emptyset$ and $\Delta_i^T \gets 0$ for every integer $1 \leq i\leq h$ \\
Set $\gamma\gets \eps/(10 h)$\\
\For{$2^h$ iterations}{
Read the next point from the input stream and add it to $Q$ \\
\If{$|Q|=2\cdot \coresetsize(\gamma,j)$\label{seven}}{
Set $(T,\Delta^T,w) \gets \Alg(Q,j,\gamma)$ \\ 
Set $i\gets1$\\
\While{$T_i\neq \emptyset$}{
Set $(T, \Delta^T,w) \gets \Alg(T\cup T_i, j, \gamma)$ \\
  Set $\Delta^T \gets \Delta^T +\Delta_i^T$\\
	Set $T_i\gets \emptyset$;	$\Delta_i^T \gets 0$ \\
  Set $i\gets i+1$ \\
}
Set $T_i\gets T$; $\Delta_i^T \gets \Delta^T$ \\
Define $S \gets \bigcup_{i=1}^{h} S_i \cup T_i$ and $\Delta^S \gets \sum_{i=1}^{h} \Delta_i^S + \Delta_i^T$\\
Set $Q\gets\emptyset$
}
}
Set $S_h\gets T$; $\Delta_h^S \gets \Delta^T$\\
}
\end{algorithm}

\begin{algorithm}[h!t]
    \caption{\algout$(j,\eps)$}
		Set $(T,\Delta^T,w) \gets \Alg(S\cup Q,j,\eps)$\\
		Set $\Delta^T \gets \Delta^T + \Delta^S$\\
\textbf{return} $(T,\Delta^T,w)$ \label{15}\\
\end{algorithm}

During the streaming, we only compute coresets of small sets of points.
The size of these sets depends on the smallest input that can be reduced by half using our specific coreset construction.
This property allows us to merge and reduce coresets of coresets for an unbounded number of levels, while introducing only multiplicative $(1+\eps)$ error.
Note that the size here refers to the cardinality of a set, regardless of the dimensionality or required memory of a point in this set.
We obtain the following result for the subspace approximation problem.

\begin{theorem}\label{streaming-linear-subspace}
Let $\eps\in(0,1/2)$ and $j \ge 1$. On input a stream of $n$ points, algorithm $\algstr$ maintains in overall time $O(ndj \log^2 n/\epsilon)$
a set $S$ of $O(j \log^2 n / \epsilon)$ points and a real value $\Delta^S$ such that for every linear $j$-subspace $C\subseteq\REAL^d$
the following inequalities are satisfied:
 $$
 \dist^2(A,C) \le \dist^2(S,C) + \Delta^S \le (1+\eps) \cdot \dist^2(A,C),
$$
where $A$ denotes the matrix whose rows are the $n$ input points.

Furthermore, algorithm $\algout$ computes in time $O(dj^2 \log^4 n/\epsilon^2)$ from $S$ and $\Delta$ a coreset $(T, \Delta_T,w)$ of size $j + \lceil j/\eps\rceil -1$ such that
$$
 \dist^2(A,C) \le \dist^2(S,C) + \Delta \le (1+3\eps) \cdot \dist^2(A,C),
$$
\end{theorem}

\begin{proof}
The proof follows earlier applications of the merge and reduce technique in the streaming setting~\cite{AHV04}.
We first observe that after $n$ points have been processed, we have $h = O(\log n)$. From this, the bound on the size of $S$ follows immediately.

To analyze the running time let $h^*$ be the maximum value of $h$ during the processing of the $n$ input points.
We observe that the overall running time $T(n)$ is dominated by the coreset computations. Since the
running time for the coreset computation for $n'$ input point is is $O(d(n')^2)$, we get
$$
T(n) \le \sum_{i=1}^{h^*} 2^i \cdot O(d i^2 j^2/\eps^2) = O(2^{h^*} d (h^*)^2 j^2/\eps^2).
$$
At the same time, we get $n \ge 2^{h^*} \cdot j (h^*-1) / \eps$ since the  value of $h$ reached the value $h^*$ and so
the stage $h^*-1$ has been fully processed. Using $h^* = O(\log n)$ we obtain
$$
T(n) = O(djn \log n / \eps).
$$
Finally, we would like to prove the bound on the approximation error. For this purpose fix some value of $h$.
We observe that the multiplicative approximation factor in the error bound for $T_i$ is $(1+\gamma)^i$ for $i \le h$.
Thus, this factor is at most $(1+\gamma)^h = (1+\frac{\epsilon}{10h})^h$. It remains to prove the following claim.
\begin{claim}
\label{claim:streaming}
$$
(1+\frac{\epsilon}{10h})^h \le  1+\eps.
$$
\end{claim}
\begin{proof}
In the following we will use the inequality $(1+1/n)^n < e < (1+1/n)^{n+1}$, which holds for all integer $n\ge 1$.
We first prove the statement when $10/\eps$ is integral. Then
$$
(1+\frac{\epsilon}{10h})^h  =  (1 +\frac{1}{\frac{10h}{\epsilon}})^{\frac{10h}{\eps} \cdot \frac{\eps}{10h} \cdot h}
 \le e^{\frac{\eps}{10}}  \le (1+\eps/10)^{(10/\eps +1) \cdot \eps/10} \le (1+\eps/10)^2 \le 1+\eps
$$

If $10/\eps$ is not an integer, we can find $\eps'$ with $\eps<\eps'< (1+1/10) \eps$ such that $10/\eps'$ is integral.
The calculation above shows that
$$
(1+\frac{\epsilon}{10h})^h \le (1+\frac{\epsilon'}{10h})^h \le (1+\eps'/10)^2 \le (1+\eps/5)^2 \le 1+\eps
$$
which finishes the proof.
\end{proof}
With the above claim the approximation guarantee follows.
Finally, we observe that the running time for algorithm \algout\ follows from Theorem \ref{thm:linearsubspace:coreset} and the claim on the quality is true because $(1+\eps)^2 \le (1+3\eps)$.
\end{proof}


\subsection{Streaming algorithms for the affine \mathj-subspace problem}


We continue with the affine $j$-subspace problem. This is the first coreset construction in this paper that uses weights.
However, we can still use the previous algorithm together with algorithm \textsc{affine-$j$-subspace-Coreset-Weighted-Inputs}
which can deal with weighted point sets. We obtain the following result. Let us use $\algstr^*$ and $\algout^*$
to denote the algorithms $\algstr$ and $\algout$ with algorithm $\Alg$ replaced by algorithm \textsc{affine-$j$-subspace-Coreset-Weighted-Inputs}.

\begin{theorem}\label{streaming-affine-subspace}
 On input a stream of $n$ points, algorithm $\algstr^*$ maintains in overall time $O(ndj \log^2 n/\epsilon)$
a set $S$ of $O(j \log^2 n / \epsilon)$ points weighted with a vector $w$ and a real value $\Delta^S$ such that for every affine
$j$-subspace $C\subseteq\REAL^d$ the following inequalities are satisfied (during the distance computation point weights
are treated as multiplicities):
 $$
 \dist^2(A,C) \le \sum_{i=1}^{|S|} w_i \cdot \dist^2(S_{i*},C) + \Delta^S \le (1+\eps) \cdot \dist^2(A,C),
$$
where $A$ denotes the matrix whose rows are the $n$ input points.

Furthermore, algorithm $\algout^*$  computes in time $O(dj^2 \log^4 n/\epsilon^2)$
from $(S,\Delta,w_S)$ an $\eps$-coreset $(T, \Delta_T,w_T)$ of size $j + \lceil j/\eps\rceil -1$ for the affine $j$-subspace problem.
\end{theorem}


\subsection{Streaming algorithms for \mathk-means clustering}


Next we consider streaming algorithms for $k$-means clustering. Again we need to slightly modify our approach due to the
fact that the best known coreset constructions are randomized. We need to make sure that the sum of all error probabilities
over all coreset constructions done by the algorithm is small. We assume that we have access to an algorithm
\textsc{$k$-MeansCoreset}$(A,k,\eps, \delta,w)$ that computes on input a weighted point set $A$ (represented by a matrix $A$ and weight vector
$w$) with probability $1-\delta$ an $\eps$-coreset $(S,\Delta,w)$ of size $\coresetsize(k,\eps,\delta)$ for the $k$-means clustering problem
as provided in Theorem \ref{tthm}.

\begin{algorithm}[h!t]
    \caption{\textsc{Streaming-$k$-Means-Approximation}$(k,\eps,\delta)$\label{algstream2}}
Set $Q\gets\emptyset$\\
Set $j \gets 2$\\
\For{every integer $h$ from $1$ to $\infty$}{
Set $S_h \gets \emptyset$; $\Delta_h^S \gets 0$; $u_h \gets 0$ \\
Set $T_i\gets \emptyset$, $v_i \gets 0$ and $\Delta_i^T \gets 0$ for every integer $1 \leq i\leq h$ \\
Set $\gamma\gets \eps/(10 h)$\\
\For{$2^h$ iterations}{
Read the next point from the input stream and add it to $Q$ \\
\If{$|Q|=2\cdot \coresetsize(k,\gamma,\delta)$}
{
Set $(T,\Delta^T,v) \gets \textsc{$k$-MeansCoreset}(Q,k,\gamma, \delta / j^2,v)$ \\ 
Set $j \gets j+1$\\
Set $i\gets1$\\
\While{$T_i\neq \emptyset$}
{
  Set $v$ to be the weight vector composed of $v$ and $v_i$\\
  Set $(T, \Delta^T,v) \gets \textsc{$k$-MeansCoreset}(T\cup T_i,k, \gamma, \delta/j^2,v)$ \\
  Set $j \gets j+1$\\
  Set $\Delta^T \gets \Delta^T +\Delta_i^T$\\
	Set $T_i\gets \emptyset$;	$\Delta_i^T \gets 0$ \\
  Set $i\gets i+1$ \\
}
Set $T_i\gets T$; $\Delta_i^T \gets \Delta^T$; $v_i \gets v$ \\
Define $S \gets \bigcup_{i=1}^{h} S_i \cup T_i$ and $\Delta^S \gets \sum_{i=1}^{h} \Delta_i^S + \Delta_i^T$\\
Define $w$ to be the weight vector corresponding to $S$\\
Set $Q\gets\emptyset$
}
Set $S_h\gets T$; $\Delta_h^S \gets \Delta^T$; $u_h \gets v$\\
}}
\end{algorithm}

\begin{theorem}\label{streaming-k-means}
 On input a stream of $n$ points, algorithm \textsc{Streaming$k$-Means-Approximation} maintains with probability at least $1-\delta$ in
overall time $nd(k\log n\log(1/\delta)/\eps)^{O(1)}$  a set $S$ of $M=(k\log n\log(1/\delta)/\eps)^{O(1)}$ points weighted with a vector $w$ and a real value $\Delta^S$
such that for every set $C\subseteq \REAL^d$ of $k$ centers the following inequalities are satisfied:
 $$
 \dist^2(A,C) \le \sum_{i=1}^{|S|} w_i \cdot \dist^2(S_{i*},C) + \Delta^S \le (1+\eps) \cdot \dist^2(A,C),
$$
where $A$ denotes the matrix whose rows are the $n$ input points.

Furthermore, with probability at least $1-\delta'$ we can compute in time $d(k\log n\log(1/\delta')/\eps)^{O(1)}$
from $(S,\Delta,w_S)$ a coreset $(T, \Delta^T,w_T)$ of size $O\left(\frac{k^3 \log^2 k}{\eps^4} \log(1/\delta')\right)$ such
that
$$
 \dist^2(A,C) \le  \sum_{i=1}^{|S|} w_i \cdot \dist^2(S_{i*},C) + \Delta \le (1+3\eps) \cdot \dist^2(A,C).
$$
Finally, we can compute in $|T|^{O(k/\eps)}$ time a $(1+O(\eps))$-approximation for the $k$-means problem from this coreset.
\end{theorem}

\begin{proof}
We first analyze the success probability of the algorithm.
In the $j$th call to a coreset construction via $\Alg$ during the execution of Algorithm~\ref{algstream}, we apply the above coreset construction with probability of failure $\delta/j^2$. After reading $n$ points from the stream, all the coreset constructions will succeed with probability at least
\[
1-\delta\sum_{j=2}^{\infty} \frac{1}{j^2}\geq 1-\delta.
\]
Suppose that all the coreset constructions indeed succeeded (which happens with probability at least $1-\delta$), the error bound follows from Claim
\ref{claim:streaming} in a similar way as in the proof of Theorem \ref{streaming-linear-subspace}. The space bound of $T$ follows from the fact
that $h=O(\log n)$ and since $j^2/\delta$ is at most $n^2/\delta$.

The running time follows from the fact that the computation time of a coreset of size $(k\log n\log(1/\delta)/\eps)^{O(1)}$ can be done in time
$d(k\log n\log(1/\delta)/\eps)^{O(1)}$.

The last result follows from the fact that for every cluster there exists
a subset of $O(1/\eps)$ points such that their mean is a $(1+\eps)$-approximation to the center of the cluster (and so we can enumerate all
such candidate centers to obtain a $(1+\eps)$-approximation for the coreset).
\end{proof}


\section{Coresets for Affine \jj-Dimensional Subspace \kk-Clustering}\label{sec:projclustering}

Now we discuss our results for the projective clustering problem. A preliminary version of parts of this chapter was published in~\cite{S14}.

\subsection{The Affine \jj-Dimensional Subspace \kk-Clustering Problem}\label{integerprojclustering}
In this section, we use the sensitivity framework to compute coresets for the affine subspace clustering problem. We do so by combining the dimensionality reduction technique from Theorem~\ref{thm:main2} with the work by Varadarajan and Xiao~\cite{VX12-soda} on coresets for the integer linear projective clustering problem.

Every set of $k$ affine subspaces of dimension $j$ is contained in a $k(j+1)$-dimensional linear subspace.
Hence, in principle we can apply Theorem~\ref{thm:dim} to the integer projective clustering problem, using $m := O(kj / \eps^2)$ and replace the input $A$ by the low rank approximation $A^{(m)}$.

The problem with combining this dimensionality reduction with known algorithms for the integer projective clustering problem is that the lower dimensional representation of a point set does not necessarily have integer coordinates even if the original points have this property. We discuss the details of this difficulty before we consider the technique by Varadarajan and Xiao to obtain the coreset result.
For a matrix $A\in\REAL^{n\times d}$ and a subset $V\subseteq\REAL^d$, we denote
\[
\dist_{\infty}(A,V)=\max_{i} \dist(A_{i*},V),
\]
where the maximum is over $i\in\br{1,\cdots,n}$.
We need the following well-known technical fact, where we denote the determinant of $A$ by $\det(A)$.
A proof can for example be found in ~\cite{GKL95}, where this theorem is the second statement of Theorem 1.4 (where the origin is a vertex of the simplex). 

\newcommand{\vol}{\mathrm{Vol}}
\begin{lemma}\label{volsim}
Let $A\in\REAL^{k\times d}$ be a matrix of full rank $k\leq d$. Let $S$ be the $k$-simplex that is the convex hull of the rows of $A$ and the origin,
\[
S=\br{a_1A_{1*}+\ldots+a_kA_{k*}\mid a_i\geq 0, 1\leq i\leq k, \sum_{j=1}^k a_j=1}.
\]
Then the $k$-volume of $S$ is
\[
\vol(S)=\frac{1}{k!}\sqrt{\det(AA^T)}.
\]
\end{lemma}

In the following lemma (that goes back to Lemma 5.1 in~\cite{VX12-soda}), we establish a lower bound on $\dist_{\infty}(A,\mathcal{T})$ for any $\mathcal T$ that is a set of $k$ affine $j$-dimensional subspaces. We will later construct a grid, move all points of $A^{(m)}$ to this grid and than scale in order to obtain integer coordinates. For bounding the error of this process, we will need Lemma~\ref{lem:lowerboundpc}.

Observe that a low dimensional $A$ could be completely covered by $\mathcal{T}$, implying that $\dist_{\infty}(A,\mathcal{T})$ would be zero. However, if the rank of $A$ is at least $k(j+1) + 1$, then there will always be at least one point that is not covered, and we can give a lower bound of the distance of this point to $\mathcal{T}$.

\begin{lemma}[Variation of Lemma~5.1 in~\cite{VX12-soda}]\label{lem:lowerboundpc}
 Let $M\ge 2, k$ and $j \leq d-1$ be positive integers.
Let $\mathcal{Q}_{jk}$ be the family of all sets of $k$ affine subspaces of $\mathbb{R}^d$, each of dimension $j$.
 Suppose that $A \in \{-M,\ldots,M\}^{n\times d}$ 
 is a matrix of rank larger than $k(j+1)$. Then
\begin{itemize}
\item
for every $C \in \mathcal{Q}_{jk}$ we have $\dist_{\infty}(A,C) \ge  \frac{1}{ (d M)^{cj}}$ for some universal constant $c>0$. 
\end{itemize}
If $A$ additionally satisfies $||A_{i*}||_2 \le M$, for all $1 \le i \le n$, then we have
\begin{itemize}
\item
$\dist_{\infty}(A,C) \ge  \frac{1}{M^{h(j)}}$ for a non-negative function $h$ that depends only on $j$.
\end{itemize}
\end{lemma}
\begin{proof}
Let $C\in\ucc{Q}_{jk}$ be any set of $k$ affine $j$-dimensional subspaces.
Consider the partitioning $\br{A_1,\cdots,A_k}$ of the rows in $A$ into $k$ matrices, according to their closest subspace in $C$. Ties broken arbitrarily. Let $A'$ be a matrix in this partition whose rank is at least $j+2$. There must be such a matrix by the assumptions of the lemma. By letting $L\in C$ denote the closest affine subspace from $C$ to the rows of $A'$, we have
\begin{equation}\label{eq144}
\dist_{\infty}(A,C)
\geq \dist_{\infty}(A',L).
\end{equation}

Without loss of generality, we assume that $L$ is a $j$-dimensional linear subspace (intersects the origin), otherwise, for the analysis 
we translate both $L$ and the rows of $A'$. \newcommand{\Sp}{\mathrm{sp}}
Let $B=\br{A_{i_1*'},\cdots, A_{i_{j+1}*}'}$ be a set of rows from $A'$ that span a $(j+1)$-dimensional linear subspace $\Sp\br{B}$. Let $V\subseteq \Sp\br{B}$ be a $j$-dimensional linear subspace that contains the projection of $L$ onto $\Sp\br{B}$. Since $A'\supseteq B$, and by the Pythagorean Theorem respectively,
\begin{equation}\label{eq244}
\dist_{\infty}(A',L)\geq \dist_{\infty}(B,L) \geq \dist_{\infty}(B,V).
\end{equation}
    
Consider a $j$-dimensional cube that is contained in $V$, and contains the origin as well as the projection of $B$ onto $V$. Suppose we choose the cube such that its side length is minimal, and let $s$ be this side length. 
For $A \in \{-M,\ldots,M\}^{n\times d}$, we know that
\begin{align}\label{boxlength-a}
s \le 2 \cdot\max_{b\in B} \norm{b}_2 \le 2 \sqrt{d} M .
\end{align}
If all points in $A$ satisfy $||A_i||\le M$, 
then \begin{align}\label{boxlength-b}
s \le 2 \cdot \max_{b\in B}\norm{b}_2 \le 2 M.
\end{align}

The cube can be extended to a $(j+1)$-dimensional box that also contains $B$ by assigning a side length of $2\dist_{\infty}(B,V)$ to the remaining orthogonal direction in $\Sp\br{B}$. 
The $(j+1)$-dimensional volume of this box is
$ \mathrm{Vol}(Box):= s^j \cdot 2\dist_{\infty}(B,V)$,
which means that 
\begin{equation}\label{rear}
\dist_{\infty}(B,V)\geq \frac{1}{2\cdot s^j} \cdot \mathrm{Vol}(Box).
\end{equation}
A lower bound for the volume of the box is obtained by noting that the box contains $B\cup\br{\vec{0}}$, and thus contains the $(j+1)$-simplex whose vertices are the points of $B$ and the origin. Observe that the origin is not contained in the convex hull of $B$ because the $j+1$ points in $B$ are linearly independent and span a $(j+1)$-dimensional linear subspace. Thus considering the simplex with vertices $B \cup \{\vec 0\}$ is well-defined, and this simplex is $(j+1)$-dimensional.
Hence, the volume $\mathrm{Vol}(Box)$ of the box is larger than the volume $\mathrm{Vol}(Simplex)$ of the simplex, i.e.,
\begin{equation}\label{box}
\mathrm{Vol}(Box)\geq \mathrm{Vol}(Simplex).
\end{equation}

By letting $F\in\REAL^{(j+1)\times {d}}$ denote the matrix whose rows are the points of $B$ and using $A=F$ in Lemma~\ref{volsim}, we have that the volume of the simplex is
\begin{equation}\label{jjj}
\mathrm{Vol}(Simplex):=\frac{1}{(j+1)!}\sqrt{|\det (F F^T)|}\geq \frac{1}{(j+1)!}\,,
\end{equation}
where the last inequality follows by combining the facts: (i) $\det(FF^T)=\det(D^2)\geq 0$ by letting $UDV^T$ denote the SVD of $F$, (ii) $\det(FF^T)\neq 0$ since $F$ is invertible (has full rank), and (iii) each entry of $F$ is an integer, so $\det(FF^T)>0$ implies $\det(FF^T)\geq 1$.
Combining the last inequalities yields 
\[
\begin{split}
\dist_{\infty}(A,C)
&\geq \dist_{\infty}(A',L)
\geq \dist_{\infty}(B,V) 
\geq \frac{1}{2s^j} \cdot \mathrm{Vol}(Box)\\
&\geq \frac{1}{2s^j} \cdot \mathrm{Vol}(Simplex)
\geq  \frac{1}{2(j+1)!\cdot s^j},
\end{split}
\]
where the inequalities hold respectively by~\eqref{eq144},~\eqref{eq244},~\eqref{rear},~\eqref{box}, and~\eqref{jjj}.
Now~\eqref{boxlength-a} implies that $\dist_{\infty}(A,C)\geq  \frac{1}{2(j+1)!\cdot (2\sqrt{d}M)^j} \ge \frac{1}{(dM)^{cj}}$ for some constant $c > 0$.
If, additionally, we have $\|A_{i*}\|_2 \le M$ then ~\eqref{boxlength-b} implies that $\frac{1}{2(j+1)!\cdot (2M)^j} \ge \frac{1}{M^{h(j)}}$ for a function $h$ that only depends on $j$.
\end{proof}


\subsection{\Linfty-coresets}


Our next step is to introduce $\L_\infty$-coresets, which will be a building block in the computation
of coresets for the affine $j$-dimensional $k$-clustering problem.
An $\L_\infty$-coreset $S$ is a coreset approximating the maximum distance between the point set and any query shape. The name is due to the fact that the maximum distance is the infinity norm of the vector that consists of the distances between each point and its closest subspace. The next definition follows \cite{EV05}.

\begin{definition}[$\L_{\infty}$-coreset]\label{linfty}
Let $A$ in $\REAL^{n \times d}$, and $\eps>0$. Let $\ucc{Q}$ be a family of closed and non-empty subsets of $\RR^d$.
A matrix $S \in \REAL^{r \times d}$ whose rows are a subset of $r$ rows from $A$ is an \emph{$\epsilon$-$\L_\infty$-coreset for $(A,\ucc{Q})$}, if for every $C \in \ucc{Q}$ we have
\[
\max_{i\in[n]} \dist(A_{i\ast},C) \le (1+\eps) \cdot \max_{i\in[r]} \dist(S_{i\ast},C).
\]
If $\ucc{Q}=\ucc{Q}_{jk}$ is the family of all sets of $k$ affine subspaces of dimension $j$, then we call the $\epsilon$-$\L_\infty$-coreset  an \emph{$\L_\infty$-$(\epsilon,j,k)$-coreset for $A$}.
\end{definition}

We need the following result on $L_\infty$-coresets for our construction.

\begin{theorem}[\cite{EV05}]\label{evmaintheorem}
Let $M\geq2$ be an integer and $A\in\br{-M,\ldots,M}^{n \times d}$.
Let $k\geq1$ and $\eps\in(0,1)$.
There is an $\L_\infty$-$(\epsilon,d-1,k)$-coreset $S$ for $A$, of size $|S|=(\log (M) / \eps)^{f(d,k)}$, where $f(d,k)$ depends only on $d$ and $k$. Moreover, $S$ can be constructed (with probability $1$) in $n\cdot |S|^{O(1)}$ time.
\end{theorem}

If $(j+1)k$ is much smaller than $d-1$, then we want to avoid the dependency on $d$. We observe that for any set of $k$ affine $j$-dimensional subspaces, the union of $A$ and the subspaces is contained in a linear subspace of dimension $k(j+1)+r$. Assume that $(j+1)k+r \le d$ and let $V$ be an arbitrary subspace of dimension $k(j+1)+r$ that contains $A$. Representing $A$ in an arbitrary orthonormal basis of $V$ results in a matrix  $A' \in \RR^{n \times (k(j+1)r)}$.

By Theorem~\ref{evmaintheorem} there exists a coreset $S' \subset A'$ of size $(\log (M) / \eps)^{f(d,k)}$ that satisfies the coreset property for $A'$ and all $j$-dimensional affine subspaces of $V$. Since $A$ and $A'$ describe the same points, this also holds for $A$ when we replace $S'$ by the corresponding subset $S$ of $A$. Now let $V'$ be any affine $j$-dimensional subspace of $\RR^d$. We can define a rotation that rotates $V'$ into $V$ while changing neither $A$ nor the distances between points in $A$ and $V'$. We get a subspace $V''$ that lies within $V$. Thus, $S$ satisfies the coreset property for $V''$ and this implies that it satisfies the coreset property for $V'$ as well. Thus, the following corollary is true.

\begin{corollary}[\cite{EV05}]\label{rankr}
Let $M\geq2$ be an integer and $A\in\br{-M,\ldots,M}^{n \times d}$ be a matrix of rank $r$.
Let $k\geq1$, $j\in\{1,\ldots,d-1\}$, and let $\eps\in(0,1)$. Assuming the singular value decomposition of $S$ is given,
an $\L_\infty$-$(\epsilon,j,k)$-coreset $S\subseteq A$ for $A$ of size $|S|=(\log (M) / \eps)^{f(j,k,r)}$ can be constructed in $(n+d) \cdot |S|^{O(1)}$ time, where $f(j,k,r)$ depends only on $j$, $k$ and $r$. 
\end{corollary}

\subsection{Sensitivity bounds via \Linfty-coresets}

A key idea in the work of Varadarajan and Xiao is a way to bound the sensitivity of a point set $A$ based on a so-called $\L_\infty$-coreset of $A$.
For this purpose, let us state our problem in terms of the sensitivity framework.
Let $F_A(\ucc{Q}):=\{ f_{A_{i\ast}}\mid i\in[n] \}$ where  $f_{A_{i\ast}} : \ucc{Q} \to [0,\infty)$ and $f_{A_{i\ast}}(C) = \dist^2(A_{i\ast},C)$ for every $i\in[n]$ and $C\in\ucc{Q}$.
The following lemma states how to obtain bounds on the sensitivities using $\L_\infty$-coresets. We slightly expand it to deal with weighted point set.

\begin{lemma}[Lemma 3.1 in \cite{VX12-soda}, weighted and applied to squared distances]\label{sensitivitylemmasquareddistances}
Let $A\in \REAL^{n\times d}$ be a matrix whose rows have non-negative weights $w_1,\dots,w_n \ge 1$ and let $W= \sum_{i=1}^n w_i$. 
Let $\ucc{Q}$ be a set of non-empty subsets of $\REAL^d$ and let $F_A(\ucc{Q}):=\{ f_{A_{i\ast}}\mid i\in[n] \}$ where  
$f_{A_{i\ast}} : \ucc{Q} \to [0,\infty)$ and $f_{A_{i\ast}}(C) = w_i \cdot \dist^2(A_{i\ast},C)$ 
for every $i\in[n]$ and $C\in\ucc{Q}$.

Suppose that for every matrix $A' \in \REAL^{m\times d}$ whose (unweighted) rows are rows in $A$ there is an $\epsilon -\L_\infty$-coreset 
$S$ for $(A',\ucc{Q})$ for $\epsilon = 1/2$ of size $|S|\leq g(m)$ that can be computed in time $t(m)$. 
Then $\tilde{\sigma}$ such that $\tilde{\sigma}(f_{A_{i*}})\geq \sigma(f_{A_{i*}})$ for every $i\in [n]$, and
\[
\mathfrak{\tilde{S}}(F) :=\sum_{i\in [n]}\tilde{\sigma}(f_{A_{i*}})\in O(\log W)\cdot g(n),
\]
can be computed in $n \cdot t(n)$ time.
\end{lemma}
\begin{proof}
In order to deal with the weights we proceed as follows. We first define $w_i' = \lfloor w_i \rfloor$ and $W' = \sum_{i=1}^n w_i'$. 
Since $w_i \ge 1$ all $w_i'$ are within a factor of $2$ of $w_i$. Then we replace the input matrix $A$ by a matrix $B$ that contains
$w_i'$ copies of row $i$ of matrix $A$, i.e. we replace each row of $A$ by a number of unweighted copies corresponding to its weight (rounded down). 

Now the proof proceeds similarly to the original proof of Varadarajan and Xiao. At the moment we ignore the running time and will address
an efficient implementation at the very end of the proof. The proof uses a nested sequence of $\ell$ subsets of $B$ for an $\ell \leq W'$. 
$B_1$ is the matrix $B$. The other sets are computed iteratively in the following way. If $B_t$ contains at most $g(n)$ points, the sequence ends and $\ell:=t$. Otherwise, an $\eps$-$\L_\infty$-coreset $S_t$ of $B_t$ for $\eps=1/2$ is computed, and the set $B_{t+1}$ is defined as $B_t \setminus S_t$.
Note that for the coreset computation we may remove multiple copies of a point from the current set $B_t$ and so the resulting coreset will
have size $g(n)$.

The result is a sequence of subsets $B_\ell\subseteq B_{\ell-1} \subseteq B_2 \subseteq B_1 = B$ with $|B_\ell| \le g(n)$ and a sequence of coresets $S_1,\ldots,S_{\ell-1}$. Notice that for $S_\ell := B_\ell$, the coresets $S_1,\ldots,S_\ell$ form a partitioning of $B$.
 Now let $B_{i\ast}$ be an input point, and let $v \in \{1,\ldots,\ell\}$ be the largest index of a coreset $S_v$ that contains $B_{i\ast}$, \ie\ $B_{i\ast} \in S_v$. Let $C \in \ucc{Q}$. The goal is to upper bound the sensitivity of $B_{i\ast}$ by lower bounding the contribution of the remaining points.

Consider the set $B_u$ for some $1 \le u \le v$, and notice that it contains $B_{i\ast}$ by definition. For each $u \in \{1,\ldots,v\}$, let $B_{i_u\ast}$ be one of the points in $S_u$ of maximum distance to $C$. By the $\L_\infty$-coreset property, this implies that
\[
\dist(B_{i\ast},C) \le (1+\eps) \cdot \dist(B_{i_u\ast},C).
\]
That is,
\[
\dist^2(B_{i_u\ast},C) \ge \frac{1}{(1+\eps)^2} \dist^2(B_{i\ast},C).
\]
Using this with the fact that $\br{B_{i_1*}, \ldots, B_{i_v*}}$ is a subset of $B$ yields
\[
\dist^2(B,C) \geq \sum_{x=1}^v \dist^2 \cdot (B_{i_u\ast},C) \ge {(1+\eps)^2} \cdot v \cdot \dist^2(B_{i\ast},C) .
\]
By the definition of the sensitivity of a point, splitting a point into $k$ equally weighted points 
leads to dividing its sensitivity by $k$. Recall that $B$ contains $w_i'$ copies of $A_{i\ast}$.
This implies that for every pair $A_{i\ast}$ and $j$ with $A_{i\ast} = B_{j \ast} \in S_v$ we get
\begin{eqnarray*}
\sigma(f_{A_i\ast}) & := & \max_{C\in \ucc{Q}}\frac{w_i \cdot \dist^2(A_{i\ast},C)}{\dist^2_w(A,C)} \\
& \le & \max_{C\in \ucc{Q}}\frac{2 w_i' \cdot \dist^2(A_{i\ast},C)}{\dist^2(B,C)} \\
& = & \max_{C\in \ucc{Q}} \frac{2 w_i' \cdot \dist^2(B_{j \ast},C)}{\dist^2(B,C)} \\
& \le & \frac{2 w_i' \dist^2(B_{j \ast},C)} {\frac{v}{(1+\eps)^2} \dist^2(B_{j\ast},C)}\\
&=& \frac{2 w_i'(1+\eps)^2}{v}.
\end{eqnarray*}

Thus, the sensitivity of $f_{A_i\ast}$ is at most $2 w_i'$ times the sensitivity of one of its $w_i'$ copies in $B$.
We define $\tilde{\sigma}(f_{A_i\ast})= 2 w_i' \frac{(1+\eps)^2 \cdot }{v}$ and $\tilde{\sigma}(f_{B_j\ast})= \frac{(1+\eps)^2 \cdot }{v}$. 
To estimate the total sensitivity we sum up the sensitivities of the rows of $B$. the total sensitivity of $B$ is bounded by
\begin{align*}
\mathfrak{\tilde{S}}(F)&=\sum_{v=1}^{\ell} \sum_{B_{i\ast}\in B_v} \tilde{\sigma}(f_{B_{i\ast}}) \le \sum_{v=1}^{\ell}  \frac{|S_v| (1+\eps)^2}{v}\\
&\le g(n) \cdot (1+\eps)^2 \cdot  \sum_{v=1}^{\ell}  \frac{1}{v}
\leq  g(n) \cdot (1+\eps)^2 \cdot  \sum_{v=1}^{W'}  \frac{1}{v}\\
& \le g(n) \cdot (1+\eps)^2 \cdot (\ln W' +1),
\end{align*}
where the second inequality follows because $|S_v|\leq g(n)$ by its definition, and the last inequality is a bound on the harmonic number $\mathcal{H}_n$.

It remains to argue how to efficiently implement the algorithm. 
For this purpose we modify the construction and construct a sequence of sets $A_i$ in the following way.
At the beginning each input point $A_{i*}$ is assigned its weight $w_i' = \lfloor w_i \rfloor$. Once we compute an $\eps$-coreset 
$S_t$ of $A_t$ for $\eps=1/2$ is computed, we assign a weight $u_t$ to all its input points that equals the minimum (current) $w_i'$
of a row $A_{i*}$ that has been included in $S_t$. In comparison to the previous construction we can think of $S_t$ as the union of 
$u_t$ unweighted coresets containing the same points as $S_t$.
For each coreset point $A_{i*}$ we then subtract $u_t$ from $w_i'$. If the weight becomes $0$ we remove the point from $A$. Note that
in every interation at least one point gets weight $0$. As before, the sequence ends, if $A_t$ contains at most $g(n)$ points and we set $\ell:=t$. 
As before, the result is a sequence of subsets $A_\ell\subseteq A_{\ell-1} \subseteq A_2 \subseteq A_1 = A$ with $|A_\ell| \le g(n)$ and a sequence of coresets 
$S_1,\ldots,S_{\ell-1}$ with weights $u_1,\dots, u_{\ell-1}$. Notice that for $S_\ell := B_\ell$, the coresets $S_1,\ldots,S_\ell$ form a fractional partitioning of $B$, i.e. the sum of weights over all ocurrences of $A_{i*}$ is $w_i'$. 
We can compute such a sequence in time $O(n t(n))$ and compute in $O(n)$ time the bounds on the sensitivities from it.
\end{proof}

\subsection{Bounding sensitivities by a movement argument}

In this section we will describe a way to bound the sensitivities using a movement argument. Such
an approach first appeared in \cite{varadarajan} and we will present a slight variation of it.

\begin{theorem}[Variant of a Theorem from \cite{varadarajan}]\label{sens_move}
Let $A\in \REAL^{n\times d}$ be a matrix. 
 Let the rows of matrix $A$ be weighted with non-negative weights 
$w_1,\dots,w_n \ge 1$ and let $W= \sum_{i=1}^n w_i$. Let $A' \in \REAL^{n \times d}$ be a matrix such that 
$$
\sum_{i=1}^n w_i \cdot \|A_{i*} - A_{i*}'\|_F^2 \le \alpha \cdot \dist_w^2(A,C^*),
$$
where $C^*$ minimizes  $\dist_w^2(A,C)$ over all sets $C$ of $k$ $j$-dimensional subspaces.
If $\dist_w^2(A,C*) >0$ then we have
$$
  (4+4 \alpha) (\sigma(f_{A_{i*}'}) + \frac{w_i \cdot \dist^2(A_{i*},A_{i'*})}{\dist_w^2(A,C^*)})  \ge \sigma(f_{A_{i*}}),
$$
where $\sigma(f_{A_{i*}}) = \sup_C \frac{w_i \cdot \dist^2(A_{i*},C)}{\dist_w^2(A,C)}$ (here the supremum is taken over all
sets $C$ of $k$ $j$-dimensional subspaces with $\dist^2(A,C) >0$).
\end{theorem}

\begin{proof}
Let $A$ and $A'$ be defined as in the theorem. Let $C$ be an arbitrary set of $k$ $j$-dimensional affine subspaces. 
For every row of $A$ we have
\begin{eqnarray*}
\frac{\dist^2(A_{i*},C)}{\dist_w^2(A,C)} & \le & \frac{2 \cdot (\dist^2(A_{i*},A_{i*}') + \dist^2(A_{i*}',C))}{\dist_w^2(A,C)}\\
& \le & 2 \cdot \left( \frac{\dist^2(A_{i*}',C)}{\dist_w^2(A,C)} + \frac{\dist^2(A_{i*},A_{i*}')}{\dist_w^2(A,C)} \right) \\
\end{eqnarray*}
The result is immediate if $\dist_w^2(A',C) =0$ since then $\dist^2(A_{i*}',C)=0$ as well. Hence, we can assume
$\dist_w^2(A',C) >0$ and obtain
\begin{eqnarray*}
& \le & (4+4 \alpha) \cdot \big( \frac{\dist^2(A_{i*}',C)}{\dist_w^2(A',C)} + \frac{\dist^2(A_{i*},A_{i*}')}{\dist_w^2(A,C^*)} \big) \\
\end{eqnarray*}
where the last inequality follows from
$$
\dist^2_w(A',C) \le 2 ( \sum_{i=1}^n w_i \cdot \|A_{i*} - A_{i*}'\|_F^2 + \dist^2_w(A,C)) \le (2+2\alpha) \cdot \dist_w^2(A,C).
$$
Now the result follows from the definition of sensitivity.
\end{proof}

\subsection{Coresets for the Affine \jj-Dimensional \kk-Clustering Problem}

In this section, we combine the insights from the previous subsections and conclude with our coreset result.

\begin{theorem}\label{thmproj}
Let $k\ge 1$, $j\ge 0$ and $M\ge 4$ be fixed integers and let $1/2 > \eps,\delta >0$.
Let $A \in \mathbb Z^{n \times d}$ be a matrix with $\|A_{i*}\|_2 \le M$ for all $1\le i \le n$.
Let the rows of $A$ be weighted with $w_1,\dots, w_n \ge 1$ and let $W= \sum_{i=1}^n w_i$.
If the rank of $A$ is at most $k (j+1)$ then in time $O(\min(n^2d,d^2n) + n(n+d) (\log M)^{h(j,k)})$ and with probability at
least $1-\delta$ we can construct
an $\eps$-coreset $(S, 0, u)$ for the affine $j$-dimensional
$k$-clustering problem of size
\[
|S|=O\big(\frac{\log W \log \log W \cdot\log(M)^{h(j,k)} \log(1/\delta)}{\epsilon^2} \big),
\]
where $h$ is a function that depends only on $j$ and $k$. Furthermore, the points in $S$ have
integer coordinates and $u_1,\dots, u_{|S|} \ge 1$ and the points have norm at most $M$.
\end{theorem} 

\begin{proof}
Let $\ucc{Q}_{jk}$ be the family of $k$ affine subspaces of $\REAL^d$, each of dimension $j$.
We would like to use Theorem \ref{sensitivitysamplingscheme} with $\ucc{Q}_{jk}$ and where for each row
$A_{i*}$ we have a function $f_{A_i}$ with $f_{A_i}(C)= \dist^2(A_{i*},C)$ for all $C\in \ucc{Q}_{jk}$.
Since $A$ has rank at most $k (j+1)$ and any set of $k$ affine $j$-subspaces is contained in a 
$k (j+1)$-dimensional linear subspace, by symmetry it will be sufficient to assume that 
$d= O(k (j+1))$.
By Corollary \ref{vc-dim} the VC-dimension of the range spaces $\mathfrak R_{\ucc{Q}_{jk},F^*}$ is therefore $O((j+1)^2k^2 \log k)$,
where $F$ is the set of the $f_{A_{i*}}$ and $F^*$ as defined in Theorem \ref{sensitivitysamplingscheme}.

It remains to argue how to compute upper bounds on the sensitivities and get an upper bound for the total sensitivity. 
The rank of $A$ is at most $r=k(j+1)$, so  Corollary~\ref{rankr} implies that an $\L_\infty$-$(j,k)$-coreset $S\subseteq A$ 
of size $g(n):=(\log M)^{f(k,j)}$ for $A$, can be constructed in $O(\min(n^2d+d^2n)+(n+d)\cdot g(n)^{O(1)}$ time, where $f(j,k)$ depends
 only on $j$ and $k$. Using this with Lemma~\ref{sensitivitylemmasquareddistances} yields an upper bound on the sensitivity 
$\sigma(f_{i})$  for every $i\in [n]$, such that the total sensitivity is bounded by
\[
\mathfrak{S}(A) := \sum_{i=1}^n \sigma(f_{i})\le O(\log W)  g(n),
\]
and the individual sensitivities can be computed in $n(n+d)\cdot g(n)^{O(1)}$ time. The result follows from
Theorem \ref{sensitivitysamplingscheme} and the fact that the coreset computed in Theorem \ref{sensitivitysamplingscheme}
is a subset of the input points.
\end{proof}

\begin{theorem}\label{thmproj2}
Let $A \in \mathbb{Z}^{n \times d}$ be a matrix of rank greater than $k(j+1)$ whose rows are weighted with weights $w_1,\dots, w_n\ge 1$ and whose maximum row norm is $M\ge 4$ . 
Let $W= \sum_{i=1}^n w_i$ and let $\eps \in (0,1/2)$.
Then in time $O(\min(n^2d,d^2n) + n(n+d) (\log WnM)^{h(j,k)}))$ and with probability at least $1-\delta$ we can construct an $\eps$-coreset $(S, \Delta, u)$ for the affine $j$-dimensional
$k$-clustering problem of size
\[
|S|=\frac{\left(\log(M W n)\right)^{h(j,k)} \cdot \log (1/\delta)}{\eps^2},
\]
where $h$ is a function that depends only on $j$ and $k$. Furthermore, the norm of each
row in $S$ is at most $M$ and $u_1,\dots, u_{|S|} \ge 1$.
\end{theorem} 
\begin{proof}
The outline of the proof is as follows. We first apply our results on dimensionality reduction 
for coresets and reduce computing a coreset for the input matrix $A$ to computing a coreset for
the low rank approximation $A^{(m)}$ for $m=O(k(j+1)/\eps^2)$. A simple argument would then be to snap
the points to a sufficiently fine grid and apply the reduction to $l_\infty$-coresets summarized in
this section. However, such an approach would give a coreset size that is exponential in $m$ (and so in $1/\eps$),
which is not strong enough to obtain streaming algorithms with polylogarithmic space. 

Therefore, we will proceed slightly differently. We still start by projecting $A$ to $A^{(m)}$. However,
the reason for this projecting is only to get a good bound on the VC-dimension. In order to compute upper
bounds on the sensitivities of the points we apply Lemma \ref{sens_move} in the following way. We project
the points of $A^{(m)}$ to an optimal $k(j+1)$-dimensional subspace and snap them to a sufficiently fine grid. 
Then we use Lemma \ref{sens_move} to get a bound on the total sensitivity. Note that we can charge the cost of
snapping the points since the input matrix has rank more than $k(j+1)$ and so by Lemma \ref{lem:lowerboundpc} there is
a lower bound on the cost of an optimal solution. We now present the construction in detail.

Our first step is to replace the input matrix $A$ by a low rank matrix. An annoying technicality is that
we would like to make sure that our low rank matrix has still optimal cost bounded away from $0$. We therefore
proceed as follows. We take an arbitrary set of $k(j+1)+1$ rows of $A$ that are not contained in a $k (j+1)$-dimensional
subspace. Such a set must exist by our assumption on the rank of $A$. We use $B_1$ to denote the matrix that
corresponds to this subset (with weights according to the corresponding weights of $A$) and we use $B_2$ to denote
the matrix corresponding to the remaining points. We then compute 
$B_2^{(m)}$ for a value $m=  \min\br{n,d,k(j+1) + \lceil 32 k(j+1)/\epsilon^2 \rceil}-1$. If the rows are weighted, then we can think of a point weight as the multiplicity of a point and 
compute the low rank approximation as described in the proof of Theorem \ref{tthm} and we let $B^* = B_2V^{(m)}(V^{(m)})^T$ denote the projection of the weighted 
points on the subspace spanned by the first $m$ right singular values of $V$, where $B_2=U\Sigma V^T$ is the singular 
value decomposition of $B_2$ (and we observe that the row norms of $B^*$ are at most $M$).
We use $B$ to denote the matrix that corresponds to the union of the matrices $B_1$ and $B^*$.
In the following we will prove the result for the unweighted case and observe that it immediately transfers to the weighted case by reducing weights to 
multiplicities of points. We observe that by Theorem \ref{thm:main2} with $\epsilon$ replaced by $\epsilon/2$ we obtain for every set $C$ that is the union of
$k$ $j$-dimensional affine subspaces:
$$
\left|  \left(\dist^2(B,C) + \norm{B_2-B_2^{(m)}}_F^2 \right)- \dist^2(A,C) \right| \le \frac{\eps}{2} \cdot \dist^2(A,C).
$$
Now let $(S,\Delta',w)$ be an $(\eps/8)$-coreset for the $j+1$-dimensional affine $k$-subspace clustering problem 
in a subspace $L$ that contains $B$ and has dimension $r + k (j+1)$, where $r$ is the rank of $B$. Using identical arguments as in the proof
of Theorem \ref{thm:dim} we obtain that $(S, \Delta' + \|B_2-B_2^{(m)}\|_F^2 , w)$ is an $\epsilon$-coreset for $A$.

Thus, it remains to show how to obtain an $(\eps/8)$-coreset for $B$. For this purpose we define $\ell = k (j+1)$. 
We observe that we can obtain matrix $B^* = AV^{(\ell)} (V^{(\ell)})^T$ (with similar modifications for the weighted case
as above). Our goal is now to use $B^*$ to obtain the sensitivities of the points.
We know by Lemma \ref{lem:lowerboundpc} and the fact that the weights are at least $1$ and by the construction of $B$ (in particular 
the selection of $B_1$) that $\dist_2^2(B,C^*) \ge \frac{1}{M^{c(j,k)}}:= L$ for some non-negative function $c$ that depends only 
on $j$ and $k$. We now define a grid on the span of $B$ such that the cell diagonal $r$ of each grid cell satisfies
$r^2 \le \dist_w^2(A,C^*) / (nW)$. This can be achieved with a grid whose side length is $O(\frac{1}{(MnW)^{c(j,k)}})$.
We then snap all grid points to their closest grid point.
Let $B'$ be the resulting matrix. We observe that
$$
\sum_{i=1}^n w_i \|B^*_{i*}-B'_{i*}\|_2^2 \le \dist_w^2(B,C^*).
$$
Furthermore, we know that the optimal solution for $B$ is contained in an $\ell$-dimensional subspace. This implies that
$$
\sum_{i=1}^n w_i\|B_{i*} - B'_{i*}\|_2^2 \le 4 \dist_w^2(B,C^*).
$$
By Lemma \ref{sens_move} it follows we can compute upper bounds on the sensitivites of $B$ by using the sensitivities of $B'$
plus a term based on the movement distance. The total sensitivity will be bounded by a constant times the total 
sensitivity of $B'$. 

It remains to argue how to compute upper bounds on the sensitivities and get an upper bound for the total sensitivity. 
The rank of $B'$ is at most $r=k(j+1)$, so  Corollary~\ref{rankr} implies that an $\L_\infty$-$(j,k)$-coreset $S\subseteq B$ 
of size $g(n):=(\log (MnW))^{f(k,j)}$ for $B$, can be constructed in $\min(n^2d,d^2n) + (n+d) \cdot g(n)^{O(1)}$ time, where $f(j,k)$ depends
 only on $j$ and $k$. Using this with Lemma~\ref{sensitivitylemmasquareddistances} yields an upper bound on the total sensitivity 
of $O(\log W) g(n)$ and the individual sensitivities can be computed in $n(n+d)\cdot g(n)^{O(1)}$ time. The result follows from
Theorem \ref{sensitivitysamplingscheme}.
\end{proof}

\section{Streaming Algorithms for Affine \jj-Dimensional Subspace \kk-Clustering}

We will consider a stream of input points with integer coordinates and whose maximum norm is bounded by $M$.
In principle, we would like to apply the merge and reduce approach similarly to what we have done in the 
previous streaming section. However, we need to deal with the fact that the resulting coreset does not have
integer coordinates, so we cannot immediately apply the coreset construction recursively. Therefore,
we will split our streaming algorithm into two cases. As long as the input/coreset points lie in a low
dimensional subspace, we apply Theorem \ref{thmproj} to compute a coreset. This coreset is guaranteed to have 
integer coordinates of norm at most $M$. 
Once we reach the situation that the input points are not contained in a $(k(j+1))$-dimensional subspace
we will switch to the coreset construction of Theorem \ref{thmproj2}. We will exploit that by Lemma \ref{lem:lowerboundpc} we have a lower
bound of, say, $L$ on the cost of the optimal solution. In order to meet the prerequisites of Theorem \ref{thmproj2}
we need to move the points to a grid. If the grid is sufficiently fine, this will change the cost of any solution 
insignificantly and we can charge it to $L$. 

We will start with the first algorithm. We assume that there is an algorithm \textsc{$(k,j)$-SubspaceCoreset}$(Q,k,j,\gamma, \delta / j^2,v)$
that computes a coreset of size $\coresetsize(\epsilon,\delta, j,k,M,W)$, where $\coresetsize(\epsilon,\delta, j,k,M,W)$
is the bound guaranteed by Theorem \ref{thmproj}. We do not specify the coreset algorithm is pseudocode since the result
is of theoretical nature and the algorithm rather complicated. 

\begin{algorithm}[h!t]
   \caption{$\algstr(\eps, \delta, k, j, M)$\label{algprojstream}}
Set $Q\gets\emptyset$\\
Set $j \gets 2$\\
\For{every integer $h$ from $1$ to $\infty$}{
Set $S_h \gets \emptyset$; $u_h \gets 0$ \\
Set $T_i\gets \emptyset$, $v_i \gets 0$ for every integer $1 \leq i\leq h$ \\
Set $\gamma\gets \eps/(10 h)$\\
Set $W \rightarrow 0$\\
\For{$2^h$ iterations}{
Read the next point from the input stream and add it to $Q$ \\
Set $W \rightarrow W+1$\\
\If{$|Q|=2\cdot \coresetsize(\gamma,\delta/j^2,k,j, M, W)$}
{
If $Q$ is not contained in a $k(j+1)$-dimensional subspace then continue with Algorithm \ref{algprojstream2}\\
Set $(T,0,v) \gets \textsc{$(k,j)$-SubspaceCoreset}(Q,k,j,\gamma, \delta / j^2,v)$ \\ 
Set $j \gets j+1$\\
Set $i\gets1$\\
\While{$T_i\neq \emptyset$}
{
  Set $v$ to be the weight vector composed of $v$ and $v_i$\\
  If $T \cup T_i$ is not contained in a $k(j+1)$-dimensional subspace then continue with Algorithm \ref{algprojstream2}\\
	Set $(T, 0,v) \gets \textsc{$(k,j)$-SubspaceCoreset}(Q,k,j,\gamma, \delta / j^2,v)$  \\
  Set $j \gets j+1$\\
	Set $T_i\gets \emptyset$ \\
  Set $i\gets i+1$ \\
}
Set $T_i\gets T$; $v_i \gets v$ \\
Define $S \gets \bigcup_{i=1}^{h} S_i \cup T_i$ \\
Define $w$ to be the weight vector corresponding to $S$\\
Set $Q\gets\emptyset$
}
Set $S_h\gets T$; $u_h \gets v$\\
}}
\end{algorithm}

Now we turn to the second algorithm. We assume that the algorithm receives a lower bound of $L$ on the 
cost of an optimal solution. Such a lower bound follows from Lemma \ref{lem:lowerboundpc} when the input consists of
integer points that are not contained on a $k(j+1)$-dimensional subspace. Since this is the case when
Algorithm \ref{algprojstream2} is invoked, we may assume that $L\ge \frac{1}{M^{h(j)}}$.

\begin{algorithm}[h!t]
   \caption{$\algstr(\eps, \delta, k, j, M, L)$\label{algprojstream2}}
Set $Q\gets\emptyset$\\
Set $j \gets 2$\\
\For{every integer $h$ from $1$ to $\infty$}{
Set $S_h \gets \emptyset$; $\Delta_h^S \gets 0$; $u_h \gets 0$ \\
Set $T_i\gets \emptyset$; $\Delta_i^T \gets 0$; $v_i \gets 0$ for every integer $1 \leq i\leq h$ \\
Set $\gamma\gets \eps/(20 h)$\\
Set $W\rightarrow 0$\\
\For{$2^h$ iterations}{
Read the next point from the input stream and add it to $Q$ \\
Set $W \rightarrow W+1$\\
\If{$|Q|=2\cdot \coresetsize(\gamma,\delta/j^2,k,j, M, W)$}
{
Snap the input points to a grid of side length $\gamma^2 \cdot L/(100 d)$  \\
Set $(T,\Delta^T,v) \gets \textsc{$(k,j)$-SubspaceCoreset}(Q,k,j,\gamma, \delta / j^2,v)$ \\ 
Set $j \gets j+1$\\
Set $i\gets1$\\
\While{$T_i\neq \emptyset$}
{
  Set $v$ to be the weight vector composed of $v$ and $v_i$\\
  If $T \cup T_i$ is not contained in a $k(j+1)$-dimensional subspace then continue with second algorithm\\
	Set $(T, \Delta^T,v) \gets \textsc{$(k,j)$-SubspaceCoreset}(Q,k,j,\gamma, \delta / j^2,v)$  \\
  Set $j \gets j+1$\\
	Set $\Delta_T \gets \Delta^T + \Delta_i^T$\\
	Set $T_i\gets \emptyset$; $\Delta_i^T \gets 0$ \\
  Set $i\gets i+1$ \\
}
Set $T_i\gets T$; $\Delta_i^T \gets \Delta^T$; $v_i \gets v$ \\
Define $S \gets \bigcup_{i=1}^{h} S_i \cup T_i$ and $\Delta^S \gets \sum_{i=1}^h \Delta_i^S + \Delta_i^T$ \\
Define $w$ to be the weight vector corresponding to $S$\\
Set $Q\gets\emptyset$
}
Set $S_h\gets T$; $\Delta_h^S \gets \Delta^T$; $u_h \gets v$\\
}}
\end{algorithm}

\begin{theorem}\label{streaming-proj}
Let $1 >\epsilon >0$.
There exists $h(j,k), \ge 0$ such that on input a stream of $n$ $d$-dimensional points with integer coordinates and maximum $l_2$-norm $M \ge4$, algorithms 
\ref{algprojstream} and \ref{algprojstream2} maintain with probability at least $1-\delta$ in overall time $nd(k\log (Mdn)\log(1/\delta)/\eps)^{O(f(j,k))}$  
a set $S$ of $=(k\log (Mn)\log(1/\delta)/\eps)^{f(j,k)}$ points weighted with a vector $w$ and a real value $\Delta^S$
such that for every set $C$ of $k$ $j$-dimensional subspaces the following inequalities are satisfied:
 $$
 \dist^2(A,C) \le \sum_{i=1}^{|S|} w_i \cdot \dist^2(S_{i*},C) + \Delta^S \le (1+\eps) \cdot \dist^2(A,C),
$$
where $A$ denotes the matrix whose rows are the $n$ input points.
\end{theorem}

\begin{proof}
We first analyze the success probability of algorithms \ref{algprojstream} and \ref{algprojstream2}.
In the $j$th call to a coreset construction during the execution of our algorithms, we apply the above coreset construction with probability of 
failure $\delta/j^2$. After reading $n$ points from the stream, all the coreset constructions will succeed with probability at least
\[
1-\delta\sum_{j=2}^{\infty} \frac{1}{j^2}\geq 1-\delta.
\]
The space bound of $S$ follows from the fact that $h=O(\log n)$ and since $j^2/\delta$ is at most $n^2/\delta$.
Furthermore, we observe that for algorithm \ref{algprojstream2} we can assume that the input has integer coordinates
and maximum norm $M^{h'(j) d n}$ for some function $h'()$ (where we use that we can assume $1/\gamma^2 \le n$ as otherwise
we can simply maintain all the points. 
The running time follows from the fact that the computation time of a coreset of size $(k\log (Mdn)\log(1/\delta)/\eps)^{h(j,k)}$ 
can be done in time $d(k\log (Mdn)\log(1/\delta)/\eps)^{O(h(j,k))}$.

It remains to proof that the resulting sets are a coreset. Here we first observe that at any stage of the algorithm a coreset that corresponding
to a set of $n$ input points can have at most $(1+\epsilon) n$ points. Otherwise, the coreset property would be violated if all centers are sufficiently far
away from the input set. For the analysis, we can replace our weighted input set by unweighted sets (written by a matrix $A$) and apply Corollary \ref{lem:movement:simple} to show that 
$$
|\dist^2(A,C)-\dist^2(A',C)| \le \frac{\gamma}{20} \cdot \dist(A,C)
$$
where $A'$ is the matrix obtained by snapping the rows of $A$ to a grid of side length $\gamma^2 L/(100d)$.
Suppose that all the coreset constructions indeed succeeded (which happens with probability at least $1-\delta$), the error bound follows from Claim
\ref{claim:streaming} in a similar way as in the proof of Theorem \ref{streaming-linear-subspace} by viewing the snapping procedure as an additional
coreset construction (so that we have $2h$ levels instead of $h$).
\end{proof}


\section{Small Coresets for Other Dissimilarity Measures}

In this section, we describe an alternative way to prove the existence of coresets with a size that is independent of the number of input points and the dimension. It has an exponential dependency on $\eps^{-1}$ and thus leads to larger coresets. However, we show that the construction works for a $k$-means variant based on a restricted class of $\mu$-similar Bregman divergences. Bregman divergences are not symmetric, and the $k$-means variant with Bregman divergences is not a $\mathcal{C}$-clustering problem as defined in Definition~\ref{c-clustering-problem}. Thus, the additional construction can solve at least one case that the previous sections do not cover.  

Bregman divergences do share an important property with squared Euclidean distances that we will see later on. It will prove critical for our construction. After we define our clustering problem and coresets for it formally, we will proceed in three steps. 
First, we define two niceness conditions for the dissimilarity measure. Any function that assigns a non-negative dissimilarity to any two points in $\REAL^d$, maps the origin to zero and satisfies the two niceness conditions is a nice dissimilarity measure for us. 
Second, we give a general construction for coresets for $k$-means with a nice dissimilarity measure, and prove its correctness. Third, we show that a restricted class of Bregman divergences satisfies the niceness conditions. 

\subsection{Clustering with dissimilarity \texorpdfstring{$d$}{d}}

Let $d : \REAL^{d}\times \REAL^{d} \to \mathbb{R}^{\ge 0}$ be any dissimilarity measure which satisfies $d(0)=0$. As before, we use abbreviations, in particular, we use $d(p,C)=\inf_{c \in C} d(p,c)$ for any  $p \in \REAL^d$, $C \subset \REAL^d$, and we use $d(A,C)=\sum_{i=1}^n d(A_{i\ast},C)$ for any $A \in \REAL^{n\times d}$, $C \in \REAL^d$.
We denote the \emph{centroid} of any finite set $Q \subset \REAL^d$ by $\mu(Q)$. It is defined as $\mu(Q) = \frac{1}{|Q|}\sum_{x \in Q} x$. We also use this notation for matrices, $\mu(A)$ is the centroid of the points stored in the rows of $A$.

\begin{definition}[Clustering with dissimilarity $d$]\label{d-clustering}
Given $A\in\REAL^{n\times d}$, compute a set $C$ of $k$ centers (points) in $\REAL^d$ such that $d(A,C)$ is minimized.
\end{definition}

We get the standard $k$-means problem if we let $d$ be squared Euclidean distances, i.e., in this case, $d(A,C) = \dist^2(A,C)$. 
We denote the optimal cost by $\opt_k(A) = \min_{C \subset \REAL^d, |C|=k} d(A,C)$. This allows us to also use $\opt_i(A)$ if we want to refer to the cost of clustering the points in $A$ with $i \neq k$ centers during our algorithm and proofs. Notice that for standard $k$-means, $\opt_1(A)$ can be computed in polynomial time since the optimum center is the centroid. This is true for Bregman divergences as well. The following definition is similar to 
Definition~\ref{coresetdef}.

\begin{definition}[Coreset for clustering with dissimilarity $d$]\label{coresetdef-d}
Let $\mathcal C$ be the family of all sets $C \subset \REAL^d$ with $k$ points. 
Let $A\in\REAL^{n\times d}$, $k\geq1$ be an integer, and $\eps>0$. A tuple $(S,\Delta,w)$ of a matrix $S\in\REAL^{m\times d}$ with a vector of $n$ weights $w=(w_1,\ldots,w_m)\in \REAL^m$ associated with its rows and a value $\Delta = \Delta(A, \epsilon,\mathcal C)$
 is an \emph{$\epsilon$-coreset for the clustering problem with dissimilarity $d$} if for every $C\in \mathcal C$ we have
\[
(1-\epsilon) d(A,C) \le \sum_{i=1}^m w_i d(S_{i*},C) + \Delta \le (1+\epsilon) \cdot d(A,C).
\]
\end{definition}

\subsection{Clustering problems with nice dissimilarity measures}
We say that a dissimilarity $d$ is \emph{nice} if the clustering problem that it induces satisfies the following two conditions. Firstly, if we have an $A$ where the best clustering with $k$ clusters is not much cheaper than the cost of $A$ with only one center, then this has to induce a coreset for $A$. We imagine this as $A$ being pseudo random; since it has so little structure, representing with fewer points is easy.
Secondly, if a subset $A' \subset A$ has negligible cost compared to $A$, then it is possible to compute a small weighted set which approximates the cost of $A'$ up to an additive error which is an $\epsilon$-fraction of the cost of $A$. Note that this is a much easier task than computing a coreset for $A$, since $A'$ may be represented by a set with a much higher error then its own cost. The following definition states our requirements in more detail. If we say that $A_1,\ldots,A_k$ is a \emph{partitioning} of $A$, we mean that the rows of $A$ are partitioned into $k$ sets which then induce $k$ matrices with $d$ columns. By $A'\subset A$ we mean that the rows of $A'$ are a subset of the rows of $A$, and by $|A|$ we mean the number of rows in $A$.

\begin{definition}\label{def:nicenessford}
We say that a dissimilarity measure $d$ is \emph{nice} if the clustering problem with dissimilarity $d$ (see Definition~\ref{d-clustering}) satisfies the following conditions.
\begin{enumerate}\setcounter{enumi}{-1}
\item It is possible to compute $\opt_1(A)$ in polynomial time for any $A \in \REAL^{n\times d}$.
\item \label{firstcondition} If an optimal $k$-clustering of $A$ is at most a $(1+\eps)$-factor cheaper than the best $1$-clustering, then this must induce a coreset for $A$:

If $\opt_1(A) \le (1+f_1(\eps)) \sum_{i=1}^k \cost(A_i)$ for all partitionings $A_1,\ldots,A_k$ of $A$ into $k$ matrices, 
 then there exists a coreset $(Z, \Delta_Z)$ of size $g(k,\eps)$ such that for any set of $k$ centers we have $\left| d(A,C) - d(Z,C) + \Delta_Z\right| \le \eps\cdot d(A,C)$, for a function $g$ which only depends on $k$ and $\eps$, and a function $f_1$ that only depends on $\epsilon$.
\item If the cost of $A'\subset A$ is very small, then it can be represented by a small set which has error $\eps \cdot d(A,C)$ for any $C, |C|=k$:

If $opt_k(A',f_2(k)) \le f_3(\eps) \opt(A,k)$ for $A' \subset A$, then there exist a set $Z$ of size $h(f_2(k),\eps)$ and a constant $\Delta_Z$ such that for any set of centers $C$ we have $\left| d(A',C) - d(A,C)+\Delta_Z\right| \le \eps \cdot d(A,C)$. \label{secondcondition}
\end{enumerate}
\end{definition}

\subsection{Algorithm for nice dissimilarity measures}

In the following, we will assume that we can solve the clustering problem optimally. This is only for simplicity of exposition; the algorithm still works if we use an approximation algorithm. Algorithms~\ref{alg-d-partition} and~\ref{alg-coresets-d-clustering} give pseudo code for the algorithm. Algorithm~\ref{alg-d-partition} is a recursive algorithm that partitions $A$ into subsets. Every subset $A'$ in the partitioning is either very cheap (defined more precisely below), or pseudo random, meaning that $\opt_1(A') \le (1+f_1(\eps)) \opt_k(A')$. This is achieved by a recursive partitioning. The trick is that whenever a set is not pseudo random, then the overall cost is decreased by a factor of $(1+f_1(\eps))$ by the next partitioning step. This means that after sufficiently many ($\lceil\log_{1+f_1(\eps)} \frac{1}{f_3(\eps)}\rceil$) levels, all sets have to be cheap. Indeed, not only are the individual sets cheap, even the sum of all their $1$-clustering costs is cheap.

\newcommand{\dpartition}{\textsc{partition-helper}}
\begin{algorithm}[h!t]
    \caption{$\dpartition(A,k,\eps)$}\label{alg-d-partition}
{\begin{tabbing}
\textbf{Input: \quad } \=$A\in\REAL^{n\times d}$, integers $k,t,\nu \geq1$, an error parameter $\eps>0$, and a set of sets $M$\\
\textbf{Output:  } \>A partitioning of $A$
\end{tabbing}}
\vspace{-0.3cm}
   Compute an optimal solution $C^\ast=\{c_1,\ldots,c_k\}$ for A.\\
	 Let $A_1,\ldots,A_k$ be the partitioning induced by $C^\ast$.\\
	 \uIf{ $t \le \nu$ or $opt_1(A) \le (1+f_1(\eps)) \sum_{i=1}^k \opt_1(A_i)$}{
	  $M = M \cup \{P\}$\\		
	 }
	 \Else
	 {
	  \For {$i=1,\ldots,k$}
		{
		  $M := M \cup \dpartition(A_i,k,t+1,\nu,\epsilon')$\\
		}		
	 }
	 \Return M\\
\end{algorithm}

\newcommand{\nicedalgo}{\textsc{coresets-for-nice-d-clustering-problems}}
\begin{algorithm}[h!t]
    \caption{$\nicedalgo(A,k,\eps)$}\label{alg-coresets-d-clustering}
{\begin{tabbing}
\textbf{Input: \quad } \=$A\in\REAL^{n\times d}$, an integer $k\geq1$ and an error parameter $\eps>0$.\\
\textbf{Output:  } \>A tuple $(S,\Delta,w)$ that satisfies Definition~\ref{coresetdef-d}.
\end{tabbing}}
\vspace{-0.3cm}
 Set $\epsilon' = \epsilon^2/50$, $S = \emptyset, w = \emptyset$ and $\Delta=0$\\
 M = \dpartition$(A,k,0,\lceil\log_{1+f_1(\eps)} \frac{1}{f_3(\eps)}\rceil, \epsilon')$\\
 \For{all $A' \in M$}{
   Compute $S_{A'},w_{A'},\Delta_{A'}$ by the routines guaranteed in Definition~\ref{def:nicenessford}\\
	Set $S = S \cup S_{A'}, w = w||w_{A'}, \Delta = \Delta+\Delta_{A'}$\\
 }
 \Return {$S, w, \Delta$}
\end{algorithm}

Let $M_i$ denote the set of all subsets generated by the algorithm on level $\nu$ (where the initial call is level $0$, and where not all sets in $M_i$ end up in $M$ since some of them are further subdivided). The input set has cost $\opt_k(A) = \opt_k(A) / (1+f_1(\eps))^0$. For every level in the algorithm, the overall cost is decreased by a factor of $(1+f_1(\eps))$. Thus, the sum of all $1$-clustering costs of sets in $M_i$ is $\opt_k(A)/(1+f_1(\eps))^i$. For $\nu = \lceil\log_{1+f_1(\eps)} \frac{1}{f_3(\eps)}\rceil$, this is smaller than $f_3(\eps)\cdot\opt_k(A)$. We have at most $f(k) := k^{\nu}$ sets that survive until level $\nu$ of the recursion, and then their overall cost is bounded by $\opt_1(A)$. By Condition \ref{secondcondition}, this implies the existence of a set $Z$ of size $h(k^\nu, \eps)$ which has an error of at most $\eps \opt_k(A)$.

For all sets where we stop early (the pseudo random sets), Condition \ref{firstcondition} directly gives a coreset of size $g(k,\eps)$.
The union of these coresets give a coreset for the union of all pseudo random sets. Altogether, they induce an error of less than $\eps \opt_k(A)$. Together with the $\eps \opt_k(A)$ error induced by the cheap sets on level $\nu$, this gives a total error of $2 \eps \opt_k(A)$. So, if we start every thing with $\eps/2$, we get a coreset for $A$ with error $\eps \opt_k(A)$.
The size of the coreset is  $k^\nu \cdot g(k,\eps/2) + h(k^\nu, \eps/2)$.

\begin{lemma}\label{groessenlemma}
If $d$ is a nice dissimilarity measure according to Definition~\ref{def:nicenessford}, then  there exists a coreset of size $k^\nu \cdot g(k,\eps/2) + h(k^\nu, \eps/2)$ for $\nu = \lceil\log_{1+f_1(\eps/2)} \frac{1}{f_3(\eps/2)}\rceil$ for the clustering problem with dissimilarity $d$.
\end{lemma}

For $k$-means, we can achieve that $g\equiv 1$ and $h(k^\nu,\eps)=k^\nu$. Thus, the overall coreset size is $2 k^{\log_{1+f_1(\eps)} \frac{1}{f_3(\eps)}}$. We do not present this in detail as the coreset is larger than the $k$-means coreset coming from our first construction. However, the proof can be deduced from the following proof for a restricted class of $\mu$-similar Bregman divergences, as the $k$-means case is easier.

\subsection{Coresets for \texorpdfstring{$\mu$}{mu}-similar Bregman divergences}

\newcommand{\dd}{d_\phi}
\newcommand{\da}{d_B}

Let $d_\phi : S\times S \to \mathbb{R}$ be a $m$-similar Bregman divergence. This means that $d_\phi$ is defined on a convex set $S \subset \mathbb{R}^d$ and there exists a Mahalanobis distance $\da$ such that $m \da(p,q) \le d_\phi(p,q) \le \da(p,q)$ for all points $p,q \in \REAL^d$ and an $m \in (0,1]$ (note that we use $m$-similar instead of $\mu$-similar in order to prevent confusion with the centroid $\mu$).

We say that $S$ is $A$-covering if it contains the union of all balls of radius $(4/m\eps)\cdot d(p,q)$ for all $p, q \in A$. For our proof, we need that $S$ is convex and $A$-covering.
Because of this additional restriction, our setting is much more restricted than in~\cite{AB09}. It is an interesting open question how to remove this restriction and also how to relax the $m$-similarity.

The fact that $\da$ is a Mahalanobis distance means that there exists a regular matrix $B$ with $\da(x,y) = \norm{B(x-y)}^2$ for all points $x,y \in \REAL^n$. In particular, $m \cdot \norm{B(x-y)}^2 \le \dd(x,y) \le \norm{B(x-y)}^2$.
By \cite{BMDG05}, Bregman divergences (also if they are not $m$-similar) satisfy the Bregman version of a famous equality that is also true for $k$-means. For Bregman divergences, it reads: For all $A \in \REAL^{n\times d}$, it holds that \begin{align}\sum_{p \in A} d_\phi(p,z) = \sum_{p \in A} d_\phi(p,\mu) + |A| \cdot d_\phi(\mu,z).\label{magicformulabreg}\end{align} 
\paragraph{Condition 1}
To show that Condition 1 holds, we set $f_1(\eps) = \frac{1}{(1+\frac{4}{m\cdot\eps})^2}$ and assume that we are given a point set $S$ that is pseudo random. This means that it satisfies  for any partitioning of $S$ into $k$ subsets $S_1,\ldots, S_k$ that
\begin{align*}
&& \sum_{s \in S} \dd(s,\mu(S)) &\le& (1 + f(\eps)) \sum_{j=1}^k \sum_{x \in S_j} \dd(x,\mu(S_j))\\
\Leftrightarrow && \sum_{j=1}^k  \sum_{s \in S} \dd(s,\mu(S_j)) +\sum_{j=1}^k |S_j| \dd(\mu(S_j),\mu(S)) &\le& (1 + f(\eps)) \sum_{j=1}^k \sum_{x \in S_j} \dd(x,\mu(S_j)) \\
 \label{meinLabel} \Leftrightarrow && \sum_{j=1}^k |S_j| \dd(\mu(S_j),\mu(S)) &\le& \frac{1}{(1+\frac{4}{m\cdot\eps})^2}  \sum_{j=1}^k \sum_{x \in S_j} \dd(x,\mu(S_j)) .
\end{align*}

We show that this restricts the error of clustering all points in $S$ with the same center, more specifically, with the center $c(\mu(S))$, the center closest to $\mu(S)$. To do so, we virtually add  points to $S$. For every $j=1,\ldots,k$, we add one point with weight $\frac{1}{4}\eps \cdot m \cdot |S_j|$ with coordinate $\mu(S)+\frac{4}{m\cdot \eps}\left(\mu(S)-\mu(S_j)\right)$ to $S_j$.
Notice that $\da$ is defined on these points because we assumed that $S$ is $A$-covering.
The additional point shifts the centroid of $S_j$ to $\mu(S)$ because
\begin{align*}
&& \frac{|S_j| \cdot \mu(S_j) + \frac{\eps m}{4}  |S_j| \left[\mu(S)+\frac{4}{m\cdot \eps}\left(\mu(S)-\mu(S_j)\right)\right]}{(1+\frac{m\cdot\eps}{4})|S_j|}\\
&=&
\frac{\frac{\eps m}{4}  |S_j| \left[\mu(S)+\frac{4}{m\cdot \eps}\mu(S)\right]}{(1+\frac{m\cdot\eps}{4})|S_j|}
 = \mu(S).
\end{align*}
We name the set consisting of $S_j$ together with the weighted added point $S_j'$ and the union of all $S_j'$ is $S'$.
Now, clustering $S'$ with center $c(\mu(S))$ is certainly an upper bound for the clustering cost of $S$ with $c(\mu(S))$. Additionally, when clustering $S_j'$ with only one center, then $c(\mu(S))$ is optimal, so clustering $S_j'$ with $c(\mu(S_j))$ can only be more expensive. Thus, clustering all $S_j'$ with the centers $c(\mu(S_j))$ gives an upper bound on the cost of clustering $S$ with $c(\mu(S))$. So, to complete the proof, we have to upper bound the cost of clustering all $S_j'$ with the respective centers $c(\mu(S_j))$. We do this by bounding the additional cost of clustering the added points with $c(\mu(S_j))$, which is
\begin{align*}
& \sum_{j=1}^k \frac{\eps m }{4} |S_j| \cdot \dd\left(\mu(S)+\frac{4}{m\cdot \eps}(\mu(S)\right. \\
& \hspace*{4cm} -\mu(S_j)),c(\mu(S_j))\bigg) \\
\le & \sum_{j=1}^k \frac{ \eps m}{4} |S_j| \cdot \norm{B(\mu(S)+\frac{4}{m\cdot \eps}(\mu(S)-\mu(S_j))-c(\mu(S_j)))}^2 \\
= & \norm{a}^2
\end{align*}
for the $k$-dimensional vector $a$ defined by \[a_j := \sqrt{\eps m|S_j|/4} \norm{B(\mu(S)+\frac{4}{m\cdot \eps}\left(\mu(S)-\mu(S_j)\right)-c(\mu(S_j)))}.\]
 By the triangle inequality, 
\[a_j \le \sqrt{\eps m|S_j|/4}\norm{B((1+\frac{4}{m \eps})(\mu(S)-\mu(S_j)))} + \sqrt{\eps m|S_j|/4} \norm{B(\mu(S_j)-c(\mu(S_j)))} = b_j+d_j\] with $b_j= \sqrt{\eps m|S_j|/4}\norm{B((1+\frac{4}{m \eps})(\mu(S)-\mu(S_j)))} $ and $d_j = \sqrt{\eps m|S_j|/4} \norm{B(\mu(S_j)-c(\mu(S_j)))}$.
Then,
\[
\norm{a} \le \norm{b+d} \le \norm{b} + \norm{d},
\]
where we use the triangle inequality again for the second inequality. Now we observe that
\begin{align*}
\norm{b}^2 =& \sum_{j=1}^k \frac{\eps m}{4} |S_j| \norm{B((1+\frac{4}{\eps m})(\mu(S)-\mu(S_j)))}^2 \\
= & \frac{\eps m}{4} \sum_{j=1}^k |S_j| (1+\frac{4}{m \eps})^2 \norm{B(\mu(S_j)-\mu(S))}^2\\
\le &  \frac{\eps m}{4} (1+\frac{4}{m \eps})^2 \sum_{j=1}^k |S_j| \frac{1}{m}  \dd(\mu(S_j),\mu(S))^2\\
\le &
\frac{\eps}{4} \sum_{j=1}^k \sum_{x \in S_j} \dd(x,\mu(S_j)).
\end{align*}
Additionally, by the definition of $m$-similarity and by Equation \eqref{magicformulabreg} it holds that
\begin{align*}
\norm{d}^2 = & \sum_{j=1}^k \frac{1}{4} \eps m |S_j| \norm{B(\mu(S_j)-c(\mu(S_j)))}^2\\
\le & \frac{\eps}{4} \sum_{j=1}^k |S_j| \dd(\mu(S_j),c(\mu(S_j)))\\
\le & \frac{\eps}{4}  \sum_{j=1}^k \sum_{x \in S_j} \dd(x,\mu(S_j)).
\end{align*}
This implies that $\norm{a} \le \norm{b}+\norm{d} \le 2 \sqrt{\eps}/2 \sqrt{ \sum_{j=1}^k \sum_{x \in S_j} \dd(x,\mu(S_j))}$
and thus \[\norm{a}^2 \le \eps \sum_{j=1}^k \sum_{x \in S_j} \dd(x,\mu(S_j)).\]
This means that Condition 1 holds: If a $k$-clustering of $S$ is not much cheaper than a $1$-clustering, then assigning all points in $S$ to the same center yields a $(1+\eps)$-approximation for arbitrary center sets. This means that we can represent $S$ by $\mu(S)$, with weight $w(S)$ and $\Delta_S=d(S,\mu(S))$. Since we only need one point for this, we even get that $g(k,f'(\eps^{-1}))\equiv 1$.

\paragraph{Condition 2} For the second condition, assume that $\mathcal{S}$ is a set of subsets of $A$ representing the $f_2(k)$ subsets according to an optimal $f_2(k)$-clustering. Let a set $C$ of $k$ centers be given, and define the partitioning $S_1,\ldots, S_k$ for every $S \in \mathcal{S}$ according to $C$ as above.
By Equation~\eqref{magicformulabreg} and by the precondition of Condition 2,
\begin{align*}
& \sum_{S \in \mathcal{S}} \sum_{j=1}^k |S_j| \dd(\mu(S_j),\mu(S)) \\
= & \sum_{S \in \mathcal{S}} \sum_{j=1}^k \sum_{x \in S_j} \dd(x,\mu(S)) -
\sum_{S \in \mathcal{S}} \sum_{j=1}^k \sum_{x \in S_j} \dd(x,\mu(S_j)) \\
\le& f_3(\eps) \cdot \opt_k(A).
\end{align*}
We use the same technique as in the proof that Condition 1 holds. There are two changes: First, there are $|\mathcal{S}|$ sets where the centroids of the subsets must be moved to the centroid of the specific $S$ (where in the above proof, we only had one set $S$). Second, the bound depends on $\opt_k(A)$ instead of $\sum_{S \in \mathcal{S}}$, so the approximation is dependent on $\opt_k(A)$ as well, but this is consistent with the statement in Condition 2.

We  set $f_3(\eps) = f_1(\eps)$ and again virtually add points. For each $S \in \mathcal{S}$ and each subset $S_j$ of $S$, we add a point with weight $\frac{m \cdot \eps}{4} |S_j|$  and coordinate $\mu(S) + \frac{4}{m \cdot \eps} (\mu-\mu_j)$ to $S_j$. Notice that these points lie within the convex set $A$ that $\da$ is defined on because we assumed that $S$ is $A$-covering.

We name the new sets $S_j'$, $S'$ and $\mathcal{S}'$.
Notice that the centroid of $S_j'$ is now
\begin{align*}
&\frac{|S_j| \cdot \mu(S_j) + \frac{\eps m}{4}  |S_j| \left[\mu(S)+\frac{4}{m\cdot \eps}\left(\mu(S)-\mu(S_j)\right)\right]}{(1+\frac{m\cdot\eps}{4})|S_j|}\\
= &\mu(S)
\end{align*} in all cases.
Again, clustering $S'$ with $c(\mu(S))$ is an upper bound for the clustering cost of $S$ with $c(\mu(S))$, and because the centroid of $S_j'$ is $\mu(S)$, clustering every $S_j'$ with $c(\mu(S_j))$ is an upper bound on clustering $S$ with $c(\mu(S))$. Finally, we have to upper bound the cost of clustering all $S_j'$ in all $S$ with $c(\mu(S_j))$, which we again do by bounding the additional cost incurred by the added points. Adding this cost over all $S$ yields
\begin{align*}
& \sum_{S \in \mathcal{S}} \sum_{j=1}^k \frac{1}{4} \eps m |S_j| \cdot \dd(\mu(S) \\
 & \quad + \frac{4}{m\cdot \eps}\left(\mu(S)-\mu(S_j)\right),c(\mu(S_j))) \\
\le & \sum_{S \in \mathcal{S}} \sum_{j=1}^k \frac{ \eps m}{4} |S_j| \cdot \norm{B(\mu(S) + \frac{4}{m\cdot \eps}\left(\mu(S)-\mu(S_j)\right)-c(\mu(S_j)))}^2=  \norm{a}^2.
\end{align*}
For the last equality, we define $|\mathcal{S}|$ vectors $a^S$ by
\[a_j^S := \sqrt{\eps m|S_j|/4} \norm{B(\mu(S)+\frac{4}{m\cdot \eps}\left(\mu(S)-\mu(S_j)\right)-c(\mu(S_j)))}\] and concatenate them in arbitrary but fixed order to get a $k \cdot |\mathcal{S}|$ dimensional vector $a$.
By the triangle inequality, 
\[a_j^S \le \sqrt{\eps m|S_j|/4}\norm{B((1+\frac{4}{m \eps})(\mu(S)-\mu(S_j)))} + \sqrt{\eps m|S_j|/4} \norm{B(\mu(S_j)-c(\mu(S_j)))} = b_j^S+d_j^S
\]
 with $b_j^S= \sqrt{\eps m|S_j|/4}\norm{B((1+\frac{4}{m \eps})(\mu(S)-\mu(S_j)))} $ and $d_j^S = \sqrt{\eps m|S_j|/4} \norm{B(\mu(S_j)-c(\mu(S_j)))}$. Define $b$ and $d$ by concatenating the vectors $b^S$ and $d^S$, respectively, in the same order as used for $a$.
Then we can again conclude that \[
\norm{a} \le \norm{b+d} \le \norm{b} + \norm{d},
\]
where we use the triangle inequality for the second inequality.
 Now we observe that
\begin{align*}
\norm{b}^2 =& \sum_{S \in \mathcal{S}} \sum_{j=1}^k \frac{\eps m}{4} |S_j| \norm{B((1+\frac{4}{\eps m})(\mu(S)-\mu(S_j)))}^2\\
=&  \frac{\eps m}{4} \sum_{S \in \mathcal{S}} \sum_{j=1}^k |S_j| (1+\frac{4}{m \eps})^2 \norm{B(\mu(S_j)-\mu(S))}^2\\
\le &  \frac{\eps m}{4} (1+\frac{4}{m \eps})^2 \sum_{S \in \mathcal{S}} \sum_{j=1}^k |S_j| \frac{1}{m}  \dd(\mu(S_j),\mu(S))^2\\
\le &
\frac{\eps}{4} \opt_k(A).
\end{align*}

Additionally, by the definition of $m$-similarity and by Equation \eqref{magicformulabreg} it holds that
\begin{align*}
\norm{d}^2 = & \sum_{S \in \mathcal{S}}\sum_{j=1}^k \frac{1}{4} \eps m |S_j| \norm{B(\mu(S_j)-c(\mu(S_j)))}^2\\
\le & \frac{\eps}{4} \sum_{S \in \mathcal{S}}\sum_{j=1}^k |S_j| \dd(\mu(S_j),c(\mu(S_j)))\\
\le & \frac{\eps}{4} \sum_{S \in \mathcal{S}} \sum_{j=1}^k \sum_{x \in S_j} \dd(x,\mu(S_j)).
\end{align*}
This implies that $\norm{a} \le \norm{b}+\norm{d} \le 2 \sqrt{\eps}/2 \sqrt{ \opt_k(A)}$
and thus \[\norm{a}^2 \le \eps \opt_k(A).\]

\begin{theorem}
If $\da : S\times S \to \mathbb{R}$ is a $m$-similar Bregman divergence on a convex and $A$-covering set $S$ with $m \in (0,1]$, then
there exists a coreset consisting of clustering features of constant size, \ie\ the size only depends on $k$ and $\eps$.
\end{theorem}
\begin{proof}
We have seen that the two conditions hold with $f_1(\eps)=f_3(\eps)=\frac{1}{(1+\frac{4}{m\cdot\eps})^2}$, and $g \equiv 1$ and $h(k^\nu,\eps)=k^\nu$.
By Lemma~\ref{groessenlemma}, this implies that we get a coreset, and that the size of this coreset is bounded by
\begin{align*}
2 k^\nu &= 2 k^{\lceil\log_{1+f_1(\eps/2)} \frac{1}{f_3(\eps/2)}\rceil}\\
&= 2  k^{\lceil\log_{1+\frac{1}{(1+\frac{8}{m\cdot\eps})^2}} {(1+\frac{8}{m\cdot\eps})^2}\rceil}\hspace*{3cm}
\end{align*}
\end{proof}


\bibliographystyle{amsalpha}
\bibliography{newbibfile}

\newcommand{\etalchar}[1]{$^{#1}$}
\providecommand{\bysame}{\leavevmode\hbox to3em{\hrulefill}\thinspace}
\providecommand{\MR}{\relax\ifhmode\unskip\space\fi MR }
\providecommand{\MRhref}[2]{%
  \href{http://www.ams.org/mathscinet-getitem?mr=#1}{#2}
}
\providecommand{\href}[2]{#2}
\begin{thebibliography}{AHPV04b}

\bibitem[AB09]{AB09}
M.~Ackermann and J.~Bl\"omer, \emph{Coresets and approximate clustering for
  bregman divergences}, \Proc 20th \SODA, 2009, pp.~1088--1097.

\bibitem[ACKS15]{ACKS15}
Pranjal Awasthi, Moses Charikar, Ravishankar Krishnaswamy, and Ali~Kemal Sinop,
  \emph{The hardness of approximation of euclidean k-means}, 31st SoCG, 2015,
  pp.~754--767.

\bibitem[ADHP09]{ADHP09}
Daniel Aloise, Amit Deshpande, Pierre Hansen, and Preyas Popat,
  \emph{{NP}-hardness of {E}uclidean sum-of-squares clustering}, Machine
  Learning \textbf{75} (2009), no.~2, 245 -- 248.

\bibitem[ADK09]{ADK09}
Ankit Aggarwal, Amit Deshpande, and Ravi Kannan, \emph{Adaptive sampling for
  $k$-means clustering}, \Proc 12th \APPROX, 2009, pp.~15--28.

\bibitem[AFZZ15]{aghamolaei2015diversity}
Sepideh Aghamolaei, Majid Farhadi, and Hamid Zarrabi-Zadeh, \emph{Diversity
  maximization via composable coresets}, Proceedings of the 27th Canadian
  Conference on Computational Geometry, 2015.

\bibitem[AHPV04a]{AHV04}
P.~Agarwal, S.~Har-Peled, and K.~Varadarajan, \emph{Approximating extent
  measures of points}, Journal of the ACM \textbf{51} (2004), no.~4, 606--635.

\bibitem[AHPV04b]{AHPV04}
Pankaj~K. Agarwal, Sariel Har-Peled, and Kasturi~R. Varadarajan,
  \emph{Approximating extent measures of points}, Journal of the ACM
  \textbf{51} (2004), no.~4, 606 -- 635.

\bibitem[AMR{\etalchar{+}}12]{AMRSLS12}
Marcel~R. Ackermann, Marcus M{\"a}rtens, Christoph Raupach, Kamil Swierkot,
  Christiane Lammersen, and Christian Sohler, \emph{Streamkm\syms: A clustering
  algorithm for data streams}, ACM Journal of Experimental Algorithmics
  \textbf{17} (2012), article 2.4, 1--30.

\bibitem[ANSW16]{ANSW16}
Sara Ahmadian, Ashkan Norouzi{-}Fard, Ola Svensson, and Justin Ward,
  \emph{Better guarantees for k-means and euclidean k-median by primal-dual
  algorithms}, CoRR \textbf{abs/1612.07925} (2016).

\bibitem[AV07]{arthur2007k}
David Arthur and Sergei Vassilvitskii, \emph{k-means++: The advantages of
  careful seeding}, Proceedings of the eighteenth annual ACM-SIAM symposium on
  Discrete algorithms, Society for Industrial and Applied Mathematics, 2007,
  pp.~1027--1035.

\bibitem[BEHW89]{BEHW89}
Anselm Blumer, Andrzej Ehrenfeucht, David Haussler, and Manfred~K. Warmuth,
  \emph{Learnability and the vapnik-chervonenkis dimension}, Journal of the
  {ACM} \textbf{36} (1989), no.~4, 929--965.

\bibitem[Bey]{Beyer}
M.~Beyer, \emph{Gartner says solving \lq big data\rq\ challenge involves more
  than just managing volumes of data},
  \url{http://www.gartner.com/it/page.jsp?id=1731916}, Gartner. Retrieved 13
  July 2011.

\bibitem[BFL16]{braverman2016new}
Vladimir Braverman, Dan Feldman, and Harry Lang, \emph{New frameworks for
  offline and streaming coreset constructions}, arXiv preprint arXiv:1612.00889
  (2016).

\bibitem[BFL{\etalchar{+}}17]{BFLSY17}
Vladimir Braverman, Gereon Frahling, Harry Lang, Christian Sohler, and Lin~F.
  Yang, \emph{Clustering high dimensional dynamic data streams}, Proceedings of
  the 34th International Conference on Machine Learning, {ICML} 2017, Sydney,
  NSW, Australia, 6-11 August 2017, 2017, pp.~576--585.

\bibitem[BMD09]{BMD09}
Christos Boutsidis, Michael~W. Mahoney, and Petros Drineas, \emph{Unsupervised
  feature selection for the \$k\$-means clustering problem}, \Proc 23rd \NIPS,
  2009, pp.~153 -- 161.

\bibitem[BMDG05]{BMDG05}
A.~Banerjee, S.~Merugu, I.~S. Dhillon, and J.~Ghosh, \emph{Clustering with
  bregman divergences}, Journal of Machine Learning Research \textbf{6} (2005),
  1705--1749.

\bibitem[BS80]{bent}
J.~L. Bentley and J.~B. Saxe, \emph{Decomposable searching problems i.
  static-to-dynamic transformation}, Journal of Algorithms \textbf{1} (1980),
  no.~4, 301--358.

\bibitem[BSS12]{batson2012twice}
Joshua Batson, Daniel~A Spielman, and Nikhil Srivastava, \emph{Twice-ramanujan
  sparsifiers}, SIAM Journal on Computing \textbf{41} (2012), no.~6,
  1704--1721.

\bibitem[BZD10]{BZD10}
Christos Boutsidis, Anastasios Zouzias, and Petros Drineas, \emph{{ Random
  Projections for {$k$}-means Clustering }}, \Proc 24th \NIPS, 2010, pp.~298 --
  306.

\bibitem[BZMD15]{BZMD15}
Christos Boutsidis, Anastasios Zouzias, Michael~W. Mahoney, and Petros Drineas,
  \emph{Randomized dimensionality reduction for k-means clustering}, {IEEE}
  Transactions on Information Theory \textbf{61} (2015), no.~2, 1045--1062.

\bibitem[CEM{\etalchar{+}}15]{CEMMP15}
Michael~B. Cohen, Sam Elder, Cameron Musco, Christopher Musco, and Madalina
  Persu, \emph{Dimensionality reduction for k-means clustering and low rank
  approximation}, Proceedings of the Forty-Seventh Annual {ACM} on Symposium on
  Theory of Computing, {STOC} 2015, 2015, pp.~163--172.

\bibitem[Che09]{C09}
Ke~Chen, \emph{On coresets for k-median and k-means clustering in metric and
  euclidean spaces and their applications}, SIAM Journal on Computing
  \textbf{39} (2009), no.~3, 923 -- 947.

\bibitem[CKM16]{CKM16}
Vincent Cohen{-}Addad, Philip~N. Klein, and Claire Mathieu, \emph{Local search
  yields approximation schemes for k-means and k-median in euclidean and
  minor-free metrics}, {IEEE} 57th Annual Symposium on Foundations of Computer
  Science, {FOCS}, 2016, pp.~353--364.

\bibitem[CNW16]{CohenNW16}
Michael~B. Cohen, Jelani Nelson, and David~P. Woodruff, \emph{Optimal
  approximate matrix product in terms of stable rank}, 43rd International
  Colloquium on Automata, Languages, and Programming, {ICALP} 2016, July 11-15,
  2016, Rome, Italy (Ioannis Chatzigiannakis, Michael Mitzenmacher, Yuval
  Rabani, and Davide Sangiorgi, eds.), LIPIcs, vol.~55, Schloss Dagstuhl -
  Leibniz-Zentrum fuer Informatik, 2016, pp.~11:1--11:14.

\bibitem[CW09]{CW09}
Kenneth~L. Clarkson and David~P. Woodruff, \emph{Numerical linear algebra in
  the streaming model}, Proceedings of the 41st STOC, 2009, pp.~205 -- 214.

\bibitem[CW13]{CW13}
\bysame, \emph{Low rank approximation and regression in input sparsity time},
  STOC 2013, 2013, pp.~81 -- 90.

\bibitem[DFK{\etalchar{+}}04]{DFKVV04}
Petros Drineas, Alan~M. Frieze, Ravi Kannan, Santosh Vempala, and V.~Vinay,
  \emph{Clustering large graphs via the singular value decomposition}, Machine
  Learning \textbf{56} (2004), 9--33.

\bibitem[DR10]{DR10}
Amit Deshpande and Luis Rademacher, \emph{Efficient volume sampling for
  row/column subset selection}, 51th FOCS, 2010, pp.~329 -- 338.

\bibitem[DRVW06]{deshpande2006matrix}
Amit Deshpande, Luis Rademacher, Santosh Vempala, and Grant Wang, \emph{Matrix
  approximation and projective clustering via volume sampling}, Theory of
  Computing \textbf{2} (2006), no.~1, 225 -- 247.

\bibitem[DTV11]{DTV11}
Amit Deshpande, Madhur Tulsiani, and Nisheeth~K. Vishnoi, \emph{Algorithms and
  hardness for subspace approximation}, \Proc 22nd \SODA, 2011, pp.~482--496.

\bibitem[DV06]{DV06}
Amit Deshpande and Santosh Vempala, \emph{Adaptive sampling and fast low-rank
  matrix approximation}, 10th RANDOM, 2006, pp.~292--303.

\bibitem[EA07]{eisenstat2007vc}
David Eisenstat and Dana Angluin, \emph{The vc dimension of k-fold union},
  Information Processing Letters \textbf{101} (2007), no.~5, 181--184.

\bibitem[EV05]{EV05}
Michael Edwards and Kasturi~R. Varadarajan, \emph{No coreset, no cry: {II}},
  \Proc 25th \FSTTCS, 2005, pp.~107--115.

\bibitem[FFS06]{FFS06}
D.~Feldman, A.~Fiat, and M.~Sharir, \emph{Coresets forweighted facilities and
  their applications}, \Proc 47th \FOCS, 2006, pp.~315--324.

\bibitem[FGS{\etalchar{+}}13]{FGSSS13}
Hendrik Fichtenberger, Marc Gill{\'e}, Melanie Schmidt, Chris Schwiegelshohn,
  and Christian Sohler, \emph{{BICO: BIRCH Meets Coresets for k-Means
  Clustering }}, \Proc 21st \ESA, 2013, pp.~481--492.

\bibitem[FL11]{FL11}
D.~Feldman and M.~Langberg, \emph{A unified framework for approximating and
  clustering data}, \Proc 43rd \STOC, 2011, See http://arxiv.org/abs/1106.1379
  for fuller version, pp.~569--578.

\bibitem[FMS07]{FMS07}
D.~Feldman, M.~Monemizadeh, and C.~Sohler, \emph{A ptas for k-means clustering
  based on weak coresets}, \Proc 23rd \SoCG, 2007, pp.~11--18.

\bibitem[FMSW10]{feldman2010coresets}
Dan Feldman, Morteza Monemizadeh, Christian Sohler, and David~P Woodruff,
  \emph{Coresets and sketches for high dimensional subspace approximation
  problems}, Proceedings of the twenty-first annual ACM-SIAM symposium on
  Discrete Algorithms, Society for Industrial and Applied Mathematics, 2010,
  pp.~630--649.

\bibitem[FRS16]{FRS16}
Zachary Friggstad, Mohsen Rezapour, and Mohammad~R. Salavatipour, \emph{Local
  search yields a {PTAS} for k-means in doubling metrics}, 57th FOCS, 2016,
  pp.~365--374.

\bibitem[FS05]{FS05}
Gereon Frahling and Christian Sohler, \emph{Coresets in dynamic geometric data
  streams}, \Proc 37th \STOC, 2005, pp.~209 -- 217.

\bibitem[FS12]{schulman}
Dan Feldman and Leonard~J Schulman, \emph{Data reduction for weighted and
  outlier-resistant clustering}, Proc. of the 23rd annual ACM-SIAM symp. on
  Discrete Algorithms (SODA), SIAM, 2012, pp.~1343--1354.

\bibitem[FSS13]{FSS13}
Dan Feldman, Melanie Schmidt, and Christian Sohler, \emph{{Turning Big Data
  into Tiny Data: Constant-size Coresets for k-means, PCA and Projective
  Clustering}}, \Proc 24th \SODA, 2013, pp.~1434 -- 1453.

\bibitem[FVR15]{FeldmanVR15}
Dan Feldman, Mikhail Volkov, and Daniela Rus, \emph{Dimensionality reduction of
  massive sparse datasets using coresets}, CoRR \textbf{abs/1503.01663} (2015).

\bibitem[FVR16]{feldman2016dimensionality}
Dan Feldman, Mikhail Volkov, and Daniela Rus, \emph{Dimensionality reduction of
  massive sparse datasets using coresets}, Advances in Neural Information
  Processing Systems, 2016, pp.~2766--2774.

\bibitem[GK65]{GK65}
Gene~H. Golub and William Kahan, \emph{Calculating the singular values and
  pseudo-inverse of a matrix}, Journal of the Society for Industrial and
  Applied Mathematics: Series B, Numerical Analysis (1965), 205--224.

\bibitem[GKL95]{GKL95}
Peter Gritzmann, Victor Klee, and David~G. Larman, \emph{Largest j-simplices
  n-polytopes}, Discrete {\&} Computational Geometry \textbf{13} (1995),
  477--515.

\bibitem[GLPW16]{ghashami2016frequent}
Mina Ghashami, Edo Liberty, Jeff~M Phillips, and David~P Woodruff,
  \emph{Frequent directions: Simple and deterministic matrix sketching}, SIAM
  Journal on Computing \textbf{45} (2016), no.~5, 1762--1792.

\bibitem[GR70]{GR70}
Gene~H. Golub and Christian Reinsch, \emph{Singular value decomposition and
  least squares solutions}, Numerische Mathematik (1970), 40--420.

\bibitem[Har04]{H04}
Sariel Har{-}Peled, \emph{No coreset, no cry}, \Proc 24th \FSTTCS, 2004,
  pp.~324--335.

\bibitem[Har06]{HP06}
\bysame, \emph{Coresets for discrete integration and clustering}, 26th FSTTCS,
  2006, pp.~33 -- 44.

\bibitem[Hel]{Hellerstein}
J.~Hellerstein, \emph{Parallel programming in the age of big data}, Gigaom
  Blog, 9th November, 2008.

\bibitem[HL11]{Hilbert}
M.~Hilbert and P.~Lopez, \emph{The world's technological capacity to store,
  communicate, and compute information}, Science \textbf{332} (2011), no.~6025,
  60--65.

\bibitem[HMT11]{HMT11}
Nathan Halko, Per{-}Gunnar Martinsson, and Joel~A. Tropp, \emph{Finding
  structure with randomness: Probabilistic algorithms for constructing
  approximate matrix decompositions}, {SIAM} Review \textbf{53} (2011), no.~2,
  217 -- 288.

\bibitem[HPK07]{HPK07}
Sariel Har-Peled and Akash Kushal, \emph{Smaller coresets for k-median and
  k-means clustering}, Discrete {\&} Computational Geometry \textbf{37} (2007),
  no.~1, 3--19.

\bibitem[HPM04]{HM04}
S.~Har-Peled and S.~Mazumdar, \emph{Coresets for $k$-means and $k$-median
  clustering and their applications}, \Proc 36th \STOC, 2004, pp.~291--300.

\bibitem[HS11]{HPS11}
Sariel Har{-}Peled and Micha Sharir, \emph{Relative (\emph{p},
  \emph{{\(\epsilon\)}})-approximations in geometry}, Discrete {\&}
  Computational Geometry \textbf{45} (2011), no.~3, 462--496.

\bibitem[IBM]{IBM}
\emph{Ibm: What is big data?. bringing big data to the enterprise}, Website,
  \url{ibm.com/software/data/bigdata/}, accessed on the 3rd of October 2012.

\bibitem[Ice]{IceCube}
Homepage of the IceCube neutrino observatory, \url{icecube.wisc.edu/}, accessed
  on the 21th of March 2013.

\bibitem[IMMM14]{indyk2014composable}
Piotr Indyk, Sepideh Mahabadi, Mohammad Mahdian, and Vahab~S Mirrokni,
  \emph{Composable core-sets for diversity and coverage maximization},
  Proceedings of the 33rd ACM SIGMOD-SIGACT-SIGART symposium on Principles of
  database systems, ACM, 2014, pp.~100--108.

\bibitem[JF]{jeff}
Phillips Jeff and Dan Feldman, personal communication.

\bibitem[KSS10]{KSS10}
Amit Kumar, Yogish Sabharwal, and Sandeep Sen, \emph{Linear-time approximation
  schemes for clustering problems in any dimensions}, J. {ACM} \textbf{57}
  (2010), no.~2, 5:1--5:32.

\bibitem[lhc]{lhcb}
Homepage of the large hadron collider beauty experiment,
  \url{lhcb-public.web.cern.ch/lhcb-public/}, accessed on the 21st of March
  2013.

\bibitem[Lib13]{liberty2013simple}
Edo Liberty, \emph{Simple and deterministic matrix sketching}, Proceedings of
  the 19th ACM SIGKDD international conference on Knowledge discovery and data
  mining, ACM, 2013, pp.~581--588.

\bibitem[LLS01]{LLS01}
Yi~Li, Philip.~M. Long, and Aravind Srinivasan, \emph{Improved bounds on the
  sample complexity of learning}, Journal of Computer and System Sciences
  (JCSS) \textbf{62} (2001), 516--527.

\bibitem[LS10]{LS10}
M.~Langberg and L.~J. Schulman, \emph{Universal epsilon-approximators for
  integrals}, \Proc 21st \SODA, 2010, pp.~598--607.

\bibitem[LSW17]{LSW17}
Euiwoong Lee, Melanie Schmidt, and John Wright, \emph{Improved and simplified
  inapproximability for k-means}, Information Processing Letters \textbf{120}
  (2017), 40--43.

\bibitem[Mah11]{Mahoney}
M.~W. Mahoney, \emph{Randomized algorithms for matrices and data}, Foundations
  and Trends{\textregistered} in Machine Learning \textbf{3} (2011), no.~2,
  123--224.

\bibitem[MNV09]{MNV09}
Meena Mahajan, Prajakta Nimbhorkar, and Kasturi~R. Varadarajan, \emph{The
  {P}lanar $k$-means {P}roblem is {NP}-{H}ard}, \Proc 3rd \WALCOM, 2009,
  pp.~274 -- 285.

\bibitem[MT82]{MT82}
Nimrod Megiddo and Arie Tamir, \emph{On the complexity of locating linear
  facilities in the plane}, Operation Research Letters \textbf{1} (1982),
  no.~5, 194--197.

\bibitem[Mut05]{Muthu05}
S.~Muthukrishnan, \emph{Data streams: Algorithms and applications}, Foundations
  and Trends in Theoretical Computer Science \textbf{1} (2005), no.~2, 117 --
  236.

\bibitem[MZ15]{mirrokni2015randomized}
Vahab Mirrokni and Morteza Zadimoghaddam, \emph{Randomized composable core-sets
  for distributed submodular maximization}, Proceedings of the Forty-Seventh
  Annual ACM on Symposium on Theory of Computing, ACM, 2015, pp.~153--162.

\bibitem[NDT09]{NDT09}
Nam~H. Nguyen, Thong~T. Do, and Trac~D. Tran, \emph{A fast and efficient
  algorithm for low-rank approximation of a matrix}, Proceedings of the 41st
  STOC, 2009, pp.~215 -- 224.

\bibitem[Pea01]{pearson1901}
Karl Pearson, \emph{On lines and planes of closest fit to systems of points in
  space}, The London, Edinburgh, and Dublin Philosophical Magazine and Journal
  of Science \textbf{2} (1901), no.~11, 559--572.

\bibitem[QSS00]{QSS10}
Alfio Quarteroni, Riccardo Sacco, and Fausto Saleri, \emph{Numerical
  mathematics}, pp.~22--25, Springer, 2000.

\bibitem[Sar06]{Sarlos06}
Tam{\'{a}}s Sarl{\'{o}}s, \emph{Improved approximation algorithms for large
  matrices via random projections}, 47th FOCS, 2006, pp.~143 -- 152.

\bibitem[Sch14]{S14}
Melanie Schmidt, \emph{Coresets and streaming algorithms for the k-means
  problem and related clustering objectives}, Ph.D. thesis, Universit{\"{a}}t
  Dortmund, 2014.

\bibitem[SH09]{Segaran}
T.~Segaran and J.~Hammerbacher, \emph{Beautiful data: The stories behind
  elegant data solutions}, O'Reilly Media, 2009.

\bibitem[Ste93]{S93}
Gilbert~W. Stewart, \emph{On the early history of the singular value
  decomposition}, SIAM {R}eview \textbf{35} (1993), 551 -- 566.

\bibitem[SV12]{SV12}
Nariankadu~D. Shyamalkumar and Kasturi~R. Varadarajan, \emph{Efficient subspace
  approximation algorithms}, Discrete {\&} Computational Geometry \textbf{47}
  (2012), no.~1, 44--63.

\bibitem[Vos91]{vose1991linear}
Michael~D. Vose, \emph{A linear algorithm for generating random numbers with a
  given distribution}, IEEE Transactions on software engineering \textbf{17}
  (1991), no.~9, 972--975.

\bibitem[VX12a]{VX12-soda}
K.~Varadarajan and X.~Xiao, \emph{A near-linear algorithm for projective
  clustering integer points}, \Proc \SODA, 2012.

\bibitem[VX12b]{varadarajan}
Kasturi Varadarajan and Xin Xiao, \emph{On the sensitivity of shape fitting
  problems}, \Proc 32nd Annual Conference on \FSTTCS, 2012, pp.~486 -- 497.

\bibitem[War68]{warren1968lower}
Hugh~E Warren, \emph{Lower bounds for approximation by nonlinear manifolds},
  Transactions of the American Mathematical Society \textbf{133} (1968), no.~1,
  167--178.

\bibitem[Whi12]{Hadoop}
T.~White, \emph{Hadoop: The definitive guide}, O'Reilly Media, 2012.

\end{thebibliography}
\end{document}